\documentclass[11pt, a4paper]{article}
\usepackage[margin=1.01in]{geometry}

\usepackage[utf8]{inputenc} 
\usepackage[english]{babel}
\usepackage[
citebordercolor={0.34 0.71 0.91},
linkbordercolor={0.84 0.37 0.0},
urlbordercolor={0.34 0.71 0.91}
]{hyperref}
\usepackage{amsmath,amssymb}
\usepackage{commath}
\usepackage[numbers,sort&compress]{natbib}
\usepackage{paralist}
\usepackage[textsize=footnotesize]{todonotes}
\usepackage{authblk}
\usepackage{xifthen}
\usepackage{xspace}
\usepackage [autostyle, english=american]{csquotes}
\MakeOuterQuote{"}
\usepackage{graphicx}
\usepackage{wrapfig}

\usepackage{microtype}
\usepackage{titlesec}
\titlespacing*{\paragraph}{0pt}{2.3ex}{1em}
\usepackage{doi}

\newlength{\punctuationfootlength}
\newcommand{\punctuationfootnote}[2]{#2\settowidth{\punctuationfootlength}%
{#2}\hspace{-0.5\punctuationfootlength}\footnote{#1}}

\DeclareMathOperator{\vol}{vol}
\DeclareMathOperator{\polylog}{polylog}

\usepackage{amsthm}
\usepackage{thmtools}
\usepackage{thm-restate}
\newtheorem{theorem}{Theorem}[section]
\newtheorem{corollary}[theorem]{Corollary}
\newtheorem{lemma}[theorem]{Lemma}

\theoremstyle{definition}

\newcommand{\EX}[1]{\mathbb E\left[ #1 \right]}
\newcommand{\eps}{\varepsilon}
\newcommand{\E}{\mathcal{E}}
\newcommand{\Pro}[1]{\mathrm{Pr}\left[ #1 \right]}
\newcommand{\eq}[1]{Equation~\eqref{eq:#1}}

\newcommand{\Oh}{O}

\newcommand{\pnt}[1]{\boldsymbol{#1}}
\newcommand{\p}{{\mathfrak{p}}}

\newcommand{\ie}{i.\,e.\xspace}
\newcommand{\whp}{w.\thinspace h.\thinspace p.\xspace}
\newcommand{\aas}{a.\thinspace a.\thinspace s.\xspace}

\newcommand{\dist}[3][]{%
  \lVert #2 - #3 \rVert\ifthenelse{\isempty{#1}}{}{_{#1}}}

\makeatletter
\g@addto@macro\bfseries{\boldmath}
\makeatother

\makeatletter
\def\footnoterule{\kern3\p@  \hrule width 3em \vspace{0.5ex}} 
\makeatother

\title{The Impact of Heterogeneity and Geometry\\on the Proof
  Complexity of Random Satisfiability\thanks{This work is partially
    funded by the project \emph{Scale-Free Satisfiability} (project
    no. 416061626) of the German Research Foundation (DFG).}}

\author[1]{Thomas Bläsius}
\author[2]{Tobias Friedrich}
\author[2]{Andreas Göbel}
\author[3]{\\Jordi Levy}
\author[2]{Ralf Rothenberger}
\affil[1]{\small  Karlsruhe Institute of Technology, Germany
  \texttt{<firstname.lastname@kit.edu>}}
\affil[2]{\small Hasso Plattner Institute, University of Potsdam, Germany
  \texttt{<firstname.lastname@hpi.de>}}
\affil[3]{\small IIIA, CSIC, Campus UAB, 08193 Bellaterra, Spain
  \texttt{<lastname@iiia.csic.es>}}

\date{}

\newif\iflong
\longtrue
\newcommand{\shortOrLong}[2]{%
  \iflong{#2}\else{#1}\fi%
}

\begin{document}

\maketitle

\begin{abstract}
  Satisfiability is considered the canonical NP-complete problem and
  is used as a starting point for hardness reductions in theory, while
  in practice heuristic SAT solving algorithms can solve large-scale
  industrial SAT instances very efficiently.  This disparity between
  theory and practice is believed to be a result of inherent
  properties of industrial SAT instances that make them tractable.
  Two characteristic properties seem to be prevalent in the majority
  of real-world SAT instances, heterogeneous degree distribution and
  locality.  To understand the impact of these two properties on SAT,
  we study the proof complexity of random $k$-SAT models that allow to
  control heterogeneity and locality.  Our findings show that
  heterogeneity alone does not make SAT easy as heterogeneous random
  $k$-SAT instances have superpolynomial resolution size.  This
  implies intractability of these instances for modern SAT-solvers.
  On the other hand, modeling locality with an underlying geometry
  leads to small unsatisfiable subformulas, which can be found within
  polynomial time.

  A key ingredient for the result on geometric random $k$-SAT can be
  found in the complexity of higher-order Voronoi diagrams.  As an
  additional technical contribution, we show an upper bound on the
  number of non-empty Voronoi regions, that holds for points with
  random positions in a very general setting.  In particular, it
  covers arbitrary $\p$-norms, higher dimensions, and weights
  affecting the area of influence of each point multiplicatively.  Our
  bound is linear in the total weight.  This is in stark contrast to
  quadratic lower bounds for the worst case.
\end{abstract}

\section{Introduction}
\label{sec:introduction}

Propositional satisfiability (SAT) is arguably among the most-studied
problems for both theoretical and practical research.  Nonetheless,
the gap between theory and practice is huge.  In theory, SAT is the
prototypical hard problem and hardness of other problems is shown via
reductions from SAT.  Achieving even a running time of $O(2^{cn})$ for
any $c < 1$ and $n$ variables would be a major breakthrough and a
somewhat surprising one at that.  On the contrary, reductions
to SAT are used to solve various problems appearing in
practice, as state-of-the-art SAT solvers can easily handle industrial
instances with millions of variables.

This theory--practice gap does not come from the lack of a sufficiently
precise theoretical analysis of modern SAT solvers.  They are actually
provably slow on most instances, i.e., drawing an instance uniformly
at random yields a hard instance with probability tending to~$1$ for
$n \to \infty$, if the clause-variable ratio is not too low or way too
high~\cite{shortProofs, manyExamples}.  Instead, the discrepancy comes
from the fact that industrial instances have properties that make them
easier than worst-case instances.  In 2014, \citet{v-bs-14} wrote that
\textquote{we have no understanding of why the specific sets of
  heuristics employed by modern SAT solvers are so effective in
  practice} and that we need this understanding to successfully
advance SAT solving further.

In recent years, scientists have been studying properties of
industrial SAT instances to gain this understanding.  By modeling SAT
instances as graphs, e.g., with edges indicating inclusion of
variables in clauses, one can benefit from the extensive research
conducted in the field of network science.  Two properties commonly
observed in real-world networks are heterogeneity and locality.
\emph{Heterogeneity} refers to the degree distribution, meaning that
vertices have strongly varying degrees.  In fact, one usually observes
a heavy-tailed distribution with many vertices of low degree and few
vertices of high degree.  A common assumption is a power-law
distribution~\cite{vhhk-snwd-18}, where the number of vertices of
degree $k$ is roughly proportional to $k^{-\beta}$.  The constant
$\beta$ is called the \emph{power-law exponent}.  \emph{Locality}
refers to the fact that edges tend to connect vertices that are close
in the sense that they remain well connected even when ignoring their
direct connection.  This can also be seen as having strong community
structures, with high connectivity within communities and loose ties
between communities.

With respect to these two properties, industrial SAT instances are
similar to real-world networks.  In many cases, the variable
frequencies are heterogeneous~\cite{abl-sisi-09} and there is a high
level of locality~\cite{agl-cssf-12}.  The latter is often measured in
terms of modularity.  Inspired by network science, researchers have
studied models that resemble industrial instances with respect to
these properties.  Particularly, Ansótegui et al.~\cite{abl-tilrsi-09}
introduced a power-law SAT model for heterogeneous instances, which
has been theoretically studied in terms of satisfiability
thresholds~\cite{fkrss-bstpldrs-17,fkrs-ptssf-17,fr-stnr2-19}.
A different model with heterogeneous degree distributions has been 
studied by \citet{cooper2SAT}, \citet{levy17}, and \citet{omelchenko2019}.
Moreover, \citet{gl-lrsi-17} introduced a model in which variable
weights lead to heterogeneity while an underlying geometry facilitates
locality. Comparing this to network models, the former
model~\cite{abl-tilrsi-09} is the SAT-variant of Chung-Lu
graphs~\cite{cl-ccrgg-02,cl-adrgged-02}.  The latter~\cite{gl-lrsi-17}
is based on the popularity-similarity model~\cite{pks-pvsgn-12}, which
is closely related to hyperbolic random graphs~\cite{kpk-hgcn-10} and
geometric inhomogeneous random graphs~\cite{bkl-sgirglt-17}.

Besides serving as somewhat realistic benchmarks for SAT
competitions~\cite{gl-dpssi-17}, these SAT models can be used to study
solver behavior depending on heterogeneity and locality.  One can
experimentally observe that a high level of heterogeneity improves the
performance of SAT solvers that also perform well on industrial
instances~\cite{abl-tilrsi-09,bfs-etcsf-19}.  Moreover, locality seems
very beneficial as solvers appear to implicitly use the locality of a
given instance~\cite{gl-lrsi-17}.  This coincides with the findings of
experiments on actual industrial instances that show that the locality
(measured using modularity) of an instance is a good predictor for
solver performance~\cite{ngf-icsssp-14,zmw-rcpssp-17,zmw-esmmssp-18}.

Up to date, there are no theoretical results supporting these
experimental observations.  On the contrary, it has been shown that
instances generated by the community attachment
model~\cite{gl-mbrsig-15}, which enforces a community structure, are
hard for modern SAT solvers~\cite{mds-hscs-16}.  With this paper, we
provide a theoretical foundation that matches the observations in
practice by studying the proof complexity of $k$-SAT instances (for
constant $k$) drawn from the power-law SAT model, and from a very
general model with underlying geometry.
The former was introduced by Ansótegui et al.~\cite{abl-tilrsi-09},
the latter is a generalization of the geometric model by
\citet{gl-lrsi-17} in the same way as geometric inhomogeneous random
graphs~\cite{bkl-sgirglt-17} are a generalization of hyperbolic random
graphs~\cite{kpk-hgcn-10}.
Our findings are that heterogeneous instances are hard asymptotically
almost surely\footnote{\emph{Asymptotically almost surely (\aas)}
  refers to a probability that tends to $1$ for $n\to \infty$.
  \emph{With high probability (\whp)} refers to the stronger
  requirement that the probability is in $1 - O(1/n)$. Additionally, we say that an event holds \emph{with overwhelming probability}, if for every $c > 0$ it holds with
probability at least $1 - O(n^{-c})$.} in that modern
SAT solvers require superpolynomial or even exponential running time
to refute unsatisfiable instances.  On the contrary, instances with a
high level of locality facilitated by an underlying geometry are
\aas{}\ easy to solve.  Our results focus on unsatisfiable instances,
i.e., on the case where a solver has to prove that no satisfying
assignment exists.  This is typically much harder than finding a
satisfying assignment, making the unsatisfiable regime arguably more
relevant.  Besides these results on SAT, we provide insights on the
complexity of weighted higher-order Voronoi diagrams in higher
dimensions, which is of independent interest.

The power-law and geometric models both mimic specific properties
observed in industrial instances while trying to make as little
additional assumptions as possible.  Though this makes the resulting
instances arguably more realistic than, e.g., instances drawn
uniformly at random, we want to stress that even the geometric model
is far from a perfect representation of industrial instances.  Thus,
our results do not claim to completely explain the efficiency of
modern SAT solvers on industrial instances.  However, to the best of
our knowledge, we provide the first theoretical result that links a
high level of locality to provably more tractable instances, which we
believe to be a first step towards closing the theory--practice gap.

\subsubsection*{Outline}

We state and discuss our main results and technical contributions in
Section~\ref{sec:results-techn-contr}.  Formal definitions are in
Section~\ref{sec:preliminaries}.  A short outline of our core
arguments is in Section~\ref{sec:core-argum-outl}\shortOrLong{.  The
  detailed proofs can be found in the full
  version~\cite{bfg-ihgpcrs-20}.}{, followed by the formal proofs:
  lower bounds for the power-law model in
  Section~\ref{sec:bipartite-expansion}, upper bounds on the
  complexity of Voronoi diagrams in
  Section~\ref{sec:compl-voron-diagr}, and upper bounds for the
  geometric SAT model in
  Section~\ref{sec:geometric-model-with-temperature}.  To not distract
  from the core arguments, results we use that were either known
  before or are straight-forward to prove are outsourced to
  Appendix~\ref{sec:basic-techn-tools}.}

\section{Results, Technical Contribution, Discussion}
\label{sec:results-techn-contr}

In this section, we state our results and discuss the contribution,
also in context to previous results.  To make the results
understandable, we briefly discuss, e.g., the probability
distributions over SAT formulas we study.  These are short and not
meant to be formal definitions.  For complete definitions, see
Section~\ref{sec:preliminaries}.

\subsection{Power-Law SAT}
\label{sec:results-power-law-sat}

The power-law SAT model has four parameters: the number of variables
$n$, the number of clauses $m$, the number $k$ of variables appearing
in each clause, and a power-law exponent~$\beta$.  To draw a formula,
power-law weights with exponent $\beta$ are assigned to the variables
and then each clause is generated independently by drawing $k$
variables without repetition using probabilities proportional to the
weights.  Each literal is negated with probability~$1/2$.

To discuss our first main contribution, let $\Phi$ be a formula drawn
from the power-law model with density at or above the satisfiability
threshold, i.e., $\Phi$ is unsatisfiable at least with constant probability. We show that,
although it is likely that $\Phi$ is unsatisfiable, it is
highly unlikely that modern SAT solvers can figure that out in
polynomial time. 
We prove this using resolution proof complexity.

Resolution is a refutation technique for propositional and first-order logic introduced by \cite{dp-cpqt-60}.
If an application of resolution steps leads to a contradiction, the formula is unsatisfiable.
The sequence of resolved clauses then serves as a proof for unsatisfiability, also called a refutation of the formula.
The resolution proof system exhibits a strong connection to modern Davis--Putnam--Logemann--Loveland (DPLL) and conflict-driven clause learning (CDCL) SAT solvers:
DPLL is polynomially equivalent to tree-like resolution~\cite{b-bsa-06} and CDCL with unlimited restarts is polynomially equivalent to resolution~\cite{pd-pcssre-11, bs-nssspe-14}.
Thus, the minimum number of steps necessary to derive a contradiction also yields a lower bound on the running time of solvers simulating the same process.
This number of steps is also called the \emph{resolution size} of a formula, \ie the minimum number of resolution steps necessary to arrive at a contradiction.
Equivalently, the width of a resolution proof is the size of the largest clause appearing in the proof and the \emph{resolution width} of a formula is the smallest width of any proof refuting that formula.
Interestingly, a lower bound $w$ on the resolution width of a formula also
implies a lower bound on its resolution size~\cite{shortProofs}: every resolution proof of a formula in $k$-CNF has size
$\exp(\Omega((w-k)^2/n))$ and every tree-like resolution proof has size
$2^{w-k}$.  

We will show a lower bound for the resolution width of unsatisfiable formulas drawn from the power-law model. 
Our results translate to lower bounds on the resolution size and thus to matching lower bounds on the running time of
conflict-driven clause learning (CDCL) solvers.  For
DPLL solvers, which use tree-like
resolution, the bounds are even stronger. 
We only consider the resolution width of unsatisfiable instances.
Thus, the probability bound we get is actually a conditional probability conditioned on instances being unsatisfiable.
Note that our bound does
not only hold above the satisfiability threshold, where a random formula $\Phi$ is
\aas{}\ unsatisfiable, but also \emph{at} the threshold, where it is
unsatisfiable with constant probability.

\begin{restatable*}{theorem}{mainone}
  \label{thm:mainone}
  Let $\Phi$ be an unsatisfiable random power-law $k$-SAT formula with
  $n$ variables, $m \in\Omega(n)$ clauses, $k \ge
  3$, and power-law exponent $\beta > \frac{2k-1}{k-1}$.  Let $\Delta = m /
  n$ be large enough so
  that~$\Phi$ is unsatisfiable at least with constant probability.  
	Let $\eps$, $\eps_1,\dots, \eps_3$ be constants with $\eps > 0$, $\eps_1 =
  \frac{k - \eps}{2} - 1 > 0$, $\eps_2 = (k - \eps)\cdot \frac{\beta -
    2}{\beta - 1} - 1 > 0$, and $0 < \eps_3 < (\frac{k}{2} - 1)\cdot
  \frac{\beta - 2}{\beta - 1} - 1$.  For the resolution width
  $w$ of $\Phi$, it holds \aas that:
  %
  \begin{enumerate}[(i)]
  \item\label{item:low-beta} If
    $\beta \in \left(\frac{2k - 1}{k - 1}, 3\right)$ and
    $\Delta \in o\left(n^{\eps_2}\right)$, then
    $w \in \Omega\left(n^{\eps_2 / \eps_1} \Delta^{-1/\eps_1}\right)$.
  \item\label{item:beta-3} If $\beta = 3$ and
    $\Delta \in o\left(n^{\eps_1} / \log^{1 + \eps_1} n\right)$, then
    $w \in \Omega\left(n \cdot \Delta^{-1/\eps_1} / \log^{1 +
        1/\eps_1} n\right)$.
  \item\label{item:high-beta-3} If $\beta > 3$ and
    $\Delta \in o\left(n^{\eps_1}\right)$, then
    $w\in \Omega\left(n\cdot \Delta^{-1/\eps_1}\right)$.
  \item\label{item:high-beta-other} If $\beta > \frac{2k - 2}{k - 2}$
    and $\Delta \in o\left(n^{\eps_3} / \log^{\eps_3} n\right)$, then
    $w \in \Omega\left(n \cdot \Delta^{-1/\eps_3}\right)$.
  \end{enumerate}
\end{restatable*}

\begin{wrapfigure}{r}{0.382\textwidth}
  \centering
  \scalebox{.87}{
    \input{power-law-plot.tex}}
  \vspace{-0.2cm}
  \caption{Exponent of the bound (\ref{item:low-beta}) in
    \thmref{thm:mainone}.  Dashed vertical lines show where the bound
    (\ref{item:high-beta-other}) takes over.}
  \label{fig:power-law-plot}
\end{wrapfigure}
%
The above lower bounds allow the density $\Delta$ to be super-constant
(even polynomial), which is asymptotically above the satisfiability
threshold.  For the sake of simplicity, assume $\Delta$ to be constant
in the following.  Starting at the bottom~(\ref{item:high-beta-3},
\ref{item:high-beta-other}), we get a linear bound for $w$ if $\beta$
is sufficiently large, i.e., greater than $3$ or $(2k - 2) / (k - 2)$.
For $\beta = 3$ (\ref{item:beta-3}), the bound is still almost linear.
Note that these results in particular imply exponential lower bounds
on the resolution size and thus on the running time of CDCL and DPLL.
For smaller $\beta$~(\ref{item:low-beta}), we get a polynomial bound
for the width with exponent $\eps_2/\eps_1$; see
Figure~\ref{fig:power-law-plot} for a plot with $\eps$ close to $0$.

Interestingly enough, our bounds only hold for power law exponents
$\beta>\frac{2k-1}{k-1}$.  This is complemented by a previous
result~\cite{fkrss-bstpldrs-17}, which shows that the satisfiability threshold of power-law random $k$-SAT is at density $\Delta=\Theta(1)$ for power law
exponents $\beta>\frac{2k-1}{k-1}$ and that asymptotically
almost surely instances with constant
constraint densities are trivially unsatisfiable for power law
exponents $\beta<\frac{2k-1}{k-1}$.  Thus, the resolution width is constant in the latter case.

Part~\ref{item:high-beta-other} of Theorem~\ref{thm:mainone} is
derived via lower bounds on the bipartite expansion of the
clause-variable incidence graph of these instances.  These results can
be of independent interest for hypergraphs with edge size $k$ and for
random $(0,1)$-matrices. Additionally, these expansion properties
yield lower bounds for the clause space complexity, which in turn
gives lower bounds on the tree-like resolution size of such
formulas\shortOrLong{.}{ (Section~\ref{sec:BipExpansion}).}  More
precisely, this results in an exponential lower bound on the tree-like
resolution size for $\beta>\frac{2k-3}{k-2}$. This is an improvement of the
bound obtained via resolution width.

It is interesting to note that this result on the non-geometric
model supports the claim that locality is a crucial factor for easy SAT 
instances. The lower bounds for the power-law model are solely based 
on the fact that every set of clauses covers a comparatively large 
set of variables.  In other words, we only use that there are no 
clusters of clauses with similar variables, i.e., we explicitly 
use the lack of locality.

\subsection{Geometric SAT}
\label{sec:results-geometric-sat}

The geometric model has the following parameters: $n$, $m$, and $k$
have the same meaning as for the power-law model.  Moreover, $w$ is a
weight function assigning each variable $v$ a weight $w_v$ and $T$ is
the so-called temperature that controls the strength of locality by
varying the impact of the geometry.  As underlying geometric space, we
use the $d$-dimensional torus
$\mathbb T^d = \mathbb R^d / \mathbb Z^d$ (see
Section~\ref{sec:preliminaries}) equipped with a $\p$-norm with
$\p \in \mathbb N^+ \cup \infty$.  
To draw a formula, the variables
and clauses are assigned random positions in $\mathbb T^d$.  Then, for
each clause, $k$ variables are drawn without repetition with
probabilities depending on the variable weight and on the geometric
distance between clause and variable.  In the extreme case of $T = 0$,
each clause deterministically includes the $k$ closest variables
(where closeness is a combination of geometric distance and weight),
while increasing the temperature $T$ increases the probability for the
inclusion of more distant variables.  For $T \to \infty$, the model
converges to uniform random SAT.  Note that the weights are a
parameter of the model and not drawn randomly.  We have the following
theorem, where $W$ denotes the sum of all variable weights.  The
condition on the weights is in particular satisfied by power-law
distributed weights.

\begin{restatable*}{theorem}{GeometricSat}
  \label{thm:geometric-sat}
  Let $\Phi$ be a formula with $n$ variables and $m \in \Theta(n)$
  clauses drawn from the weighted geometric model with ground space
  $\mathbb T^d$ equipped with a $\p$-norm, temperature $T < 1$,
  $W \in O(n)$, and $w_v \in O(n^{1-\eps})$ for every $v \in V$ and
  any constant $\eps > 0$.  Then, $\Phi$ contains \aas{}\ an
  unsatisfiable subformula of constant size, which can be found in
  $O(n \log n)$ time.
\end{restatable*}

To briefly explain how we prove this, consider a simplified version
where variables and clauses are points in the Euclidean plane and each
clause contains the $k$ variables geometrically closest to it
(temperature~$T = 0$).  Now consider the equivalence relation obtained
by defining two points of the plane equivalent if and only if they
have the same set of $k$ closest variables.  The equivalence classes
of this relation are the regions of the order-$k$ Voronoi diagram of
the variable positions.  With this connection, we can use upper bounds
on the complexity of order-$k$ Voronoi diagrams~\cite{l-nnvdp-82} to
prove the existence of small and easy to find unsatisfiable
subformulas.  We note that this result is of asymptotic nature.  In
particular for small densities, the number of variables $n$ has to be
very large before the instances actually get as easy as stated in
Theorem~\ref{thm:geometric-sat}.  Nevertheless, this results strongly
suggests that an underlying geometry makes SAT instances more
tractable.

To extend the above argument to the general statement in
Theorem~\ref{thm:geometric-sat}, we extend the complexity bounds for
order-$k$ Voronoi diagrams in various ways; see next section for more
details.  Moreover, for non-zero temperatures, clauses no longer
include exactly the $k$ closest variables but can, in principle,
consist of any set of $k$ variables.  However, we can show that, with
high probability, a linear fraction of clauses behaves as in the
$T = 0$ case.  We note that analyses of similar structures, such as
hyperbolic random graphs, are often restricted to the simpler but less
realistic $T = 0$ case, e.g.,
\cite{bffk-svcpt-20,bff-espsf-18,bfm-gcrhg-13,ms-dkrg-19}.  We believe
that our analysis provides insights on the non-zero temperature case
that can be helpful for such related questions.

We note that our results seem to contradict the
results of \citet{mds-hscs-16}, stating that
\begin{inparaenum}[(i)]
\item a strong community structure is not sufficient to have tractable
  SAT instances and that
\item the community attachment model~\cite{gl-mbrsig-15}, which
  enforces a community structure, generates hard instances.
\end{inparaenum}
However, at a closer look, this is not a contradiction at all.  Though
measuring the community structure, e.g., via modularity, is a good
indicator for locality, the concept of locality goes deeper.  If the
instance can be partitioned such that there are strong ties within
each partition and loose ties between partitions, then the instance
has a strong community structure.  However, to have a high level of
locality, this concept has to hierarchically repeat on different
levels of magnitude, i.e., there needs to be community structure
within each partition and between the partitions.  To state this
slightly differently, consider locality based on a notion of
similarity between objects (here: variables or clauses).  In this
paper, we use distances between random points in a geometric space as
a measure for similarity, which gives us a continuous range of more or
less similar objects.  In contrast to that, in the above mentioned
papers focusing on a flat community
structure~\cite{gl-mbrsig-15,mds-hscs-16}, similarity is a binary
equivalence relation: two objects are either similar or they are not.

\subsection{Voronoi Diagrams}
\label{sec:results-voronoi-diagrams}

Consider a finite set of points, called sites, in a geometric space.
The most commonly studied type of Voronoi diagram assumes the
2-dimensional Euclidean plane as ground space and has one Voronoi
region for each site, containing all points closer to this site than
to any other site.  We deviate from this default setting in four ways:
\begin{inparaenum}[(i)]
\item We allow an arbitrary constant dimension $d$, where the ground
  space is the torus or a hypercube in $\mathbb R^d$.
\item We consider the order-$k$ Voronoi diagram, which has for every
  subset $A$ of sites with $|A| = k$ a (possibly empty) Voronoi region
  containing all points for which $A$ are the $k$ nearest sites.  The
  number of non-empty order-$k$ Voronoi regions is called the
  \emph{complexity} of the diagram.
\item The sites have multiplicative weights that scale the influence
  of the different sites.  Without loss of generality, we assume the
  weights to be scaled such that the minimum is~$1$.
\item We allow the $\p$-norm for arbitrary
  $\p \in \mathbb N^+ \cup \infty$.
\end{inparaenum}

\begin{restatable*}{theorem}{ComplexityRandomVoronoiDiagram}
  \label{thm:complexity-random-voronoi-diagram}
  Let $S$ be a set of $n$ sites with minimum weight $1$, total weight
  $W$, and random positions on the $d$-dimensional torus equipped with
  a $\p$-norm, for constant~$d$.  For every fixed $k$, the expected
  number of regions of the weighted order-$k$ Voronoi diagram of $S$
  is in $O(W)$.  The same holds for random sites in a hypercube.
\end{restatable*}

To set this result into context, we briefly discuss previous work on
the complexity of Voronoi diagrams in different settings.  See the
book by \citet{akl-vddt-13} for a general overview on Voronoi
diagrams.  To this end, we use the following theorem that relates the
complexity in terms of Voronoi regions (which is what we are concerned
with in this paper) with the complexity in terms of
vertices\punctuationfootnote{Although the Voronoi regions are not
  necessarily polytopes in the weighted setting, we adopt the notion
  for polytopes and call the corners of Voronoi regions
  \emph{vertices}.  I.e., vertices are the $0$-dimensional elements
  (a.k.a.\ points) of the boundary, where higher-dimensional elements
  (a.k.a.\ edges, faces, etc.) intersect.  They are represented as
  small black dots in Figure~\ref{fig:weighted-voronoi-complexity}.}.

\begin{restatable*}{theorem}{VoronoiComplVerticesRegions}
  \label{thm:complexity-vertices-vs-regions}
  Let $S$ be a set of $n$ weighted sites in general position in
  $\mathbb R^d$ equipped with a $\p$-norm.  If the order-$k$ Voronoi
  diagram has $\ell$ vertices, then the order-$(k + d)$ Voronoi
  diagram has $\Omega(\ell)$ non-empty regions.
\end{restatable*}

We note that, using insights from previous work, this theorem is not
hard to prove.  One basically has to generalize the result by
\citet{l-vdlrd-96} bounding the number of $d$-spheres going through
$d + 1$ points in $d$-dimensional space to weighted sites, and then
observe how the Voronoi diagram changes in the construction by
\citet{l-nnvdp-82} for $d = 2$, when going from order-$k$ to
order-$(k + 1)$.  However, we are not aware of previous work stating
this connection between vertices and non-empty regions in higher
orders explicitly.

The four above-mentioned generalizations of the basic Voronoi diagram
(higher dimension, higher order, multiplicative weights, and different
$\p$-norms) have all been considered before.  However, to the best of
our knowledge, not all of them together.  

\begin{figure}
  \centering
  \includegraphics{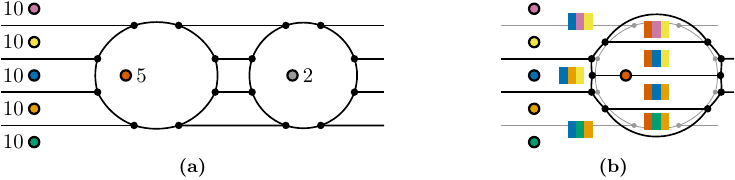}
  \caption{\textbf{(a)}~Weighted Voronoi diagram (order-$1$) of the
    colored sites.  Continuing the construction with $n/2$ high-weight
    sites on the left and $n/2$ low-weight sites towards the right
    yields $\Theta(n^2)$ vertices (small black dots).  Note that each
    vertex lies on the boundary of three regions and has thus equal
    weighted distance to its three closest sites.  \textbf{(b)}~The
    order-$3$ Voronoi diagram for the same sites (excluding one).  The
    colored boxes indicate the three closet sites.  The
    \mbox{order-$1$} diagram is shown in the background.  Each
    order-$1$ vertex lies in the interior of an order-$3$ region as it
    has equal weighted distance to its three closest sites.  As at
    most two order-$1$ vertices share an order-$3$ region, we get
    $\Omega(n^2)$ order-$3$ regions.
    Theorem~\ref{thm:complexity-vertices-vs-regions} generalizes this
    observation.}
  \label{fig:weighted-voronoi-complexity}
\end{figure}

Higher-order Voronoi diagrams have been introduced by Shamos and
Hoey~\cite{sh-cpp-75}.  \citet{l-nnvdp-82} showed that the order-$k$
Voronoi diagram in the plane (unweighted with Euclidean metric) has
complexity $O(k\,(n - k))$ (in terms of number of regions), which is
linear for constant $k$.  For the $1$- and $\infty$-norm,
\citet{lpl-osall-11} improved this bound to
$O(\min\{k\,(n - k), (n - k)^2\})$.  Closely related to the $1$-norm,
\citet{gllw-hocvd-12} showed similar complexity bounds for
higher-order Voronoi diagrams on transportation networks of
axis-parallel line segments.  \citet{bck-choavd-15} show an upper
bound of $2k\,(n - k)$ for the much more general setting of abstract
Voronoi diagrams.  There, the metric is replaced by curves separating
pairs of sites such that certain natural (but rather technical)
conditions are satisfied.  One obtains normal Voronoi diagrams when
using perpendicular bisectors for these curves.  This in particular
shows that the $2k\,(n - k)$ bound on the number of regions in the
order-$k$ Voronoi diagram holds for arbitrary $\p$-norms in
$2$-dimensional space and for the hyperbolic plane.  As the hyperbolic
plane is closely related to $1$-dimensional space with sites having
multiplicative power-law weights~\cite{bkl-sgirglt-17}, we suspect
that the bound by \citet{bck-choavd-15} also covers this case.

In general one can say that higher-order Voronoi diagrams of
unweighted sites in $2$-dimensional space are well-behaved in that
they have linear complexity.  This still holds true for arbitrary
$\p$-norms.  However, this picture changes for weighted sites or
higher dimensions.

Voronoi diagrams with multiplicative weights were first considered by
\citet{b-wtp-80}\footnote{In this paper, Voronoi regions are called
  \emph{Thiessen polygons}.} due to applications in economics.  Beyond
that, multiplicatively weighted Voronoi diagrams have applications in
sensor networks~\cite{cgvc-rcdsn-06}, logistics~\cite{gncs-mwvdald-06}
and the growth of crystals~\cite{dd-cvdd-06}.  However, even in the
most basic setting of 2-dimensional Euclidean space and order $1$,
weighted Voronoi diagrams can have quadratic
complexity~\cite{ae-oacwvdp-84} (in terms of number of vertices).
This comes from the fact that Voronoi cells are not necessarily
connected; see Figure~\ref{fig:weighted-voronoi-complexity}a for the
construction of Aurenhammer and Edelsbrunner~\cite{ae-oacwvdp-84} that
proves the lower bound.  With
Theorem~\ref{thm:complexity-vertices-vs-regions}, and as illustrated
in Figure~\ref{fig:weighted-voronoi-complexity}, this implies that
even the order-$3$ Voronoi diagram of weighted sites in 2-dimensional
Euclidean space has a quadratic number of non-empty regions.  As a
special case, Theorem~\ref{thm:complexity-random-voronoi-diagram}
shows that this complexity is only linear in the total weight for
sites positioned randomly in the unit square.  Moreover, this also
implies that the number of vertices of the corresponding order-$1$
Voronoi diagram is linear.  This nicely complements the result by
\citet{hr-crwmvd-15}, who show that the expected complexity of
order-$1$ Voronoi diagrams of sites in $2$-dimensional Euclidean space
with random weights is $O(n \polylog n)$.  Only recently,
\citet{fr-lecdmvd-20} showed that sites with weights chosen randomly
form a constant-sized set of possible weights yield Voronoi diagrams
with linear complexity.  Moreover, more closely related, they show
that the Voronoi diagram of sites with arbitrary weights and with
random positions chosen in the unit square has linear complexity in
expectation.  We are not aware of any results concerning the
complexity of Voronoi diagrams when combining multiplicative weights
with higher dimension, higher order or other norms.

For higher dimensions, even normal (first order, unweighted) Voronoi
diagrams in $3$-dimensional Euclidean space can have
$\Theta(n^2)$~\cite{k-cdvd-80,s-nfhdvd-87} vertices.
Theorem~\ref{thm:complexity-vertices-vs-regions} thus implies that the
order-$4$ Voronoi diagram has a quadratic number of non-empty regions.
Moreover, the complexity of higher-order Voronoi diagrams in higher
dimensions has been considered before by \citet{m-lavd-91}, who obtains
polynomial bounds with the degree of the polynomial depending on the
dimension.  Our Theorem~\ref{thm:complexity-random-voronoi-diagram} in
particular shows that this complexity is much lower, namely linear,
for the hypercube with randomly positioned sites.  Moreover, via
Theorem~\ref{thm:complexity-vertices-vs-regions} this gives a linear
bound on number of vertices in the normal order-$1$ Voronoi diagram in
higher dimensions.  We note that this special case of our result
coincides with a previous result by \citet{bdhs-accvd-05}.  Similarly,
\citet{d-hdvdlet-91} showed that sites drawn uniformly from a higher
dimensional unit sphere (instead of a hypercube) yield Voronoi
diagrams of linear complexity in expectation.  Moreover, due to
\citet{gn-acvdr-03} and \citet{dhr-ecvdt-16}, the same is true for
random sites on $3$-dimensional polytopes and random sites on
polyhedral terrains, respectively.  Thus, though higher dimensional
Voronoi diagrams can be rather complex in the worst case, these
results indicate that one can expect most instances to be rather well
behaved.  An alternative explanation of why the complexity of
practical instances is lower than the worst-case indicates is given by
\citet{e-npsch-01,e-dpshsdt-02}, who studies the complexity of
$3$-dimensional Voronoi diagrams depending on the so-called spread of
the sites.

The above results for higher dimensional Voronoi diagrams consider the
Euclidean norm.  For general $\p$-norms, \citet{l-vdlrd-96} showed
that the complexity of the Voronoi diagram is bounded by $O(n^c)$,
where $c$ is a constant independent of $\p$ but dependent on the
dimension $d$.  With the same argument as above,
Theorem~\ref{thm:complexity-random-voronoi-diagram} together with
Theorem~\ref{thm:complexity-vertices-vs-regions} implies a linear
bound for this complexity that holds in expectation.  Moreover,
\citet{bsty-vdhdc-98} show more precise bounds of
$\Theta(n^{\lceil d/2 \rceil})$ and $\Theta(n^2)$ for the $\infty$-
and the $1$-norm, respectively.  Again, our result implies linear
bounds for random sites in this setting.

\section{Formal Definitions}
\label{sec:preliminaries}

Here we provide formal definitions for all concepts we use throughout
the paper, including the power-law and geometric random SAT models, Resolution,
and Voronoi diagrams.

\subsubsection*{$k$-SAT}

We let $x_1,x_2,\ldots,x_n$ denote Boolean variables that can be either true or false.
A clause is a disjunction of literals $\ell_1 \lor \ldots \lor \ell_k$, where each literal assumes a (possibly negated) variable. 
For a literal $\ell_i$ let $|\ell_i|$ denote the variable of the literal.
A formula $\Phi$ in conjunctive normal form (CNF) is a conjunction of clauses $c_1 \land \ldots \land c_m$ and a formula in $k$-CNF is a conjunction of clauses, where each clause contains exactly three distinct literals.
We conveniently interpret a Boolean formula in CNF as a set of clauses and a clause $c$ both as a Boolean formula and as a set of literals.
We say that $\Phi$ is satisfiable if there exists an assignment of variables $x_1, \ldots, x_n$ such that the formula evaluates to true.

\subsubsection*{Power-Law Random $k$-SAT}

The power-law model can be defined via the more general
\emph{non-uniform model}.  To draw a $k$-SAT formula from the
non-uniform model, let $n$ and $m$ be the number of variables and
clauses, respectively, and let $w_1, \dots, w_n$ be variable weights.
We sample $m$ clauses independently at random.  Each clause is sampled
by drawing $k$ variables without repetition with probabilities
proportional to their weights.  Then each of the $k$ variables is
negated independently at random with probability $1/2$.

The \emph{power-law model} for a power-law exponent $\beta > 2$ is an
instantiation of the non-uniform model with discrete power-law weights
\begin{equation*}
  w_i = i^{-\frac{1}{\beta - 1}}.
\end{equation*}

\subsubsection*{Resolution}

The resolution proof system uses two rules, the resolution rule and the weakening rule.
Given two clauses $a\vee x$ and $b\vee \overline{x}$, where $a$ and $b$ are clauses and $x$ is a Boolean variable, the resolution rule states
\[\frac{a\vee x\hspace{5ex}b\vee \overline{x}}{a\vee b},\]
\ie the clause $a\vee b$ is a logical consequence of the two given clauses.
The weakening rule states that for any two clauses $a$ and $b$ it holds that
\[\frac{a}{a\vee b},\]
\ie if $a$ holds, then $a\vee b$ holds as well.
For a formula $\Phi=\left\{c_1,c_2,\ldots,c_m\right\}$ in CNF a resolution \emph{derivation} of a clause $c$ from $\Phi$ is a sequence of clauses $\left(d_1,d_2,\ldots,c\right)$ such that each clause $d_i$ is either one of the initial clauses $c_1,\ldots,c_m$ or derived from previous clauses with either the resolution rule or the weakening rule.
A resolution \emph{refutation} is a resolution derivation of the empty clause.
The \emph{size} of a derivation is the number of clauses it contains.
The size of a formula in CNF is the size of a smallest refutation for it. 
The \emph{width} of a derivation is the size of the largest clause in it.
The width of a formula in CNF is the smallest width of any refutation for it.

\subsubsection*{Graph Representation and Expansion}

Let $\Phi$ be a SAT-formula with variable set $V$ and clause set $C$.
The \emph{clause-variable incidence graph} $G(\Phi)$ of $\Phi$ has
vertex set $C \cup V$, with an edge between a clause and a variable if
and only if the clause contains the variable.  Clearly, $G(\Phi)$ is
bipartite. It is an \emph{$(r, c)$-bipartite expander} if for all
$C' \subset C$ with $|C'| \le r$ it holds that
$|N(C')| \ge (1 + c)\cdot |C'|$, where $N(C')$ is the neighborhood of $C'$.

\subsubsection*{Geometric Ground Space}

We regularly deal with points with random positions in some geometric
space.  With \emph{random point}, we refer to the uniform distribution
in the sense that the probability for a point to lie in a region $A$
is proportional to its volume $\vol(A)$.  For this to work, the volume
of the ground space has to be bounded.  Canonical options are, e.g., a
unit-hypercube or a unit-ball.  These, however, lead to the necessity
of special treatment for points close to the boundary, which makes the
analysis more tedious without giving additional insights.  To
circumvent this, we use a torus as ground space, which is completely
symmetric.

The \emph{$d$-dimensional torus} $\mathbb T^d$ is defined as the
$d$-dimensional hypercube $[0, 1]^d$ in which opposite borders are
identified, i.e., a coordinate of $0$ is identical to a coordinate of
$1$\punctuationfootnote{For convenience reasons, we sometimes work
  with $[-0.5, 0.5]^d$ instead of $[0, 1]^d$.}.
It is equipped with the $\p$-norm as metric, for arbitrary but fixed
$\p \in \mathbb N^+ \cup \{\infty\}$.  To define it formally for the
torus, let $\pnt{p} = (p_1, \dots, p_d)$ and
$\pnt{q} = (q_1, \dots, q_d)$ be two points in $\mathbb T^d$.  The
circular difference between the $i$th coordinates is
$|p_i - q_i|_{\circ} = \min\{|p_i - q_i|, 1 - |p_i - q_i|\}$.  With
this, the \emph{distance} between $\pnt{p}$ and $\pnt{q}$ is
\begin{align*}
  \dist{\pnt{p}}{\pnt{q}} = 
  \begin{cases}
    \sqrt[\p]{\sum_{i \in [d]} \left|p_i - q_i\right|_{\circ}^\p}&\text{
      for } \p \not= \infty,\\
    \max_{i \in [d]} \{\left|p_i - q_i\right|_{\circ}\}&\text{
      for } \p = \infty.
  \end{cases}
\end{align*}

\subsubsection*{Random Points}

We obtain the uniform distribution for a point
$\pnt{p} = (p_1, \dots, p_d)$ by drawing each coordinate $p_i$
uniformly at random from $[0, 1]$.  For two random points $\pnt{p}$
and $\pnt{q}$, their distance $\dist{\pnt{p}}{\pnt{q}}$ is a random
variable.  Let $F_{\mathrm{dist}}(x)$ be its \emph{cumulative
  distribution function (CDF)}, i.e.,
$F_{\mathrm{dist}}(x) = \Pro{\dist{\pnt{p}}{\pnt{q}} \le x}$.  To
determine $F_{\mathrm{dist}}(x)$, fix the position of $\pnt{p}$.
Then, for $x \le 0.5$, the set of points of distance at most $x$ to
$\pnt{p}$ is simply the ball $B_{\pnt{p}}(x)$ of radius $x$ around
$\pnt{p}$, yielding
\begin{align}
  \label{eq:cdf-distance} 
  F_{\mathrm{dist}}(x)
  &= \vol(B_{\pnt{p}}(x))\\
  &= \Pi_{d, \p}\cdot x^d \quad \text{for } 0 \le x \le 0.5, \notag\\
  \text{with } \Pi_{d, \p}
  &= \frac{\left(2\Gamma\left(1/\p +
    1\right)\right)^d}{\Gamma\left({d}/{\p} + 1\right)},\notag
%
\end{align}
where $\Gamma$ is the gamma function.  Note that $\Pi_{d, \p}$ only
depends on $d$ and $\p$ but is constant in $x$.  Moreover
$\Pi_{2, 2} = \pi$ (thus the name $\Pi$), and
$\Pi_{d, \infty} = \lim_{\p \to \infty} \Pi_{d, \p} = 2^d$.  For
distances $x > 0.5$, the formula for $F_{\mathrm{dist}}(x)$ is more
complicated (we basically have to subtract the parts reaching out of
the hypercube).  However, for our purposes, it suffices to know
$F_{\mathrm{dist}}(x)$ for $x \le 0.5$ and use the obvious bound
$F_{\mathrm{dist}}(x) \le 1$ for $x > 0.5$.

\subsubsection*{Weighted Points and Distances}

We regularly consider a fixed set of $n$ points equipped with weights,
which we call \emph{sites}.  For a site $\pnt{s}_i$ with weight $w_i$,
the \emph{weighted distance} of a point $\pnt{p}$ to $\pnt{s}_i$ is
$\dist{\pnt{s}_i}{\pnt{p}}/w_i^{1/d}$.  For a fixed value $x$, the set
of points with weighted distance at most $x$ are the points with
$\dist{\pnt{s}_i}{\pnt{p}} \le x w_i^{1/d}$.  Note that the volume of
this set is proportional to $w_i$.  Intuitively, the region of
influence of a site is thus proportional to its weight.  To simplify
notation in some places, we define \emph{normalized weights}
$\omega_i = w_i^{1/d}$\punctuationfootnote{We note, in the context of
  weighted Voronoi diagrams, it is common to only use the normalized
  weights (just calling them ``weights'').  In the context of random
  networks, however, the non-normalized weights are more common.  As
  both notions have their advantages in different situations, we use
  both.}.

\subsubsection*{Geometric Random k-SAT}

In the \emph{geometric model}, we sample positions for the variables
and clauses uniformly at random in the $d$-dimensional torus
$\mathbb T^d$.  For $v \in V$ and $c \in C$, we use $\pnt{v}$ and
$\pnt{c}$ to denote their positions, respectively.  Let
$w_1, \dots, w_n$ be \emph{variable weights} that are normalized such
that the smallest weight is $1$.  Moreover, let
$W = \sum_{v = 1}^n w_v$.  For a clause $c$ and a variable $v$, define
the \emph{connection weight}
\begin{equation*}
  X(c, v) = \left(\frac{w_v}{\dist{\pnt{c}}{\pnt{v}}^d}\right)^{1/T}.
\end{equation*}
This is the reciprocal of the weighted distance between $\pnt{v}$ and
$\pnt{c}$ raised to the power $d/T$.  The $k$ variables for the clause
$c$ are drawn without repetition with probabilities proportional to
$X(c, v)$.  Among all possible combinations, we choose which of the
$k$ variables to negate uniformly at random, without repetition if
possible, i.e., we only get the same clause twice if we have more than
$2^k$ clauses with the same variable set.  For $T \to 0$ the model
converges to the threshold case where $c$ contains the $k$ variables
with smallest weighted distance.

The connection weight $X({c}, {v})$ is a random variable.  We denote
the CDF of $X({c}, {v})$ with $F_X(x)$.  With the CDF for the distance
between two random points in Equation~\eqref{eq:cdf-distance}, we
obtain the following\shortOrLong{~\cite{bfg-ihgpcrs-20}:}{; see
  Lemma~\ref{lem:cdf-connection-weights} for a proof:}
\begin{equation}
  \label{eq:cdf-probability-weight}
  F_X(x) = 1 - \Pi_{d, \p} w_v x^{-T} \quad \text{for } x \ge
  \left(2^d w_v\right)^{1/T}.
\end{equation}

\subsubsection*{Voronoi Diagrams}

Let $S = \{\pnt{s}_1, \dots, \pnt{s}_n\}$ be a set of sites with
weights $w_1, \dots, w_n$.  A point $\pnt{p}$ belongs to the
\emph{(open) Voronoi region} of a site $\pnt{s}_i$ if its weighted
distance to $\pnt{s}_i$ is smaller than its weighted distance to any
other site.  The collection of all Voronoi regions is the
\emph{Voronoi diagram} of $S$.  \emph{Order-$k$ Voronoi regions} are
defined analogously for subsets $A \subseteq S$ with $|A| = k$, i.e.,
the region of $A$ contains a point $\pnt{p}$ if and only if the
weighted distances of $\pnt{p}$ to all sites in $A$ is smaller than
the weighted distance to any site not in $A$.  More formally,
$\pnt{p}$ belongs to the order-$k$ Voronoi region of $A$ if there
exists a radius $r$ such that
$\dist{\pnt{s}_i}{\pnt{p}} \le \omega_i r$ for $\pnt{s}_i \in A$ and
$\dist{\pnt{s}_i}{\pnt{p}} > \omega_i r$ for $\pnt{s}_i \notin A$.
Note that the order-$k$ Voronoi region of $A$ is potentially empty.
The \emph{order-$k$ Voronoi diagram} is the collection of all
non-empty order-$k$ Voronoi regions.  Its \emph{complexity} is the
number of such non-empty regions.

\section{Core Arguments}
\label{sec:core-argum-outl}

\shortOrLong{%
  Here we only briefly discuss the core arguments.  See the full version
  for detailed proofs~\cite{bfg-ihgpcrs-20}.%
}{%
  Before delving into the technical details of our proofs in the
  subsequent sections, we briefly discuss the core arguments.%
}

\subsection{Power-Law SAT}\label{arg:pl}

We use a framework that \citet{shortProofs} introduced for the uniform
SAT model.  We prove lower bounds for the resolution width, which
imply lower bounds for the resolution size and the tree-like
resolution size, which then imply lower bounds for the running times
of CDCL and DPLL solvers, respectively.

To bound the resolution width, we essentially have to show that
different clauses do not overlap too heavily.  Specifically, a formula
has resolution width $\Omega(w)$ if
\begin{inparaenum}[(1)]
\item every set $S$ of at most $w$ clauses contains at least $|S|$
  different variables and
\item every set $S$ of $\frac{1}{3}w\le|S|\le\frac{2}{3}w$ clauses contains at 
least a constant fraction of unique variables.
\end{inparaenum}

We achieve the bounds in
Theorem~\ref{thm:mainone}~(\ref{item:low-beta}--\ref{item:high-beta-3})
by showing the above two properties directly.
For the bound in
Theorem~\ref{thm:mainone}~(\ref{item:high-beta-other}), we first
observe that both properties are fulfilled if the clause-variable
incidence graph of a $k$-CNF formula $\Phi$ has high enough bipartite
expansion.  Recall the definition of bipartite expansion from
Section~\ref{sec:preliminaries} and note how the requirement that the
neighborhood of clause vertices is large resembles the requirement
that clauses do not overlap too heavily.  We show that $G(\Phi)$ is a
bipartite expander asymptotically almost surely if $\Phi$ is drawn
from the power-law model, which yields the lower bound of
Theorem~\ref{thm:mainone}~(\ref{item:high-beta-other}).

Compared to the uniform case, the weights make the properties required
for the lower bounds less likely.  Variables with high weight appear
in many clauses, making the clauses less diverse.  Thus, it is less
likely that every clause set covers a large variety of variables.

\subsection{Geometric SAT}
\label{sec:core-arguments-geometric-sat}

To explain the core idea of our proof, consider the following
simplified geometric model.  Map $n$ variables and $m$ clauses to
distinct points in the 2-dimensional Euclidean plane (randomly or
deterministically).  Build the SAT instance by including in each
clause $c$ the $k$ variables with the smallest geometric distance to
$c$.
Now consider the order-$k$ Voronoi diagram defined by the positions of
the $n$ variables.  As a clause $c$ contains the $k$ closest
variables, the $k$ variables contained in $c$ are exactly the $k$
variables defining the Voronoi region of $c$'s position.  Independent
of the positions of the $n$ variables, there are only at most
$2k\,(n - k)$ regions in the order-$k$ Voronoi
diagram~\cite{bck-choavd-15}.  Thus, if we have at least
$2^k 2k\,(n - k)$ clauses, then, by the pigeonhole principle, at least
one Voronoi region contains $2^k$ clauses.  As $k$ is considered to be
a constant, this number of clauses is linear in $n$, i.e., we still
have constant density.  Moreover, as repeating the same clause (with
the same variable negations) is avoided whenever possible, there is a
set of $k$ variables that has a clause for every combination of
literals.  Thus, we have an unsatisfiable subformula of constant size
$2^k$, which implies low proof complexity.

This result can be varied and strengthened in multiple ways, e.g., by
allowing weighted variables, a higher dimensional ground space, or by
softening the requirement that every clause contains the $k$ closest
variables (model with higher temperature).
In the following, we briefly discuss how these generalizations can
be achieved.

\subsubsection*{Abstract Geometric Spaces}

The result by~\citet{bck-choavd-15} on the complexity of order-$k$
Voronoi diagrams is very general in the sense that it holds for
abstract Voronoi diagrams.  Roughly speaking, abstract Voronoi
diagrams are based on separating curves between pairs of points that
take the role of perpendicular bisectors.  In this way, one can
abstract from the specific geometric ground space.  Whether a point
$\pnt{p}$ is closer to site $\pnt{s}_1$ or to site $\pnt{s}_2$ is no
longer determined by comparing distances $\dist{\pnt{s}_1}{\pnt{p}}$
and $\dist{\pnt{s}_2}{\pnt{p}}$ but by the curve separating
$\pnt{s}_1$ from $\pnt{s}_2$.  For this to work, the separating curves
have to satisfy a handful of basic axioms.  These are for example
satisfied by perpendicular bisectors in the Euclidean or the
hyperbolic plane.  Thus, the above argumentation for low proof
complexity directly carries over to the hyperbolic plane, or more
generally, to any abstract geometric space satisfying the axioms.

\subsubsection*{Lower Density Via Random Clause Positions}

Assume the variable positions are fixed.  Now choose random positions
for the clauses and observe in which regions of the order-$k$ Voronoi
diagram they end up.  We want to know whether there is a region that
contains at least $2^k$ clauses.  This comes down to a balls into bins
experiment.  Each Voronoi region is a bin and each clause is a ball.
Thus, there are $O(n)$ bins and $m$ balls.  Moreover, we are
interested in the maximum load, i.e., the maximum number of balls that
land in a single bin.  Due to a result by \citet{rs-bb-98}, the
maximum load is \aas{}\ in $\Omega(\frac{\log n}{\log\log n})$ if we
throw $\Omega(\frac{n}{\polylog n})$ balls.  Thus, even for a slightly
sublinear number of balls, the maximum load is superconstant.  We note
that this result holds for uniform bins.  In our case, we have
non-uniform bins, as the probability for a clause to end up in a
particular Voronoi region is proportional to the area of the region.
However, it is not hard to see that the result by \citet{rs-bb-98}
remains true for non-uniform bins; see \shortOrLong{the full
  version~\cite{bfg-ihgpcrs-20}.}{Section~\ref{sec:balls-into-heter}.}
Thus, even if the number of clauses $m$ is slightly sublinear in the
number of variables $n$, we get a small unsatisfiable subformula
asymptotically almost surely if the Voronoi diagram has low
complexity.

\subsubsection*{Positive or Negative Literals with Repetition}

Above we assumed that we get the exact same clause with coinciding
negations twice only if we already have more than $2^k$ clauses with
the same set of $k$ variables.  Although this is arguably a reasonable
assumption for the model, we can make a similar argument without it.
Assume instead that for each variable, we choose the positive and
negative literal uniformly at random, independently of all other
choices.  Moreover, assume for an increasing function $f$, that there
are $f(n)$ clauses that have the same set of $k$ variables.  With the
above balls into bins argument, we, e.g., have
$f(n) \in \Omega(\frac{\log n}{\log\log n})$.  Then the probability
that there is a combination of positive and negative literals that we
did not see at least once is at most $2^k (1 - 2^{-k})^{f(n)}$.  This
probability goes to $0$ for $n\to \infty$, i.e., \aas, there is an
unsatisfiable subformula of constant size~$2^k$.

\subsubsection*{Higher Dimension and Weighted Variables}

At the core of our argument lies the fact that order-$k$ Voronoi
diagrams have linear complexity in the plane.  As already mentioned in
Section~\ref{sec:results-voronoi-diagrams}, this is no longer true for
order-$k$ Voronoi diagrams in higher dimensions or if the variables
have multiplicative weights.  \shortOrLong{}{A formal argument for why
  this property breaks is in Section~\ref{sec:lower-bound-voronoi}.}
However, for sites distributed uniformly at random, we
\shortOrLong{can show}{show in Section~\ref{sec:expect-numb-regions}}
that the complexity can be expected to be linear in the total weight,
even in the more general setting.  Thus, using that the variables have
random positions (a requirement we did not need before), we can apply
the above argument to obtain low proof complexity.

\subsubsection*{Non-Zero Temperature}

Non-zero temperatures make it so that clauses do not necessarily
contain the $k$ closest variables.  Instead, variables are included
with probabilities depending on the distance.  Thus, we cannot simply
look at the order-$k$ Voronoi diagram to determine which variables are
contained in a given clause.  \shortOrLong{To resolve this,}{We
  resolve this issue in
  Section~\ref{sec:geometric-model-with-temperature}.  For this,} we
call a clause \emph{nice}, if it behaves as it would in the $T = 0$
case, i.e., if it includes the $k$ closest variables.  \shortOrLong{We
  can show that,}{In Section~\ref{sec:expected-number-nice-clause} we
  show that,} in expectation, a constant fraction of clauses is
actually nice.  Moreover, \shortOrLong{we can show}{in
  Section~\ref{sec:conc-nice-claus}, we show} that the number of nice
clauses is concentrated around its expectation.  With this, we can
apply the same arguments as before to only the nice clauses, of which
we have linearly many, to obtain a low proof complexity.

\subsection{Voronoi Diagrams}
\label{sec:voronoi-diagrams}

The worst-case lower bounds for the complexity of order-$k$ Voronoi
diagrams follow from existing lower bounds on the number of vertices
together with Theorem~\ref{thm:complexity-vertices-vs-regions}, which
connects the complexity in terms of regions with the complexity in
terms of vertices.  This connection is obtained by observing how the
order-$k$ Voronoi diagram changes when increasing $k$.

For the average-case linear upper bound on the number of regions, the
argument works roughly as follows, assuming the unweighted case for
the sake of simplicity.  For each size-$k$ subset $A$ of the sites, we
devise an upper bound on the probability that $A$ has a non-empty
order-$k$ Voronoi region.  This region is non-empty if and only if
there are points that have $A$ as the $k$ closest sites, i.e., if
there is a ball that contains the sites of $A$ and no other sites.
With this observation, we can use a win-win-style argument.  Either
the radius of this ball is small, which makes it unlikely that all
sites of $A$ lie in the ball, or the ball is large, which makes it
unlikely that it contains no other sites.

\shortOrLong{
}{

\section{Resolution Size of Power-Law Random k-SAT}
\label{sec:bipartite-expansion}

\subsection{The Direct Approach}

As stated in Section~\ref{arg:pl}, a formula
has resolution width $\Omega(w)$ if
\begin{inparaenum}[(1)]
\item every set $S$ of at most $w$ clauses contains at least $|S|$
  different variables and
\item every set $S$ of $\frac{1}{3}w\le|S|\le\frac{2}{3}w$ clauses contains at 
least a constant fraction of unique variables.
\end{inparaenum}
In this section we are going to show that both conditions are satisfied for power-law exponents $\beta>\frac{2k-1}{k-1}$ and clause-variable ratios $\Delta\in\Omega(1)$.
The first condition can also be interpreted in terms of bipartite expansion.
It states that the clause-variable incidence graph $G(\Phi)$ is a $(w,0)$-bipartite expander.
The following lemma states bounds on $w$ for which $G(\Phi)$ is a $(w,0)$-bipartite expander asymptotically almost surely.
These bounds depend on the power-law exponent $\beta$ as well as on the clause-variable ratio $\Delta$.
Note that our choices of $k$ and $\beta$ in the lemma ensure $\eps_1, \eps_2>0$.


\begin{lemma}\label{lem:property1}
Let $\Phi$ be a random power-law $k$-SAT formula with
  $n$ variables, $\Delta\cdot n = m \in\Omega(n)$ clauses, $k \ge
  3$, and power-law exponent $\beta > \frac{2k-1}{k-1}$.
Let $\eps_1=k\cdot\frac{\beta-2}{\beta-1}-1>0$ and $\eps_2=(k-2)\frac{\beta-2}{\beta-1}>0$.
Then $G\left(\Phi\right)$ has $(w,0)$-bipartite expansion \aas if
\begin{enumerate}[(i)]
\item $\beta\in\left(\frac{2k-1}{k-1},3\right)$, $\Delta\in o\left(n^{\eps_1}/\log^{\eps_2}(n)\right)$, and $w\in\Oh\left(n^{\eps_1/\eps_2}\cdot\Delta^{-1/\eps_2}\right)$
\item $\beta=3$, $\Delta\in o\left(n^{(k-2)/2}/\log^{1+(k-2)/2}(n)\right)$, and $w\in\Oh\left(n\cdot\left(\Delta\cdot\ln{n}\right)^{-2/(k-2)}\right)$.
\item $\beta>3$, $\Delta\in o\left(n^{\eps_2}/\log^{\eps_2}{n}\right)$, and $w\in\Oh\left(n\cdot\Delta^{-1/\eps_2}\right)$.
\end{enumerate}
\end{lemma}
\begin{proof}
We are interested in showing $|N(C')|\ge |C'|$ for all $C'\subseteq C$ with $|C'|\le w$.
We consider a smallest $C'$ such that $|N(C')|\le |C'|-1$ and denote it by $\hat{C}$.
Let $\mathcal{E}_i$ be the event that $|\hat{C}|=i$.
Thus, $\mathcal{E}_i$ implies that for all $C'\subseteq C$ with $|C'|< i$ it holds that $|N(C')|\ge |C'|$.
This implies that every variable in $N(\hat{C})$ has to appear at least twice.
Otherwise, one could delete a clause with a unique variable from $\hat{C}$ to get a set $\hat{C}'$ with $|\hat{C}'|=i-1$ and $|N(\hat{C}')|\le i-2$.
This would violate the minimality of $\hat{C}$.
Also, $\hat{C}$ must contain exactly $i-1$ different variables. 
Otherwise, we could remove any clause from $\hat{C}$ and violate minimality.

Now we bound 
\[\sum_{i=1}^w \Pr\left(\mathcal{E}_i\right) \le \sum_{i=1}^w \binom{m}{i} P_{i},\]
where $P_{i}$ is the probability to draw $i$ clauses which contain at most $i-1$ different variables and all of them at least twice.
We can now imagine the $k\cdot i$ variables of the $i$ clauses to be drawn independently with replacement.
This would only increase the probability that the $i$ clauses contain at most $i-1$ different variables and all of them at least twice.
Thus, the probability we consider is an upper bound.
Now we consider the $i-1$ different variables drawn.
Then, we choose the $i-1$ pairs of positions where each variable appears for the first and second time.
As a rough upper bound we have at most \[\binom{\binom{k\cdot i}{2}}{i-1}\le\left(\frac{(k\cdot i)^2\cdot e}{2\cdot (i-1)}\right)^{i-1}\] many possibilities by simply choosing $i-1$ from all $\binom{k\cdot i}{2}$ possible pairs.
Now we bound the probability that at these pairs of positions the same variables do appear.
This is at most $\sum_{j=1}^n p_j^2$ per pair of positions.
At the remaining $k\cdot i-2\cdot (i-1)$ positions we can only choose from at most those $i-1$ variables.
Thus, the probabilities at all other positions are the sum of the $i-1$ variable probabilities, which is at most the sum of the $i-1$ highest variable probabilities.
Let $F(i)$ be the sum of the $i$ highest variable probabilities. 
Then it holds that
\begin{align*}
P_{i} 
	& \le \left(\frac{(k\cdot i)^2\cdot e}{2\cdot (i-1)}\right)^{i-1}\cdot\left(\sum_{j=1}^n p_j^2\right)^{i-1}\cdot F(i-1)^{k\cdot i- 2\cdot (i-1)}\\
	& \le \kappa^{i-1}\cdot\left(\frac{i^2}{i-1}\right)^{i-1}\cdot \left(\sum_{j=1}^n p_j^2\right)^{i-1}\cdot \left(\frac{i-1}{n}\right)^{(k\cdot i- 2\cdot (i-1))\frac{\beta-2}{\beta-1}}
\end{align*}
for a constant $\kappa=\kappa(k,\beta)>0$ that might depend on other parameters, which are fixed to constants as well.
We will use $\kappa$ to collect all constant factors.
According to \lemref{lem:powerlaw}
\begin{equation*}\sum_{j=1}^n p_j^2=
\begin{cases}
\Theta\left(n^{-2\frac{\beta-2}{\beta-1}}\right),& \beta<3;\\
\Theta\left(\ln n/n\right),& \beta=3;\\
\Theta\left(n^{-1}\right),& \beta>3.
\end{cases}
\end{equation*}
Thus, our result depends on the power law exponent $\beta$.
For $\beta<3$ we get
\begin{align*}
P_{i} 
	& \le \kappa^i \cdot\left(\frac{i^2}{i-1}\right)^{i-1}\cdot n^{-(i-1)\cdot2\frac{\beta-2}{\beta-1}}\cdot \left(\frac{i-1}{n}\right)^{(k\cdot i- 2\cdot (i-1))\frac{\beta-2}{\beta-1}}\\
	& \le \kappa^i \cdot n^{-k\cdot i\cdot \frac{\beta-2}{\beta-1}}\cdot i^{i-1} \cdot (i-1)^{(k\cdot i- 2\cdot (i-1))\frac{\beta-2}{\beta-1}}\\
	& \le \kappa^i \cdot  n^{-k\cdot i\cdot \frac{\beta-2}{\beta-1}}\cdot i^{(k\cdot i- 2\cdot (i-1))\frac{\beta-2}{\beta-1}+(i-1)}\\
	& = \kappa^i \cdot  n^{-k\cdot i\cdot \frac{\beta-2}{\beta-1}}\cdot i^{(k\cdot i- 2\cdot i)\frac{\beta-2}{\beta-1}+i+2\frac{\beta-2}{\beta-1}-1}\\
	& \le \kappa^i \cdot  n^{-k\cdot i\cdot \frac{\beta-2}{\beta-1}}\cdot i^{(k\cdot i- 2\cdot i)\frac{\beta-2}{\beta-1}+i}=\kappa^i \cdot  n^{-i\cdot (\eps_1+1)}\cdot i^{i\cdot(\eps_2+1)},
\end{align*}
where we used $\left(\frac{i^2}{i-1}\right)^{i-1}\le {e\cdot i^{i-1}}$ in the second line and upper-bounded $i-1$ in the base by $i$, which we can do since $(k\cdot i- 2\cdot (i-1))\frac{\beta-2}{\beta-1}>0$ due to $k\ge3$ and $\beta>2$.
In the third line, we used $2\frac{\beta-2}{\beta-1}-1<0$, which holds since $\beta<3$.

We can now see that 
\begin{align}
\sum_{i=1}^w \binom{m}{i} P_{i}
	& \le \sum_{i=1}^w \binom{m}{i}\cdot\kappa^i \cdot  n^{-i\cdot (\eps_1+1)}\cdot i^{i\cdot(\eps_2+1)}\notag\\
	& \le \sum_{i=1}^w \kappa^i \cdot \Delta^i\cdot n^{-i\cdot\eps_1}\cdot i^{i\cdot\eps_2},\label{eq:sum}
\end{align}
which holds since we assume $m=\Delta\cdot n$ and $\binom{m}{i}\le\left(\frac{e\cdot m}{i}\right)^i$.
In order to have a sum which is $o(1)$ we want to ensure that 
\[\kappa\cdot\Delta\cdot n^{-\eps_1}\cdot i^{\eps_2}\]
is at most a constant smaller than 1.
It is easy to check that this holds for 
\[i\in\Oh\left(n^{\eps_1/\eps_2}\cdot\Delta^{-1/\eps_2}\right).\]
Thus, we can set $w$ to this value.
If we split the sum in Equation~\eqref{eq:sum} at $i_0=\left\lfloor \eps_1\log n\right\rfloor$, the part with $i\le i_0$ is upper-bounded by $\Oh\left(\Delta\cdot n^{-\eps_1}\cdot {i_0}^{\eps_2}\right)\in\Oh\left(\Delta\left(\log^{\eps_2}(n)/n^{\eps_1}\right)\right)$ via a geometric series. 
The part with $i > i_0$ is upper-bounded by the first term.
If we chose $w\in\Theta(n^{\eps_1/\eps_2}\cdot\Delta^{-1/\eps_2})$ so that $\Delta\cdot n^{-\eps_1}\cdot i^{\eps_2}\le c$ for a constant $c\in(0,1)$, the second term yields at most $c^{\Theta(\log n)}=o(1)$.
Thus, we get $\left(\Theta\left(n^{\eps_1/\eps_2}\cdot\Delta^{-1/\eps_2}\right),0\right)$-expansion with probability at least $1-\Theta\left(\Delta\left(\log^{\eps_2}(n)/n^{\eps_1}\right)\right)$ or \aas if $\Delta\in o(n^{\eps_1}/(\log (n))^{\eps2})$.

For $\beta>3$ we get
\begin{align*}
P_{i} 
	& \le \kappa^i \cdot\left(\frac{i^2}{i-1}\right)^{i-1}\cdot n^{-(i-1)}\cdot \left(\frac{i-1}{n}\right)^{(k\cdot i- 2\cdot (i-1))\frac{\beta-2}{\beta-1}}\\
	& \le \kappa^i \cdot \left(\frac{i}{n}\right)^{i-1+(k\cdot i- 2\cdot (i-1))\frac{\beta-2}{\beta-1}}.
\end{align*}
Thus,
\begin{align*}
\sum_{i=1}^r \binom{m}{i} P_{i}
	& \le \sum_{i=1}^w \left(\frac{n}{i}\right)^i\cdot \kappa^i \cdot \Delta^i\cdot\left(\frac{i}{n}\right)^{i-1+(k\cdot i- 2\cdot (i-1))\frac{\beta-2}{\beta-1}}\\
	& = \sum_{i=1}^w \kappa^i \cdot \Delta^i\cdot \left(\frac{i}{n}\right)^{(k\cdot i- 2\cdot i)\frac{\beta-2}{\beta-1}+2\frac{\beta-2}{\beta-1}-1}\\
	& \le \sum_{i=1}^w \kappa^i \cdot \Delta^i\cdot \left(\frac{i}{n}\right)^{i\cdot(k-2)\frac{\beta-2}{\beta-1}}\\
	& \le \sum_{i=1}^w \kappa^i \cdot \Delta^i\cdot \left(\frac{i}{n}\right)^{i\cdot\eps_2},
\end{align*}
which holds since $\frac{i}{n}\le1$ and $2\cdot\frac{\beta-2}{\beta-1}-1\ge 0$ for $\beta\ge3$.
It is now easy to show that $\kappa\cdot\Delta\cdot \left(i/n\right)^{\eps_2}$ is at most a small constant for $w\in\Theta(n\cdot\Delta^{-1/\eps_2})$ sufficiently small.
By splitting the sum as before, we can show $((n\cdot\Delta^{-1/\eps_2}),0)$-expansion with probability at least $1-\Theta(\Delta\cdot\log^{\eps_2}{n}/n^{\eps_2})$ or \aas for $\Delta\in o\left(n^{\eps_2}/\log^{\eps_2}n\right)$.

For $\beta=3$ we get the same result as for $\beta>3$, except for an additional factor of $\left(\ln{n}\right)^{i-1}$. 
Thus,
\begin{align*}
\sum_{i=1}^w \binom{m}{i} P_{i}
	& \le \sum_{i=1}^w \kappa^i \cdot\Delta^i\cdot \left(\frac{i}{n}\right)^{i\cdot(k-2)\frac{\beta-2}{\beta-1}} \ln^{i-1}{n}\\
	& \le \sum_{i=1}^w \kappa^i \cdot \Delta^i\cdot\left(i\cdot\frac{\ln^{2/(k-2)}{n}}{n}\right)^{i\cdot\frac{k-2}{2}}.
\end{align*}
By assuming 
\[w\in\Theta\left(n\cdot\left(\Delta\cdot\log{n}\right)^{-2/(k-2)}\right)\]
small enough, we can ensure that this sum is at most $\Oh(\Delta\cdot\log{n}\cdot\left(\log(n)/n\right)^{(k-2)/2})$ by splitting the expression at $\left\lfloor i_0=\ln n\right\rfloor$ again.
Hence, we get $\big(\Theta(n\cdot(\Delta\cdot\log{n})^{-2/(k-2)}),0\big)$-expansion with probability at least $1-O\big(\Delta\cdot\ln{n}\cdot(\log(n)/n)^{(k-2)/2}\big)$ or \aas for $\Delta\in o\big(n^{(k-2)/2}/\log^{(k-2)/2+1}(n)\big)$.
\end{proof}

Now we want to show the second requirement of \thmref{thm:resolution}, that every set $S$ of $\frac{1}{3}w\le|S|\le\frac{2}{3}w$ clauses contains at 
least a constant fraction of unique variables.
Again, our choices of $k$ and $\beta$ in the lemma ensure that we can always choose an $\eps>0$ with $\eps_1, \eps_2>0$.

\begin{lemma}\label{lem:property2}
Let $\Phi$ be a random power-law $k$-SAT formula with
  $n$ variables, $\Delta\cdot n = m \in\Omega(n)$ clauses, $k \ge
  3$, and power-law exponent $\beta > \frac{2k-1}{k-1}$.
Let $\eps, \eps_1, \eps_2$ be constant such that $\eps>0$, $\eps_1=\frac{k-\eps}{2}-1>0$, and $\eps_2=(k-\eps)\cdot\frac{\beta-2}{\beta-1}-1>0$.
There is a $W$ such that for all $w\in \omega(1)$ with $w\le W$ \aas all sets $C'$ of clauses from $\Phi$ with $\frac13 w\le|C'|\le\frac23 w$ contain at least $\eps\cdot|C'|$ unique variables. 
It holds that:
\begin{enumerate}[(i)]
\item If $\beta\in\left(\frac{2k-1}{k-1},3\right)$ and $\Delta\in o\left(n^{\eps_2}\right)$, then $W\in\Theta\left(n^{\eps_2/\eps_1}\cdot\Delta^{-1/\eps_1}\right)$.
\item If $\beta=3$ and $\Delta\in o\left(n^{\eps_1}/\ln^{\eps_1+1}n\right)$, then $W\in\Theta\left(n\cdot\Delta^{-1/\eps_1}/\ln^{1+\frac{1}{\eps_1}}n\right)$.
\item If $\beta>3$ and $\Delta\in o\left(n^{\eps_1}\right)$, then $W\in\Theta\left(n\cdot\Delta^{-1/\eps_1}\right)$.
\end{enumerate}
\end{lemma}
\begin{proof}
Let $\eps\in(0,\min\{k-\frac{\beta-1}{\beta-2},k-2\})$ be a constant. 
The upper bounds on $\eps$ ensure $\eps_1>0$ and $\eps_2>0$.
We want to bound the probability that there is a set of clauses $C'$ with $\frac13 w \le |C'| \le \frac23 w$ and at most $\eps\cdot |C'|$ many unique variables.
Let $P_i$ be the probability that there is a set $C'$ of size $i$ with that property.
We assume the $k\cdot i$ Boolean variables to be drawn independently at random, \ie, we allow duplicate variables inside clauses.
This only decreases the probability of having unique variables.
Additionally, we split the probability into parts depending on the number $j$ of different variables that appear in $C'$ in addition to the $\eps\cdot i$ unique ones.
It holds that
\begin{align*}
P_i & \le \underbrace{\binom{m}{i}}_{\substack{\text{choices of}\\\text{clauses}}} \cdot \sum_{j=1}^{\frac{k-\eps}{2}\cdot i} \underbrace{\binom{k\cdot i}{\eps\cdot i}}_{\substack{\text{possible posi-}\\\text{tions for the}\\\text{$\eps\cdot i$ unique}\\\text{variables}}} \underbrace{\binom{\binom{(k-\eps)\cdot i}{2}}{j}}_{\substack{\text{possible positions}\\\text{for the first}\\\text{two appearances}\\\text{of the $j$}\\\text{other variables}}} \cdot \underbrace{1^{\eps\cdot i}}_{\substack{\text{probability to}\\\text{draw a new}\\\text{variable}}} \underbrace{\left(\sum_{x=1}^n p_x^2\right)^j}_{\substack{\text{probability}\\\text{that variables}\\\text{are same at}\\\text{positions for}\\\text{first two}\\\text{appearances}}} \cdot\underbrace{F(j)^{k\cdot i -\eps\cdot i -2j}}_{\substack{\text{upper bound on}\\\text{probability to draw}\\\text{$j$ chosen variables again}}}\\
 & \le \kappa^i\cdot\Delta^i\left(\frac{n}{i}\right)^i \cdot \sum_{j=1}^{\frac{k-\eps}{2}\cdot i} \left(\frac{k\cdot i}{\eps\cdot i}\right)^{\eps\cdot i}\left(\frac{(k-\eps)^2 \cdot i^2 }{j}\right)^{j}\cdot \left(\sum_{x=1}^n p_x^2\right)^j \cdot F(j)^{k\cdot i -\eps\cdot i -2j}\\
 & \le \kappa^i\cdot\Delta^i\left(\frac{n}{i}\right)^i \cdot \sum_{j=1}^{\frac{k-\eps}{2}\cdot i} \left(\frac{i^2}{j}\right)^{j}\cdot \left(\sum_{x=1}^n p_x^2\right)^j \cdot \left(\frac{j}{n}\right)^{(k\cdot i -\eps\cdot i -2j)\frac{\beta-2}{\beta-1}},
\end{align*}
where $\kappa=\kappa(k,\eps,\beta)>0$ is a constant that might depend on other parameters, which are fixed to constants. 
Note that we estimated the probability to draw a new (unique) variable with $1$.
Thus, this also accounts for the probability to draw a variable that is not actually new.
Especially, it accounts for the probability to draw one of the $j$ non-unique variables.
This means, the expression we have is an upper bound for the probability to draw \emph{at most} $\eps\cdot i$ unique variables.
As in the proof of \lemref{lem:property1} we have to distinguish three cases depending on the power law exponent $\beta$.
Using \lemref{lem:powerlaw} we see that for $\beta<3$
\begin{align}
P_i &\le \kappa^i\cdot\Delta^i\left(\frac{n}{i}\right)^i \cdot \sum_{j=1}^{\frac{k-\eps}{2}\cdot i} \left(\frac{i^2}{j}\right)^{j}\cdot \left(\sum_{x=1}^n p_x^2\right)^j \cdot \left(\frac{j}{n}\right)^{(k\cdot i -\eps\cdot i -2j)\frac{\beta-2}{\beta-1}}\notag\\
	& \le \kappa^i\cdot\Delta^i\left(\frac{n}{i}\right)^i \cdot \sum_{j=1}^{\frac{k-\eps}{2}\cdot i} \left(\frac{i^2}{j}\right)^{j}\cdot n^{-2j\frac{\beta-2}{\beta-1}} \cdot \left(\frac{j}{n}\right)^{(k\cdot i -\eps\cdot i -2j)\frac{\beta-2}{\beta-1}}\notag\\
		& = \kappa^i \cdot\Delta^i\cdot n^{i\left(1-(k-\eps)\frac{\beta-2}{\beta-1}\right)} \cdot i^{-i}\cdot \sum_{j=1}^{\frac{k-\eps}{2}\cdot i} \left(\frac{i^2}{j}\right)^{j}\cdot j^{(k\cdot i -\eps\cdot i -2j)\frac{\beta-2}{\beta-1}}\label{eq:pi-bound}.
\end{align}
Now it remains to bound the inner sum.
In order to do so, we will split it at $j_0=\frac{3-\beta}{4}(k-\eps)\cdot i$. 
It is easy to see that $0< \frac{3-\beta}{4} < \frac{1}{4}$ for $2< \beta <3$, thus this choice of $j$ is valid.
For the first part of the sum it holds that
\begin{align*}
\sum_{j=1}^{\frac{3-\beta}{4}(k-\eps)\cdot i} \left(\frac{i^2}{j}\right)^{j}\cdot j^{(k\cdot i -\eps\cdot i -2j)\frac{\beta-2}{\beta-1}}	& \le \kappa^i\sum_{j=1}^{\frac{3-\beta}{4}(k-\eps)\cdot i} i^{2\cdot j}\cdot j^{-j}\cdot i^{((k-\eps)\cdot i -2j)\frac{\beta-2}{\beta-1}}\\
	& \le \kappa^i\cdot i^{(k-\eps)\cdot i\cdot\frac{\beta-2}{\beta-1}}\sum_{j=1}^{\frac{3-\beta}{4}(k-\eps)\cdot i} i^{\frac{2\cdot j}{\beta-1}}\\
  & \le \kappa^i\cdot i^{(k-\eps)\cdot i\cdot\frac{\beta-2}{\beta-1}}\cdot i^{2\cdot \frac{3-\beta}{4}\cdot\frac{k-\eps}{\beta-1}\cdot i}\\
	& = \kappa^i\cdot i^{\frac{k-\eps}{2}\cdot i},
\end{align*}
where we used $j\le\frac{3-\beta}{4}(k-\eps)\cdot i$ and $((k-\eps)\cdot i -2j)\ge 0$ in the first line.
The derived sum in the second line is a geometric series with base $i^{\frac{2}{\beta-1}}\ge 1$.
This series is dominated by the term with $j=\frac{3-\beta}{4}(k-\eps)\cdot i$.
Additional factors of at most $c^i$ for positive constants $c$ are hidden in $\kappa^i$.
For the second part of the sum it holds that
\begin{align*}
\sum_{j=\frac{3-\beta}{4}(k-\eps)\cdot i}^{\frac{k-\eps}{2}\cdot i} \left(\frac{i^2}{j}\right)^{j}\cdot j^{(k\cdot i -\eps\cdot i -2j)\frac{\beta-2}{\beta-1}}	& \le \kappa^i\sum_{j=\frac{3-\beta}{4}(k-\eps)\cdot i}^{\frac{k-\eps}{2}\cdot i} i^{2\cdot j}\cdot j^{-j}\cdot i^{((k-\eps)\cdot i -2j)\frac{\beta-2}{\beta-1}}\\
	& \le \kappa^i\cdot i^{(k-\eps)\cdot i\cdot\frac{\beta-2}{\beta-1}}\sum_{j=\frac{3-\beta}{4}(k-\eps)\cdot i}^{\frac{k-\eps}{2}\cdot i} i^{j-2\cdot j\frac{\beta-2}{\beta-1}}\\
	& \le \kappa^i\cdot i^{(k-\eps)\cdot i\cdot\frac{\beta-2}{\beta-1}}\cdot i^{\frac{3-\beta}{\beta-1}\cdot\frac{k-\eps}{2}\cdot i}\\
	& = \kappa^i\cdot i^{\frac{k-\eps}{2}\cdot i},
\end{align*}
where we used $j\ge\frac{3-\beta}{4}(k-\eps)\cdot i$ in the second and a geometric series in the third line.
The base of the series is $i^{\frac{3-\beta}{\beta-1}}\ge 1$.
Thus, the last term with $j=\frac{k-\eps}{2}\cdot i$ dominates and we get the shown estimate with factors $c^i$ for positive constants $c$ hidden in $\kappa^i$ again.

Thus,
\[\sum_{j=1}^{\frac{k-\eps}{2}\cdot i} \left(\frac{i^2}{j}\right)^{j}\cdot j^{(k\cdot i -\eps\cdot i -2j)\frac{\beta-2}{\beta-1}}\le \kappa^i \cdot i^{\frac{k-\eps}{2}\cdot i}\]
and plugging this into \eq{pi-bound} yields
\[P_i\le \kappa^i \cdot\Delta^i\cdot n^{i\left(1-(k-\eps)\frac{\beta-2}{\beta-1}\right)} \cdot i^{i\left(\frac{k-\eps}{2}-1\right)}=\kappa^i \cdot\Delta^i\cdot n^{-\eps_2\cdot i} \cdot i^{\eps_1\cdot i}.\]

Since we want to sum over all $P_i$ with  $\frac13 w \le i \le \frac23 w$ for some $w$, it holds that
\begin{align*}
\sum_{i=\frac13 w}^{\frac23 w} P_i & \le \sum_{i=\frac13 w}^{\frac23 w} \kappa^i \cdot\Delta^i\cdot n^{-\eps_2\cdot i} \cdot i^{\eps_1\cdot i} \\
 & \le \sum_{i=\frac13 w}^{\frac23 w} \left(\kappa \cdot\Delta\cdot n^{-\eps_2}\cdot w^{\eps_1}\right)^i
\end{align*}
This sums up to $o(1)$ as soon as $\kappa\cdot\Delta\cdot n^{-\eps_2}\cdot w^{\eps_1}$ is a suitably small constant and $w$ is super-constant.
In our case, we see that this holds for some 
\[w\in\Oh\left(n^{\eps_2/\eps_1}\Delta^{-1/\eps_1}\right).\]

For $\beta=3$ we get
\begin{align}
P_i &\le \kappa^i\cdot\Delta^i\left(\frac{n}{i}\right)^i \cdot \sum_{j=1}^{\frac{k-\eps}{2}\cdot i} \left(\frac{i^2}{j}\right)^{j}\cdot \left(\sum_{x=1}^n p_x^2\right)^j \cdot \left(\frac{j}{n}\right)^{((k-\eps)\cdot i -2j)\frac{\beta-2}{\beta-1}}\notag\\
	& \le \kappa^i\cdot\Delta^i\left(\frac{n}{i}\right)^i \cdot \sum_{j=1}^{\frac{k-\eps}{2}\cdot i} \left(\frac{i^2}{j}\right)^{j}\cdot \left(\frac{\ln n}{n}\right)^{j} \cdot \left(\frac{j}{n}\right)^{((k-\eps)\cdot i  -2j)\frac12}\notag\\
		& = \kappa^i \cdot\Delta^i\cdot \left(\frac{n}{i}\right)^{i} \cdot n^{-\frac{k-\eps}{2}i}\sum_{j=1}^{\frac{k-\eps}{2}\cdot i} i^{2j}\cdot j^{\frac{k-\eps}{2}i-2j}\cdot \left(\ln n\right)^j\label{eq:pi-bound2}.
\end{align}
We want to show that this inner sum is at most $\kappa^i\cdot\left(i\cdot \ln n\right)^{\frac{k-\eps}{2}i}$.
As before, we can split the sum.
This time we split it at $j_0=\frac{k-\eps}{4}i$.
For the first part we get
\begin{align*}
\sum_{j=1}^{\frac{k-\eps}{4}i} i^{2j}\cdot j^{\frac{k-\eps}{2}i-2j}\cdot \left(\ln n\right)^j
& \le \kappa^i\cdot\sum_{j=1}^{\frac{k-\eps}{4}i} i^{2j}\cdot i^{\frac{k-\eps}{2}i-2j}\cdot \left(\ln n\right)^j\\
& \le \kappa^i\cdot i^{\frac{k-\eps}{2}i}\cdot\sum_{j=1}^{\frac{k-\eps}{4}i} \left(\ln n\right)^j\\
& \le \kappa^i\cdot i^{\frac{k-\eps}{2}i}\cdot\left(\ln n\right)^{\frac{k-\eps}{4}i},
\end{align*}
where we used that $\frac{k-\eps}{2}i-2j\ge0$ in the first line.
The second line contains a geometric series with base $\ln n\ge 1$ again that we estimated by its dominating term $\left(\ln n\right)^{\frac{k-\eps}{4}i}$.
The second part of the sum yields
\begin{align*}
\sum_{j=\frac{k-\eps}{4}i}^{\frac{k-\eps}{2}i} i^{2j}\cdot j^{\frac{k-\eps}{2}i-2j}\cdot \left(\ln n\right)^j
& \le \kappa^i\cdot\sum_{j=\frac{k-\eps}{4}i}^{\frac{k-\eps}{2}i} i^{2j}\cdot i^{\frac{k-\eps}{2}i-2j}\cdot \left(\ln n\right)^j\\
& \le \kappa^i\cdot i^{\frac{k-\eps}{2}i}\sum_{j=\frac{k-\eps}{4}i}^{\frac{k-\eps}{2}i} \left(\ln n\right)^j \le \kappa^i\cdot i^{\frac{k-\eps}{2}i}\left(\ln n\right)^{\frac{k-\eps}{2}i},
\end{align*}
since $j\in\Theta(i)$.
Plugging this into \eq{pi-bound2} gives us
\[P_i \le \kappa^i \cdot\Delta^i\cdot \left(\frac{n}{i}\right)^{i\left(1-\frac{k-\eps}{2}\right)}\cdot\left(\ln n\right)^{\frac{k-\eps}{2}i}=\kappa^i \cdot\Delta^i\cdot \left(\frac{n}{i}\right)^{-\eps_1\cdot i}\cdot\left(\ln n\right)^{(\eps_1+1)\cdot i}.\]
As before, we can see that this is at most $\kappa^i$ for some constant $\kappa\in(0,1)$ if 
\[w\in\Oh\left(n/\ln^{1+\frac{1}{\eps_1}} n\cdot\Delta^{-1/\eps_1}\right)\] is small enough.

For $\beta>3$ we get 
\begin{align}
P_i &\le \kappa^i\cdot \Delta^i\left(\frac{n}{i}\right)^i \cdot  \sum_{j=1}^{\frac{k-\eps}{2}\cdot i} \left(\frac{i^2}{j}\right)^{j}\cdot \left(\sum_{x=1}^n p_x^2\right)^j \cdot \left(\frac{j}{n}\right)^{((k-\eps)\cdot i -2j)\frac{\beta-2}{\beta-1}}\notag\\
	& \le \kappa^i\cdot\Delta^i\left(\frac{n}{i}\right)^i \cdot \sum_{j=1}^{\frac{k-\eps}{2}\cdot i} \left(\frac{i^2}{j}\right)^{j}\cdot n^{-j} \cdot \left(\frac{j}{n}\right)^{((k-\eps)\cdot i-2j)\frac{\beta-2}{\beta-1}}\notag\\
		& = \kappa^i \cdot\Delta^i\cdot \left(\frac{n}{i}\right)^i\cdot n^{-(k-\eps)\frac{\beta-2}{\beta-1}i} \sum_{j=1}^{\frac{k-\eps}{2}\cdot i} i^{2j}\cdot n^{j\left(2\frac{\beta-2}{\beta-1}-1\right)}\cdot j^{((k-\eps)\cdot i-2j)\frac{\beta-2}{\beta-1}-j}.\label{eq:pi-bound3}
\end{align}
This time we are going to show that the inner sum is bounded by $i^{(k-\eps)\frac{\beta-2}{\beta-1}i}\cdot\left(\frac{n}{i}\right)^{\frac{k-\eps}{2}i\left(2\frac{\beta-2}{\beta-1}-1\right)}$.
Again, we split the sum.
This time at 
\[j_0=\frac{(k-\eps)\frac{\beta-2}{\beta-1}}{1+2\frac{\beta-2}{\beta-1}}i.\]
Our choice ensures $((k-\eps)\cdot i-2j)\frac{\beta-2}{\beta-1}-j\ge 0$ for $j\le j_0$.
Thus, in the first part of the sum all exponents are positive. 
It now holds that 
\[j^{((k-\eps)\cdot i-2j)\frac{\beta-2}{\beta-1}-j}\le j_0^{((k-\eps)\cdot i-2j)\frac{\beta-2}{\beta-1}-j}\le \kappa^i\cdot i^{((k-\eps)\cdot i-2j)\frac{\beta-2}{\beta-1}-j}\]
for some constant $\kappa$ that we can incorporate in the $\kappa$ we already have.
In the second part of the sum the exponent $((k-\eps)\cdot i-2j)\frac{\beta-2}{\beta-1}-j$ is negative.
However, we know that the base is $j\ge j_0=\frac{(k-\eps)\frac{\beta-2}{\beta-1}}{1+2\frac{\beta-2}{\beta-1}}i$.
Thus, 
\[j^{((k-\eps)\cdot i-2j)\frac{\beta-2}{\beta-1}-j}\le j_0^{((k-\eps)\cdot i-2j)\frac{\beta-2}{\beta-1}-j}\le \kappa^i\cdot i^{((k-\eps)\cdot i-2j)\frac{\beta-2}{\beta-1}-j}\]
as well.
This yields
\begin{align*}
& \sum_{j=1}^{\frac{k-\eps}{2}\cdot i} i^{2j}\cdot n^{j\left(2\frac{\beta-2}{\beta-1}-1\right)}\cdot j^{((k-\eps)\cdot i-2j)\frac{\beta-2}{\beta-1}-j}\\
& \le \kappa^i \cdot  \sum_{j=1}^{\frac{k-\eps}{2}\cdot i} i^{2j}\cdot n^{j\left(2\frac{\beta-2}{\beta-1}-1\right)}\cdot i^{((k-\eps)\cdot i-2j)\frac{\beta-2}{\beta-1}-j}\\
& = \kappa^i \cdot  i^{(k-\eps)\frac{\beta-2}{\beta-1}i}\cdot\sum_{j=1}^{\frac{k-\eps}{2}\cdot i} \left(\frac{n}{i}\right)^{j\left(2\frac{\beta-2}{\beta-1}-1\right)}\\
& \le \kappa^i \cdot  i^{(k-\eps)\frac{\beta-2}{\beta-1}i}\cdot\left(\frac{n}{i}\right)^{\frac{k-\eps}{2}i\left(2\frac{\beta-2}{\beta-1}-1\right)},
\end{align*}
where the last line holds, since $2\frac{\beta-2}{\beta-1}-1>0$, which implies that we have a geometric series with base at least one again, that we estimate by its dominating term, \ie the term with $j=\frac{k-\eps}{2}\cdot i$.
If we plug our estimate into \eq{pi-bound3} this gives us
\[P_i \le \kappa^i \cdot \Delta^i\left(\frac{n}{i}\right)^{\left(1-\frac{k-\eps}{2}\right)i}=\kappa^i \cdot \Delta^i\left(\frac{n}{i}\right)^{-\eps_1\cdot i}.\]
We can now find a $w\in\Theta(n\cdot \Delta^{-1/\eps_1})$ small enough such that the property holds as desired.

In all three cases we can choose $w$ in such a way that the probability for the property not to hold is at most $\kappa^{\frac{w}{3}}$ for some constant $\kappa\in\left(0,1\right)$.
This means, the property holds \aas for $w\in\omega(1)$.
\end{proof}

The two properties we showed in \lemref{lem:property1} and \lemref{lem:property2} can be used to derive lower bounds on resolution width via the following theorem by \citet{shortProofs}.
\begin{theorem}[\cite{shortProofs}] \label{thm:resolution}
Let $\Phi$ be an unsatisfiable $k$-CNF formula with $k\ge3$.
If there is a $w\in \mathbb{N}$ such that
\begin{enumerate}[(i)]
\item for all sets of clauses $C'$ with $|C'|\le w$ it holds that $C'$ contains at least $|C'|$ different Boolean variables and
\item for all sets of clauses $C'$ with $\frac13 w\le|C'|\le\frac23 w$ it holds that $C'$ contains at least $\eps\cdot|C'|$ unique variables for some constant $\eps>0$.
\end{enumerate}
then the resolution width of $\Phi$ is $\Omega(w)$.
\end{theorem}


\lemref{lem:property1} and \lemref{lem:property2} together with \thmref{thm:resolution} imply Corollary~\ref{cor:direct}.
However, \thmref{thm:resolution} only works for unsatisfiable instances. 
Since the two lemmas do not condition on instances being unsatisfiable, we also need to make sure that the probability for having unsatisfiable instances is large enough.
In particular, we have to guarantee that this probability is larger than the error probabilities of \lemref{lem:property1} and \lemref{lem:property2}.
If the probability of generating unsatisfiable instances is asymptotically larger than those error probabilities, the conditional probability of our width lower bounds to hold conditioned on instances being unsatisfiable will be approaching one.
Since the error probabilities of the two lemmas are $o(1)$, we want the clause-variable ratio $\Delta$ to be high enough for instances to be unsatisfiable with at least constant probability.
The resulting corollary is stated below.
It only holds for unsatisfiable instances as well, \ie the probability bound on resolution width is actually a conditional probability conditioned on instances being unsatisfiable.
\begin{corollary}\label{cor:direct}
Let $\Phi$ be an unsatisfiable random power-law $k$-SAT formula with $n$ variables, $m\in\Omega(n)$ clauses, $k\ge3$, and power-law exponent $\beta>\frac{2k-1}{k-1}$ constant.
Let $\Delta=m/n$ be large enough so that $\Phi$ is unsatisfiable at least with constant probability.
Let $\eps,\eps_1,\eps_2$ be constants with $\eps>0$, $\eps_1=\frac{k-\eps}{2}-1>0$, and $\eps_2=(k-\eps)\cdot\frac{\beta-2}{\beta-1}-1>0$.
For the resolution width $w$ of $\Phi$, it holds \aas that:
%
%
\begin{enumerate}[(i)]
\item If $\beta\in\left(\frac{2k-1}{k-1},3\right)$ and $\Delta\in o\left(n^{\eps_2}\right)$, then $w\in\Omega\left(n^{\eps_2/\eps_1}\cdot \Delta^{-1/\eps_1}\right)$.
\item If $\beta=3$ and $\Delta\in o\left(n^{\eps_1}/\log^{\eps_1+1}n\right)$, then $w\in\Omega\left(n\cdot\Delta^{-1/\eps_1}/\log^{1+\frac{1}{\eps_1}}n\right)$.
\item If $\beta>3$ and $\Delta\in o\left(n^{\eps_1+1}\right)$, then $w\in\Omega\left(n\cdot\Delta^{-1/\eps_1}\right)$.
\end{enumerate}
\end{corollary}
\begin{proof}
If both \lemref{lem:property1} and \lemref{lem:property2} hold, we can use \thmref{thm:resolution} to get the desired bound on resolution width.
As stated before, \thmref{thm:resolution} only holds for unsatisfiable instances.
Thus, if a random formula $\Phi$ is unsatisfiable at least with constant probability, it holds that the conditional probability for the bounds stated in the corollary to hold is at least
\[\frac{\Pr\left(\Phi\text{ unsat}\right)-o(1)}{\Pr\left(\Phi\text{ unsat}\right)}=1-o(1),\]
conditioned on $\Phi$ being unsatisfiable, where the $o(1)$ term is the error probability from \lemref{lem:property1} and \lemref{lem:property2}.
We are going to show that the values of $w$ from \lemref{lem:property2} are smaller than those from \lemref{lem:property1}.
The expansion bound from \lemref{lem:property1} also holds for those smaller values of $w$ due to the definition of bipartite expansion.
Thus, the bound from \lemref{lem:property2} gives us the maximum $w$ we can achieve.

First, consider the case $\beta\in(\frac{2k-1}{k-1},3)$.
Let $\eps_3=k\frac{\beta-2}{\beta-1}-1$ and $\eps_4=(k-2)\frac{\beta-2}{\beta-1}$.
We want to show that
\begin{equation}\label{eq:unify}
n^{\eps_2/\eps_1}\cdot \Delta^{-1/\eps_1}\le n^{\eps_3/\eps_4}\cdot\Delta^{-1/\eps_4}.
\end{equation}
Both bounds only hold for 
\[\Delta\in o\left(n^{\eps_2}\right)\subseteq o\left(n^{\eps_3}/\log^{\eps_4}(n)\right),\]
since $\eps_2=(k-\eps)\frac{\beta-2}{\beta-1}-1<k\frac{\beta-2}{\beta-1}-1=\eps_3$.
It holds that 
\[\eps_2/\eps_1=\frac{(k-\eps)\cdot\frac{\beta-2}{\beta-1}-1}{\frac{k-\eps}{2}-1}<\frac{k\cdot\frac{\beta-2}{\beta-1}-1}{\frac{k}{2}-1}<\frac{k\frac{\beta-2}{\beta-1}-1}{(k-2)\cdot\frac{\beta-2}{\beta-1}}=\eps_3/\eps_4.\]

We can now distinguish four cases.
First, assume $\Delta\ge1$.
If $\eps_1\le\eps_4$, then $\Delta^{-1/\eps_1}\le \Delta^{-1/\eps_4}$, which implies Inequality~\eqref{eq:unify}.
If $\eps_1>\eps_4$, we need to ensure \[\Delta\le n^{\left(\frac{\eps_3}{\eps_4}-\frac{\eps_2}{\eps_1}\right)/\left(\frac{1}{\eps_4}-\frac{1}{\eps_1}\right)}.\]
This is already the case, since we assume $\Delta\in o(n^{\eps_2})$ and $\eps_2\le(\frac{\eps_3}{\eps_4}-\frac{\eps_2}{\eps_1})/(\frac{1}{\eps_4}-\frac{1}{\eps_1})$ due to $\eps_1>\eps_4$ and $\eps_3\ge\eps_2$. 
Thus, Inequality~\eqref{eq:unify} holds.

Now assume $\Delta<1$.
If $\eps_1\le\eps_4$, we need to ensure that \[\Delta\ge n^{(\frac{\eps_2}{\eps_1}-\frac{\eps_3}{\eps_4})/(\frac{1}{\eps_1}-\frac{1}{\eps_4})}.\]
This already holds, since we assume $\Delta\in\Omega(1)$ and $(\frac{\eps_2}{\eps_1}-\frac{\eps_3}{\eps_4})/(\frac{1}{\eps_1}-\frac{1}{\eps_4})\le 0$ due to $\eps_1\le\eps_4$ and $\eps_2/\eps_1\le\eps_3/\eps_4$.
Thus, Inequality~\eqref{eq:unify} holds.
If $\eps_4\le\eps_1$, then $\Delta^{-1/\eps_1}\le\Delta^{-1/\eps_4}$ and Inequality~\eqref{eq:unify} holds as well.

Now consider $\beta=3$.
We need to show that
\[n/\left(\ln n\right)^{\frac{\eps_1+1}{\eps_1}}\Delta^{-1/\eps_1}=n/\left(\ln n\right)^{\frac{k-\eps}{k-\eps-2}}\Delta^{-2/(k-\eps-2)}\in\Oh(n\cdot\left(\Delta\cdot\ln{n}\right)^{-2/(k-2)}).\]
Again, the left-hand side is from \lemref{lem:property2} and the right-hand side is from \lemref{lem:property1}.
This holds, due to our assumption $\Delta\in\Omega(1)$ and since $\eps_1=\frac{k-\eps}{2}-1>0$ implies $0<\eps<k-2$ and thus $\frac{k-\eps}{k-\eps-2}>\frac{2}{k-\eps-2}>\frac{2}{k-2}$.
Additionally, the bound only holds up to $\Delta\in o \big(n^{(k-\eps-2)/2}/\ln^{(k-\eps)/2}(n)\big)\subseteq o\big(n^{(k-2)/2}/\log^{(k-2)/2+1}(n)\big)$.

For $\beta>3$ we have to show
\[n\cdot\Delta^{-1/\eps_1}\in\Oh(n\cdot\Delta^{-1/\eps_4})\]
as well as $\Delta\in o(n^{\eps_1})\subseteq o\big((n/\log{n})^{\eps_4}\big)$.
This holds since $\eps_1=\frac{k-\eps}{2}-1\le(k-2)\frac{\beta-2}{\beta-1}=\eps_4$ due to $\beta>3$.
This shows that in all three cases the bounds from \lemref{lem:property2} are smaller, thus giving us the lower bounds on resolution width as stated in the corollary.
\end{proof}

This is nearly the statement of \thmref{thm:mainone}.
However, via bipartite expansion we can already show linear resolution width at constant clause-variable ratios for $\beta>\frac{2k-2}{k-2}$ instead of $\beta>3$.
This gives a better bound for $k\ge5$.
The bounds on bipartite expansion and the resulting bounds on resolution width will be derived in the next section.

\subsection{A Lower Bound on Bipartite Expansion}\label{sec:BipExpansion}

In this section we show an improved bound on the bipartite expansion.
We will use it to obtain a linear lower bound on resolution width for $\beta > \frac{2k - 2}{k - 2}$, which is
potentially smaller than $3$, and therefore improves the previous bound.  
Recall that linear resolution width implies exponential resolution size, 
and thus also exponential tree-like resolution size.
Moreover, our bound on the bipartite expansion can also be used to
bound the so-called resolution clause space, which additionally yields
an exponential lower bound on tree-like resolution size for
$\beta > \frac{2k - 3}{k - 2}$ as we will see at the end of this section.
The following lemma shows the bipartite
expansion property.


\begin{lemma}\label{lem:expansion}
  Let $\Phi$ be a random power-law $k$-SAT formula with $n$ variables,
  $m$ clauses, $k \ge 3$, power-law exponent
  $\beta > \frac{2k-3}{k-2}$, and let $\eps\in(0,(k-1)\cdot\frac{\beta-2}{\beta-1}-1)$ constant.  If
  $\Delta = m/n \in o\big(n^{\eps}/\log^{\eps} n\big)$, then there
  exists an $r \in \Theta\big(n\cdot\Delta^{-1/\eps}\big)$ such
  that the clause-variable incidence graph $G(\Phi)$ is an
  $(r, c)$-bipartite expander \aas{}\ for
  $c=(k-1)-(1+\eps)\cdot\frac{\beta-1}{\beta-2}$.
\end{lemma}
\begin{proof}
First, note that our choice of $\beta > \frac{2k-3}{k-2}$ guarantees that the interval $(0,(k-1)\cdot\frac{\beta-2}{\beta-1}-1)$ from which we choose $\eps$ is not empty.
This interval is chosen in such a way that $c>0$ is guaranteed.
As in the proof of ~\cite[Lemma~5.1]{BenGalesi}, we define a bad event $\E$, that $G\left(\Phi\right)$ is not an $(r,c)$-bipartite expander.
If $\E$ happens, then there is a set $C' \subseteq C$ with $1\le|C'|\le r$ such that $|N(C')|<(1+c)\cdot|C'|$.
Given a set $C'\subseteq C=[m]$ of clause indices with $|C'|=i$ we want to bound the probability $P_{i}$ that the $k\cdot i$ indices of variables appearing in those clauses contain at most $(1+c)\cdot i$ different variables.
Since clauses contain variables without repetition, it holds that $P_i$ is dominated by the probability to draw at most $(1+c)\cdot i$ different variables when drawing $k\cdot i$ Boolean variables independently at random.
Now imagine sampling these $k\cdot i$ variables in some arbitrary, but fixed order.
It holds that the probability to draw a new variable is at most $1$, while the probability to draw an old variable is at most the probability to draw one of the $(1+c)\cdot i$ variables of maximum probability.
As before, the sum of these probabilities is denoted by $F( (1+c)\cdot i)$.
This gives us
\[P_i\le \binom{m}{i}\cdot\binom{k\cdot i}{(1+c)\cdot i}\cdot 1^{(1+c)\cdot i}\cdot F( (1+c)\cdot i)^{k\cdot i-(1+c)\cdot i}.\]
Note that this expression also captures the case that we draw fewer than $(1+c)\cdot i$ different variables, since the probability to draw a new variable is bounded by one and thus also captures the probability that this new variable is in fact an old one.
In the case of a power-law distribution, we have
\[F( (1+c)\cdot i)\sim \left(\frac{(1+c)\cdot i}{n}\right)^{\frac{\beta-2}{\beta-1}}\]
due to \lemref{lem:powerlaw} and thus
\begin{align*}
P_i	&\le \binom{m}{i}\cdot\binom{k\cdot i}{(1+c)\cdot i}\cdot \left(\frac{(1+c)\cdot i}{n}\right)^{(k-(1+c))\cdot\frac{\beta-2}{\beta-1}\cdot i}\\
		& \le \left(\frac{e\cdot m}{i}\right)^i \cdot \left(\frac{e\cdot k}{1+c}\right)^{(1+c)\cdot i}\cdot\left(\frac{(1+c)\cdot i}{n}\right)^{(k-(1+c))\cdot\frac{\beta-2}{\beta-1}\cdot i}\\
		& = \kappa(c,\beta,k)^i \cdot \Delta^i \left(\frac{i}{n}\right)^{i((k-(1+c))\cdot\frac{\beta-2}{\beta-1}-1)}\\
		& = \kappa(c,\beta,k)^i \cdot \Delta^i \left(\frac{i}{n}\right)^{i\cdot \eps}
\end{align*}
for some constant $\kappa(c,\beta,k)>0$, $m=\Delta\cdot n$, and $c=(k-1)-(1+\eps)\cdot\frac{\beta-1}{\beta-2}$.

Summing over all $i\ge 1$ now yields
\[\Pro{\E} \le \sum_{i=1}^r \kappa(c,\beta,k)^i \cdot \Delta^i\cdot \left(\frac{i}{n}\right)^{i\cdot\eps}\]
We split this sum into two parts, the first part from $i=1$ to $\left\lfloor \eps\cdot \log n\right\rfloor$ and the second part from $\left\lceil \eps\cdot \log n\right\rceil$ to $r$.
For the first part we get
\begin{align*}
	\sum_{i=1}^{\left\lfloor \eps\cdot \log n\right\rfloor} \kappa(c,\beta,k)^i \cdot\Delta^i\cdot \left(\frac{i}{n}\right)^{i\cdot\eps}
	& \le \sum_{i=1}^{\left\lfloor \eps\cdot \log n\right\rfloor} \kappa(c,\beta,k)^i \cdot\Delta^i\cdot \left(\frac{\eps\cdot \log n}{n}\right)^{i\cdot\eps}\\
	& \le 2\cdot \kappa(c,\beta,k)\cdot\Delta\cdot \left(\frac{\eps\cdot \log n}{n}\right)^{\eps}\\
	& \in\Oh\left(\Delta\left(\frac{\log n}{n}\right)^{\eps}\right),
\end{align*}
which holds, since $\sum_{i=1}^m \alpha^i\le 2\cdot \alpha$ for all $m\ge1$ and $\alpha<\frac12$. 
This holds for big enough values of $n$ and for $\Delta\in o(n^\eps/\log^\eps n)$.
For the second part we get
\begin{equation*}
	\sum_{i=\left\lceil\eps\cdot \log n\right\rceil}^r \kappa(c,\beta,k)^i \cdot \Delta^i\cdot \left(\frac{i}{n}\right)^{i\cdot\eps}\\
	\le \sum_{i=\left\lceil \eps\cdot \log n\right\rceil}^r 2^{-i}
	\in\Oh\left(\left(\frac{1}{n}\right)^{\eps}\right),
\end{equation*}
which holds if we choose 
\[r\in\Oh\left(n\cdot\Delta^{-1/\eps}\right)\]
small enough so that $\Delta\cdot\left(\frac{r}{n}\right)^{\eps}<\frac{1}{2\cdot\kappa(c,\beta,k)}$.
\end{proof}

This notion of bipartite expansion is connected to the resolution width of a formula.
The following corollary, implicitly stated by \citet{shortProofs}, formalizes this connection.

\begin{corollary}[\cite{shortProofs}]\label{cor:exp-res}
Let $k\ge3$ integer and constant, let $\eps>0$ constant, and let $\Phi$ be an unsatisfiable Boolean formula in $k$-CNF.
If there is a constant $\eps>0$ such that $G(\Phi)$ is a $\left(r,\frac{k+\eps}{2}-1\right)$-bipartite expander, then $\Phi$ has resolution width at least $\Omega(r)$.
\end{corollary}
\begin{proof}
Due to the definition of bipartite expansion, $\frac{k+\eps}{2}>1$ ensures the first condition of \thmref{thm:resolution}.
We will show that the second condition is fulfilled as well.
Let $G(\Phi)=(C,V,E)$ and let $C'\subseteq C$ with $\frac13 r\le|C'|\le\frac23 r$.
Let $\delta C'$ denote the set of unique variables from $C'$, \ie $\delta C'=\left\{v\in N(C')\mid |N(v)\cap C'|=1\right\}$.
As Ben-Sasson and Widgerson state in~\cite[proof of Theorem~6.5]{shortProofs} it holds that:
\[|N(C')| - |\delta C'| \le (k\cdot |C'|-|\delta C'|)/2,\]
which implies
\[|\delta C'|\ge 2|N(C')|-k\cdot |C'|\ge \eps \cdot |C'|\]
due to the $\left(r,\frac{k+\eps}{2}-1\right)$-bipartite expansion.
These two properties imply a resolution width of $\Omega(r)$.
\end{proof}

This result on the bipartite expansion of power-law random $k$-SAT allows us to derive the following corollary on resolution width.
Again, we require the clause-variable ratio $\Delta$ to be high enough for instances to be unsatisfiable with at least constant probability.

%
%
%
\begin{corollary}\label{cor:expansion-resolution}
Let $\Phi$ be an unsatisfiable random power-law $k$-SAT formula with $n$ variables, $m\in\Omega(n)$ clauses, $k\ge3$, and power-law exponent $\beta>\frac{2k-2}{k-2}$.
Let $\Delta=m/n$ be large enough so that $\Phi$ is unsatisfiable at least with constant probability.
For $0<\eps<\frac{k}{2}\cdot\frac{\beta-2}{\beta-1}-1$ constant and $\Delta\in o(n^{\eps}/\log^{\eps}n)$ it holds \aas that $\Phi$ has resolution width $w\in\Omega(n\cdot\Delta^{-1/\eps})$.
\end{corollary}
\begin{proof}
Due to \lemref{lem:expansion} $G(\Phi)$ is a $(\Omega(n\cdot\Delta^{-1/\eps}),c)$-bipartite expander for $c=(k-1)-(1+\eps)\frac{\beta-1}{\beta-2}$.
With $\beta>\frac{2k-2}{k-2}$, it holds that we can choose an $\eps>0$ so that $c>\frac{k}{2}-1$.
This means, the requirement of Corollary~\ref{cor:exp-res} is fulfilled and implies the statement.
\end{proof}

Together with Corollary~\ref{cor:direct} the former corollary implies \thmref{thm:mainone}.

\mainone

Additionally, \citet{BenGalesi} state a theorem that directly connects bipartite expansion and tree-like resolution size.
An application of this theorem yields a slightly better bound on tree-like resolution size than the ones derived from resolution width.

\begin{theorem}[\cite{BenGalesi}] \label{thm:treelike}
Let $\Phi$ be an unsatisfiable CNF and let $G\left(\Phi\right)=\left(U\cup V, E\right)$ be the clause-variable incidence graph of $\Phi$.
If $G\left(\Phi\right)$ is a $(r,c)$-bipartite expander then $\Phi$ has resolution clause space of at least $\frac{c\cdot r}{2+c}$ and tree-like resolution size of at least $\exp\left(\Omega(\frac{c\cdot r}{2+c})\right)$.
\end{theorem}
\begin{proof}
\cite[Theorems~4.2 and~3.3]{BenGalesi} state together that any bipartite graph that is an $(r,c)$-bipartite expander has a resolution clause space of at least $\frac{c\cdot r}{2+c}$.
Thus, with~\cite[ Theorem~1.6]{spaceBounds}, it holds that the resolution size for formulas whose clause-variable incidence graph is an $(r,c)$-bipartite expander, is at least $\exp\left(\frac{c\cdot r}{2+c}\right)$.
\end{proof}

This leads to the following corollary, which already asserts exponential tree-like resolution size for constant clause-variable ratios at $\beta>\frac{2k-3}{k-2}$.

\begin{corollary}\label{cor:treelike-size}
Let $\Phi$ be an unsatisfiable random power-law $k$-SAT formula with
  $n$ variables, $m = \Omega(n)$ clauses, $k \ge
  3$, and power-law exponent $\beta>\frac{2k-3}{k-2}$.  Let $\Delta = m /
  n$ be large enough so
  that~$\Phi$ is unsatisfiable at least with constant probability.  
	For $0<\eps<(k-1)\cdot\frac{\beta-2}{\beta-1}-1$ constant and $\Delta\in o((n/\log n)^{\eps})$, it holds that $\Phi$ has tree-like resolution size $\exp(\Omega(n\cdot\Delta^{-1/\eps}))$.
\end{corollary}
\begin{proof}
Using \lemref{lem:expansion} we see that for $\beta>\frac{2k-3}{k-2}$ the clause-variable incidence graph \aas is a $(\Theta(n\cdot\Delta^{-1/\eps}),c)$-bipartite expander for some constant $c>0$.
Thus, \thmref{thm:treelike} implies the statement.
\end{proof}

\section{The Complexity of Voronoi Diagrams}
\label{sec:compl-voron-diagr}

We first show quadratic lower bounds on the complexity (number of
non-empty regions) of order-$k$ Voronoi diagrams that already hold in
rather basic settings.  Afterwards, we consider random point sets and
prove a linear upper bound.

\subsection{Worst-Case Lower Bounds}
\label{sec:lower-bound-voronoi}

In this section, we show worst-case lower bounds on the number of
non-empty regions of higher-order Voronoi diagrams.  As already
mentioned in Section~\ref{sec:results-voronoi-diagrams}, our lower
bounds are based on previously known lower bounds on the number of
vertices of Voronoi diagrams, in conjunction with a new theorem connecting the
number of vertices with the number of regions in higher orders.  This
theorem relies on the fact that there are not too many different
points with equal distance to a set of $d + 1$ sites in
$d$-dimensional space.  For the unweighted case and for
$\p \not= \infty$, the result in the next lemma was shown by
L{\^e}~\cite{l-vdlrd-96}.  We extend it to weighted sites and
$\p = \infty$, following along the lines of L{\^e}'s
proof~\cite{l-vdlrd-96} (at least for $\p \not= \infty$):
  \begin{inparaenum}[(i)]
  \item Observe that the points with equal distance to the $d + 1$
    sites is the set of solutions to a system of polynomial equations.
  \item Show that the so-called \emph{additive complexity} of these
    polynomial equations is bounded by a constant only depending on
    $d$.
  \item Apply \cite[Proposition~3]{l-vdlrd-96}, giving an upper
    bound on the number of solutions to a system of equations that
    only depends on $d$ and on the additive complexities of the
    equations.
  \end{inparaenum}

\begin{lemma}
  \label{lem:number-equidistant-points}
  Let $A$ be a set of $d + 1$ weighted sites in general
  position\footnote{For a formal definition what \emph{general
      position} means in this context, see \cite{l-vdlrd-96}.  As
    usual, the configurations excluded by the assumption of general
    position have measure~0.} in $\mathbb R^d$ equipped with a
  $\p$-norm.  Then, the number of points with equal weighted distance
  to all sites in $A$ only depends on $d$.
\end{lemma}
\begin{proof}
  Assume $\p \not= \infty$, and let $\pnt{s}_0, \dots, \pnt{s}_d$ be
  $d + 1$ sites with normalized weights $\omega_0, \dots, \omega_d$.
  Recall that the weighted distance between $\pnt{s}_i$ and a point
  $\pnt{p}$ is $\dist{\pnt{s}_i}{\pnt{p}}/\omega_i$.  Thus, $\pnt{p}$
  has the same distance to all $d + 1$ sites if, for all $i \in [d]$,
  it satisfies
  \begin{equation}
    \label{eq:polynomial-equations}
    \frac{\dist{\pnt{s}_0}{\pnt{p}}}{\omega_0} -
    \frac{\dist{\pnt{s}_i}{\pnt{p}}}{\omega_i} = 0.
  \end{equation}
  We note that this polynomial has the same form in the unweighted
  case \cite[Equation~10]{l-vdlrd-96}, except we have the additional
  factors $1/\omega_0$ and $1/\omega_i$.

  Concerning (ii), it thus suffices to note that these additional
  factors do not significantly increase the so-called additive
  complexity.  We do not fully define the additive complexity here,
  but rather cite the properties crucial for this proof.  The additive
  complexity $L_+(P)$ of a polynomial $P$ is defined to be $0$ if $P$
  is a monomial.  Moreover, by \cite[Lemma~4]{l-vdlrd-96}, it holds
  that
  \begin{align*}
    L_+(P_1 + \dots + P_n) &\le n - 1 + L_+(P_1) + \dots + L_+(P_n), \\
    L_+(P^m) &\le L_+(P)\text{, for any } m \in \mathbb N, \text{ and}\\
    L_+(PQ) &\le L_+(P) + L_+(Q),
  \end{align*}
  where all $P_i$, $P$, and $Q$ are polynomials.  With this, it is
  easy to see that the additive complexity of the polynomial in
  Equation~\eqref{eq:polynomial-equations} is bounded by a constant
  only depending on $d$.  In fact, the last bound,
  $L_+(PQ) \le L_+(P) + L_+(Q)$ in conjunction with the property that
  constants are monomials with additive complexity~$0$, makes it so
  that the additional constant factors $\omega_0$ and $\omega_i$ do
  not increase the additive complexity at all.  Thus, the additive
  complexity is bounded by $4d - 1$~\cite[Lemma~5]{l-vdlrd-96}.

  Finally, applying~\cite[Proposition~3]{l-vdlrd-96} directly yields
  the claim, which concludes the proof for $\p \not= \infty$.

  For $\p = \infty$, we cannot use the same argument, as
  Equation~\eqref{eq:polynomial-equations} is not polynomial:
  $\dist{\pnt{s_i}}{\pnt{p}}$ involves the maximum over all
  coordinates.  However, for each $s_i$, there are only $d$
  possibilities to which coordinate the maximum is evaluated, leading
  to $d^{d + 1}$ combinations.  For each of these combinations, we
  consider its own system of equations.  Denote the resulting set of
  systems of equations with $\mathcal E$.  Clearly, every solution for
  the system of equations in~\eqref{eq:polynomial-equations} is a
  solution to at least one system in $\mathcal E$.  Thus, the number
  of solutions to~\eqref{eq:polynomial-equations} is bounded by the
  total number of solutions to systems in~$\mathcal E$.  Clearly, with
  the same argument as above, the number of solutions to each system
  of equations in $\mathcal E$ is bounded by a constant only depending
  on $d$.  As $\mathcal E$ contains only $d^{d + 1}$ systems, this
  bounds the number of solutions to~\eqref{eq:polynomial-equations} by
  a constant only depending on $d$.
\end{proof}

With this, we can now prove the theorem establishing the connection
between vertices and non-empty regions.

\VoronoiComplVerticesRegions
\begin{proof}
  We first show that a vertex of the order-$k$ Voronoi diagram is an
  interior point of a non-empty region of the order-$(k + d)$ Voronoi
  diagram.  Afterwards, we show that only a constant number of
  different vertices can end up in the same region.

  Let $\pnt{p} \in R^d$ be a vertex of the order-$k$ Voronoi diagram.
  Then $\pnt{p}$ has equal weighted distance to exactly $d + 1$ sites
  (the sites are in general position).  Let
  $\{\pnt{s}_1, \dots, \pnt{s}_{d + 1}\} = A \subseteq S$ be these
  sites and let $P$ be the \emph{$\eps$-environment} of $\pnt{p}$,
  i.e., a ball with sufficiently small radius $\eps$ centered at
  $\pnt{p}$.  For a point $\pnt{p}' \in P$, sort all sites in $S$ by
  weighted distance from $\pnt{p}'$.  Then all sites in $A$ appear
  consecutive in this order.  Moreover, we obtain almost the same
  order of $S$ for every $\pnt{p}' \in P$.  The only difference is
  that the sites of $A$ might be reordered.  Also, as $\pnt{p}$ is a
  vertex of the order-$k$ Voronoi diagram, at least one site from $A$
  belongs to the $k$ sites with smallest weighted distance to
  $\pnt{p}$.  It follows that the first $k + d$ sites in this order
  completely include all sites from $A$.  Thus, the $k + d$ closest
  sites are the same for all points in the $\eps$-environment $P$
  around $\pnt{p}$; let $B$ be the set of these sites.  It follows
  that $B$ has non-empty Voronoi region in the order-$(k + d)$ Voronoi
  diagram as this region has $\pnt{p}$ in its interior.

  It remains to show that only a constant number of vertices of the
  order-$k$ Voronoi diagram can be contained in the same region of the
  order-$(k + d)$ Voronoi diagram, i.e., the order-$(k + d)$ region
  belonging to $B$ includes only a constant number of order-$k$
  vertices.  As stated above, every order-$k$ vertex belongs to a
  subset $A \subseteq B$ with $|A| = d + 1$.  There are only
  $\binom{|B|}{|A|} \le \binom{k + d}{d + 1}$ such subsets $A$,
  which is constant for constant $k$ and $d$.  Moreover, every fixed
  subset $A$ of $d + 1$ sites is responsible for only a constant
  number of vertices due to Lemma~\ref{lem:number-equidistant-points}.
  Thus, only a constant number of order-$k$ vertices end up in the
  same order-$k + d$ region, which concludes the proof.
\end{proof}

Theorem~\ref{thm:complexity-vertices-vs-regions} transfers some known
lower bounds on the number of vertices of Voronoi diagrams to lower
bounds on the number of non-empty regions of order-$k$ Voronoi
diagrams.  In particular, we get the following corollaries.

\begin{corollary}
  In the worst case, the order-$4$ Voronoi diagram of $n$ (unweighted)
  sites in $3$-dimensional Euclidean space has $\Omega(n^2)$ non-empty
  regions.
\end{corollary}
\begin{proof}
  In the worst case, the ordinary (order-$1$, unweighted) Voronoi
  diagram of $n$ sites in $3$-dimensional Euclidean space has
  $\Omega(n^2)$ vertices~\cite{k-cdvd-80,s-nfhdvd-87}.  Applying
  Theorem~\ref{thm:complexity-vertices-vs-regions} yields the claim.
\end{proof}

\begin{corollary}
  In the worst case, the order-$3$ Voronoi diagram of $n$ weighted
  sites in $2$-dimensional Euclidean space has $\Omega(n^2)$ non-empty
  regions.
\end{corollary}
\begin{proof}
  In the worst case, the order-$1$ Voronoi diagram of $n$ weighted
  sites in $2$-dimensional Euclidean space has $\Omega(n^2)$
  vertices~\cite{ae-oacwvdp-84}; also see
  Figure~\ref{fig:weighted-voronoi-complexity}.  Applying
  Theorem~\ref{thm:complexity-vertices-vs-regions} yields the claim.
\end{proof}

\subsection{Upper Bounds for Sites with Random Positions}
\label{sec:expect-numb-regions}

Let $S = \{\pnt{s}_1, \dots, \pnt{s}_n\} \subseteq \mathbb T^d$ be $n$
randomly positioned sites with weights $w_1, \dots, w_n$.  In the
following, we bound the complexity of the weighted order-$k$ Voronoi
diagram in terms of non-empty regions.  Recall from
Section~\ref{sec:preliminaries} that the torus $\mathbb T^d$ is the hypercube
$[0, 1]^d$ that wraps around in every dimension in the sense that
opposite sides are identified.  However, the following arguments do
not require this property.  Thus, the exact same results hold for
Voronoi diagrams in hypercubes.

For the normalized weights $\omega_1, \dots, \omega_n$, recall from
Section~\ref{sec:preliminaries}, that the point $\pnt{p} \in \mathbb T^d$
belongs to the Voronoi region corresponding to $A \subseteq S$ with
$|A| = k$ if there exists a radius $r$ such that
$\dist{\pnt{p}}{\pnt{s}_i} \le \omega_i r$ if $\pnt{s}_i \in A$ and
$\dist{\pnt{p}}{\pnt{s}_i} > \omega_i r$ if $\pnt{s}_i \notin A$.
Thus, $A$ has non-empty order-$k$ Voronoi region if and only if there
exists such a point $\pnt{p}$.  Our goal in the following is to bound
the probability for its existence.

Our general approach to achieve such a bound is the following.  The
condition $\dist{\pnt{p}}{\pnt{s}_i} \le \omega_i r$ for
$\pnt{s}_i \in A$ basically tells us the sites in $A$ are either close
together or that $r$ has to be large.  In contrast to that, the
condition $\dist{\pnt{p}}{\pnt{s}_i} > \omega_i r$ for
$\pnt{s}_i \notin A$ tells us that many sites (namely all $n - k$
sites in $S \setminus A$) have to lie sufficiently far away from
$\pnt{p}$, which is unlikely if $r$ is large.  How unlikely this is of
course depends on $r$ and thus on how close the sites in $A$ lie
together.  Therefore, to follow this approach, we first condition on
how close the sites in $A$ lie together.

To formalize this, consider a size-$k$ subset $A \subseteq S$ and
assume without loss of generality that
$A = \{\pnt{s}_1, \dots, \pnt{s}_k\}$.  The site in $A$ with the
lowest weight, without loss of generality $\pnt{s}_1$, will play a
special role.  We define the random variable $R_A$ to be
\begin{equation}
  \label{eq:def-R-A}
  R_A = \max_{i \in [k]}\frac{\dist{\pnt{s}_1}{\pnt{s}_i}}{\omega_1 +
    \omega_i}.
\end{equation}
The intuition behind the definition of $R_A$ is the following.  The
weighted center between $\pnt{s}_1$ and $\pnt{s}_i$ is the point
$\pnt{p}$ on the line between them such that
$\dist{\pnt{s}_1}{\pnt{p}} = \omega_1 r$ and
$\dist{\pnt{s}_i}{\pnt{p}} = \omega_i r$ for a radius
$r \in \mathbb R$.  Then $R_A$ is the maximum value for $r$ over
$i \in [k]$.  In the unweighted setting, $R_A$ is just half the
maximum distance between $\pnt{s}_1$ and any other site $\pnt{s}_i$.
In a sense, $R_A$ describes how close the sites in $A$ lie together.
Thus, it provides a lower bound on $r$.

Based on $R_A$, we slightly relax the condition on $A$ having
non-empty Voronoi region.  We call $A$ \emph{relevant} if there exists
a point $\pnt{p} \in \mathbb T^d$ and a radius $r \ge R_A$ such that
$\dist{\pnt{s}_1}{\pnt{p}} \le \omega_1 r$ and
$\dist{\pnt{s}_i}{\pnt{p}} > \omega_i r$ for $i > k$.
The following lemma states that being relevant is in fact a weaker
condition than having non-empty order-$k$ Voronoi region.  Thus,
bounding the probability that a set is relevant from above also bounds
the probability for a non-empty Voronoi region from above.

\begin{lemma}
  \label{lem:non-empty-region-implies-relevance}
  A subset of $k$ sites that has a non-empty order-$k$ Voronoi region
  is relevant.
\end{lemma}
\begin{proof}
  Assume $A = \{\pnt{s}_1, \dots, \pnt{s}_k\}$ has a non-empty
  order-$k$ Voronoi region.  Then there exists a point $\pnt{p}$ and a
  radius $r$ such that $\dist{\pnt{s}_i}{\pnt{p}} \le \omega_i r$ if
  and only if $i \le k$.  Thus,
  $\dist{\pnt{s}_1}{\pnt{p}} \le \omega_1 r$ and
  $\dist{\pnt{s}_i}{\pnt{p}} > \omega_i r$ for $i > k$ clearly holds,
  and it remains to show $r \ge R_A$.  From
  $\dist{\pnt{s}_i}{\pnt{p}} \le \omega_i r$ for $i \in [k]$ it
  follows that
  $\dist{\pnt{s}_1}{\pnt{p}} + \dist{\pnt{s}_i}{\pnt{p}} \le \omega_1
  r + \omega_i r$ holds for any $i \in [k]$.  Thus, by rearranging and
  applying the triangle inequality, we obtain
  $r \ge (\dist{\pnt{s}_1}{\pnt{p}} +
  \dist{\pnt{s}_i}{\pnt{p}})/(\omega_1 + \omega_i) \ge
  \dist{\pnt{s}_1}{\pnt{s}_i} / (\omega_1 + \omega_i)$.  This
  immediately yields $r \ge R_A$.
\end{proof}

Now we proceed to bound the probability that a set $A$ is relevant.
The following lemma bounds this probability conditioned on the random
variable~$R_A$.  At its core, we have to bound the probability of the
event $\dist{\pnt{s}_i}{\pnt{p}} > \omega_ir$ for
$\pnt{s}_i \notin A$.  For a fixed point $\pnt{p}$ and a fixed radius
$r$, this is rather easy.  Thus, most of the proof is concerned with
eliminating the existential quantifiers for $\pnt{p}$ and $r$.

\begin{lemma}
  \label{lem:relevance-probability-cond-R}
  For constants $c_1$ and $c_2$ depending only on $d$ and $\p$, it
  holds that
  \begin{equation*}
    \Pro{A \text{ is relevant} \mid R_A}
    \le  c_1 \min_{\pnt{s}_i \in A}\{w_i\}
    \exp\left(-c_2 R_A^d\sum_{\pnt{s}_i \notin A} w_i\right).
  \end{equation*}
\end{lemma}
\begin{proof}
  As before, we assume that $A = \{\pnt{s}_1, \dots, \pnt{s}_k\}$ and
  that $\pnt{s}_1$ has minimum weight among sites in $A$, i.e.,
  $\min_{\pnt{s}_i \in A}\{w_i\} = w_1$.  By definition, $A$ is
  relevant conditioned on $R_A$, if and only if there exists a radius
  $r \ge R_A$ and point $\pnt{p} \in \mathbb T^d$ such that
  $\dist{\pnt{s}_1}{\pnt{p}} \le \omega_1 r$ and
  $\dist{\pnt{s}_i}{\pnt{p}} > \omega_i r$ for $i > k$, i.e., formally
  we have
  \begin{equation}
    \label{eq:a-relevant}
    \exists r \ge R_A\; \exists \pnt{p} \in \mathbb T^d\; \forall i > k
    \colon \dist{\pnt{s}_1}{\pnt{p}} \le \omega_1 r \wedge
    \dist{\pnt{s}_i}{\pnt{p}} > \omega_i r.
  \end{equation}
  The core difficulties of bounding the probability for this event are
  the existential quantifiers that quantify over the continuous
  variables $r$ and $\pnt{p}$.  In both cases, we resolve this by
  using an appropriate discretization, for which we then apply the
  union bound.

  We get rid of the existential quantifier for $r$ by dividing the
  interval $[R_A, \infty)$, which covers the domain of $r$, into
  pieces of length at most $R_A$.  More formally, we split the event
  $\exists r \ge R_A$ with the desired property into the disjoint
  events $\exists r \in [j R_A, (j + 1) R_A)$ for $j \in \mathbb N^+$.
  For a fixed $j$, $r \ge j R_A$ and
  $\dist{\pnt{s}_i}{\pnt{p}} > \omega_i r$ implies
  $\dist{\pnt{s}_i}{\pnt{p}} > \omega_i j R_A$.  Moreover,
  $r \le (j + 1) R_A$ and $\dist{\pnt{s}_1}{\pnt{p}} \le \omega_1 r$
  implies $\dist{\pnt{s}_1}{\pnt{p}} \le \omega_1 (j + 1) R_A$.  Note
  that this completely eliminates the variable $r$ from the event,
  which lets us drop the existential quantifier for $r$.  Thus, the
  event in Equation~\eqref{eq:a-relevant} implies
  \begin{equation}
    \label{eq:a-relevant-eliminated-exists-r}
    \exists j \in \mathbb N^+ \;\exists \pnt{p} \in \mathbb T^d
    \;\forall i > k \colon \dist{\pnt{s}_1}{\pnt{p}} \le \omega_1 (j +
    1) R_A \wedge \dist{\pnt{s}_i}{\pnt{p}} > \omega_i j R_A.
  \end{equation}
  Note that the new existential quantifier for $j$ is not an issue: as
  $j$ is discrete, we can simply use the union bound and sum over the
  probabilities we obtain for the different values of $j$.  We will
  later see that this sum is dominated by the first term corresponding
  to $j = 1$.

  To deal with the existential quantifier for $\pnt{p}$, assume
  $j \in \mathbb N^+$ to be a fixed number.  First note that
  $\dist{\pnt{s}_1}{\pnt{p}} \le \omega_1 (j + 1) R_A$ implies that
  $\pnt{p}$ lies somewhat close to $\pnt{s}_1$.  We discretize the
  space around $\pnt{s}_1$ using a grid such that the point $\pnt{p}$
  is guaranteed to lie inside a grid cell.  By choosing the distance
  between neighboring grid vertices sufficiently small, we guarantee
  that $\pnt{p}$ lies close to a grid vertex.  Then, instead of
  considering $\pnt{p}$ itself, we deal with its closest grid vertex.
  To define the grid formally, let
  $\omega_{\min} = \min_{i = k + 1}^n \{\omega_i\}$ be the minimum
  weight of sites not in $A$ and let
  $x = \omega_{\min} j R_A / \sqrt[\p]{d}$ ($x$ will be the width of
  our grid cells).  To simplify notation, assume that $\pnt{s}_1$ is
  the origin.  Otherwise, we can simply translate the grid defined in
  the following to be centered at $\pnt{s}_1$ to obtain the same
  result.  Let
  $\Gamma = \{\ell x \mid \ell \in \mathbb Z \wedge |(\ell - 1) x| \le
  \omega_1 (j + 1) R_A\}$ be the set of all multiples of $x$ that are
  not too large.  We use the grid defined by the Cartesian product
  $\Gamma^d$.  Then the following three properties of $\Gamma^d$ are
  easy to verify.
  \begin{enumerate}[(i)]
  \item \label{item:grid-p-lies-in-cell} A point $\pnt{p}$ with
    $\dist{\pnt{s}_1}{\pnt{p}} \le \omega_1 (j + 1) R_A$ lies in a
    grid cell.
  \item \label{item:grid-max-dist-from-vertex} The maximum distance
    between a point in a grid cell and its closest grid vertex is
    $\sqrt[\p]{d} x / 2 = \omega_{\min} j R_A / 2$.
  \item \label{item:grid-numer-of-vertices}$\Gamma^d$ has at most
    $c_1' (\omega_1 j R_A / x)^d = c_1 (\omega_1/\omega_{\min})^d \le
    c_1 \omega_1^d$ vertices for constants $c_1$ and $c_1'$ only
    depending on $d$ and $\p$.
  \end{enumerate}

  Going back to the event in
  Equation~\eqref{eq:a-relevant-eliminated-exists-r}, let $\pnt{p}$ be
  a point with $\dist{\pnt{s}_1}{\pnt{p}} \le \omega_1(j + 1)R_A$ and
  $\dist{\pnt{s}_i}{\pnt{p}} > \omega_i j R_A$ (for all $i > k$).  By
  the first inequality and Property~\ref{item:grid-p-lies-in-cell},
  $\pnt{p}$ lies in a grid cell of $\Gamma^d$.  Let
  $\pnt{p}' \in \Gamma^d$ be the grid vertex with minimum distance to
  $\pnt{p}$.  Then, by Property~\ref{item:grid-max-dist-from-vertex},
  $\dist{\pnt{p}}{\pnt{p}'} \le \omega_{\min} j R_A / 2$.  Thus, using
  the triangle inequality and
  $\dist{\pnt{s}_i}{\pnt{p}} > \omega_i j R_A$, we obtain
  \begin{equation*}
    \dist{\pnt{s}_i}{\pnt{p}'} \ge \dist{\pnt{s}_i}{\pnt{p}} -
    \dist{\pnt{p}}{\pnt{p}'} > \omega_i j R_A - \frac{\omega_{\min}j
      R_A}{2} \ge \frac{\omega_i j R_A}{2}.
  \end{equation*}
  It follows that the event in
  Equation~\eqref{eq:a-relevant-eliminated-exists-r} implies
  \begin{equation*}
    \exists j \in \mathbb N^+ \;\exists \pnt{p}' \in \Gamma^d \;\forall
    i > k \colon \dist{\pnt{s}_i}{\pnt{p}'} > \frac{\omega_i j
      R_A}{2}.
  \end{equation*}
  For this event, we can now bound the probability.  First note that
  $\dist{\pnt{s}_i}{\pnt{p}'} > \omega_i j R_A / 2$ implies that the
  ball of radius $\omega_i j R_A / 2$ around $\pnt{p}'$ does not
  contain $\pnt{s}_i$.  By
  Lemma~\ref{lem:volume-intersection-ball-cube}, the volume of this
  ball intersected with $[-0.5, 0.5]^d$ is
  $\min\{1, c_2 (\omega_i j R_A)^d\}$ for a constant $c_2$ depending
  only on $d$ and $\p$.  As the $\pnt{s}_i$ are chosen independently
  and using that $1 - x \le \exp(-x)$ for $0 \le x \le 1$, we obtain
  \begin{align*}
    \Pro{\forall i > k \colon \dist{\pnt{s}_i}{\pnt{p}'} >
    \frac{\omega_i j R_A}{2}}
    &= \prod_{i = k + 1}^n \max\left\{0, 1 - c_2 \left(\omega_i j
      R_A\right)^d\right\}\\
    &\le\prod_{i = k + 1}^n \exp\left(-c_2 \left(\omega_i j
      R_A\right)^d\right)\\
    &=\exp\left(-c_2 j^d R_A^d\sum_{i = k + 1}^n \omega_i^d\right).
  \end{align*}
  We resolve the two existential quantifiers for $j$ and $\pnt{p}'$
  using the union bound.  Recall from
  Property~\ref{item:grid-numer-of-vertices} that the grid $\Gamma^d$
  contains only $c_1\omega_1^d$ vertices.  Using that
  $\omega_i = w_i^{1/d}$, we obtain
  \begin{align*}
    \Pro{A \text{ is relevant} \mid R_A}
    &\le \Pro{\exists j \in \mathbb N^+ \;\exists \pnt{p}' \in \Gamma^d
      \;\forall i > k \colon \dist{\pnt{s}_i}{\pnt{p}'} >
      \frac{\omega_i j R_A}{2}}\\
    &\le \sum_{j = 1}^{\infty} c_1 w_1 \exp\left(-c_2 j^d
      R_A^d\sum_{i = k + 1}^n w_i\right).
  \end{align*}
  To conclude the proof, it remains to show that the sum over $j$ is
  dominated by the first term corresponding to $j = 1$.  For this,
  note that
  \begin{align*}
    \sum_{j = 1}^{\infty} \exp\left(-xj^d\right)
    &= \exp(-x) \cdot \sum_{j = 1}^{\infty}
      \frac{\exp\left(-xj^d\right)}{\exp(-x)}\\
    &= \exp(-x) \cdot \sum_{j = 1}^{\infty} \exp\left(-x\left(j^d -
      1\right)\right)\\
    &\le \exp(-x) \cdot \sum_{j = 1}^{\infty}
      \left(\exp(-x)\right)^{j - 1}
  \end{align*}
  As $x$ is positive in our case, the sum is bounded by a constant due
  to the convergence of the geometric series.  This concludes the proof.
\end{proof}

Now that we know the probability that $A \subseteq S$ is relevant
conditioned on $R_A$, we want to understand how $R_A$ is distributed.
The following lemma gives an upper bound on its density function.

\begin{lemma}
  \label{lem:R-A-density-function}
  There exists a constant $c$ depending only on $k$, $d$, and $\p$,
  such that the density function $f_{R_A}(x)$ of the random variable
  $R_A$ satisfies
  \begin{equation*}
    f_{R_A}(x)
    \le c x^{dk - d - 1} \frac{1}{\min\limits_{\pnt{s}_i \in A}\{w_i\}}
    \prod_{\pnt{s}_i \in A} w_i.
  \end{equation*}
\end{lemma}
\begin{proof}
  The density function $f_{R_A}(x)$ is the derivative of the
  distribution function $F_{R_A}(x) = \Pro{R_A \le x}$.  Thus, we have
  to upper bound the slope of $\Pro{R_A \le x}$.  As before, we assume
  that $A = \{\pnt{s}_1, \dots, \pnt{s}_k\}$ and that $\pnt{s}_1$ has
  minimum weight among sites in $A$, i.e.,
  $\min_{\pnt{s}_i \in A}\{w_i\} = w_1$.  Recall the definition of
  $R_A$ in Equation~\eqref{eq:def-R-A}.  It follows directly that
  $R_A \le x$ if and only if
  $\dist{\pnt{s}_1}{\pnt{s}_i} / (\omega_1 + \omega_i) \le x$ for all
  $i \in [k]$.  Note that this clearly holds for $i = 1$.  For greater
  $i$, this is the case if and only if $\pnt{s}_i$ lies in the ball
  $B_{\pnt{s}_1}((\omega_1 + \omega_i) x)$ of radius
  $(\omega_1 + \omega_i) x$ around $\pnt{s}_1$.  To simplify notation,
  we denote this ball with $B(x_i)$ in the following.  Note that the
  volume $\vol(B(x_i))$ is exactly the probability for $\pnt{s}_i$ to
  lie sufficiently close to $\pnt{s}_1$.  As the positions of the
  different sites $\pnt{s}_i$ are independent, we obtain
  \begin{equation*}
    F_{R_A}(x) = \Pro{R_A \le x} = \prod_{i = 2}^k \vol\left(B(x_i)\right).
  \end{equation*}
  To upper bound the derivative of this, we have to upper bound the
  growth of $\vol(B(x_i))$ depending on $x_i$.  For sufficiently small
  $x_i$, this volume is given by the volume of a ball in
  $\mathbb R^d$.  For larger $x_i$, due to the fact that our ground
  space\footnote{Again, this is true for the torus as well as for the
    Hypercube.} is bounded, the growth of this volume slows down.
  Thus, to get an upper bound on the derivative, we can simply use the
  volume of a ball in $\mathbb R^d$.  Thus, for appropriate constants
  $c_1$ and $c_2$ only depending on $d$ and $\p$, we obtain
  \begin{equation*}
    \frac{\dif}{\dif x} \vol\left(B(x_i)\right)
    \le \frac{\dif}{\dif x} c_1 \left((\omega_1 + \omega_i) x\right)^d
    \le \frac{\dif}{\dif x} c_2 \left(\omega_i x\right)^d
    = \frac{\dif}{\dif x} c_2 w_i x^d.
  \end{equation*}
  With this, it follows that
  \begin{equation*}
    f_{R_A}(x) = \frac{\dif}{\dif x}F_{R_A}(x) \le \frac{\dif}{\dif x}
    \prod_{i = 2}^k \left(c_2 w_i x^d\right), 
  \end{equation*}
  which immediately yields the claimed bound.
\end{proof}

By Lemma~\ref{lem:relevance-probability-cond-R}, we know the
probability for a set $A$ to be relevant conditioned on $R_A$ and by
Lemma~\ref{lem:R-A-density-function} we know how $R_A$ is distributed.
Based on this, we can bound the unconditional probability that $A$ is
relevant.

\begin{lemma}
  \label{lem:probability-for-relevance}
  Let $A \subseteq S$.  For a constant $c$ only depending on $k$, $d$,
  and $\p$, the probability that $A$ is relevant satisfies
  \begin{equation*}
    \Pro{A\text{ is relevant}} \le c\frac{\prod_{\pnt{s}_i \in A}
      w_i}{\left(\sum_{\pnt{s}_i \notin A} w_i\right)^{k-1}}.
  \end{equation*}
\end{lemma}
\begin{proof}
  Let $A \subseteq S$ and let $R_A$ be the random variable as defined
  before; see Equation~\eqref{eq:def-R-A}.  Note that
  $0 \le R_A \le \sqrt[\p]{d}$.  By the law of total probability, we
  have
  \begin{equation*}
    \Pro{A\text{ is relevant}} = \int_{0}^{\sqrt[\p]{d}} \Pro{A\text{ is
        relevant} \mid R_A = x}\cdot f_{R_A} (x) \dif x. 
  \end{equation*}
  Using Lemma~\ref{lem:relevance-probability-cond-R} and
  Lemma~\ref{lem:R-A-density-function}, we obtain
  \begin{equation*}
    \Pro{A\text{ is relevant} \mid R_A = x} \cdot f_{R_A} (x) \le c_1
    \prod_{\pnt{s}_i \in A} w_i \exp\left(-c_2 x^d\sum_{\pnt{s}_i
        \notin A} w_i\right) x^{dk - d - 1},
  \end{equation*}
  for constants $c_1$ and $c_2$ only depending on $k$, $d$, and $\p$.
  Ignoring the factors independent of $x$ for now, this expression has
  the form
  \begin{equation*}
    x^{\alpha d - 1} \exp\left(-\beta x^d\right) \text{ with } \alpha
    = k - 1 \text{ and } \beta = c_2 \sum_{\pnt{s}_i \notin A} w_i,
  \end{equation*}
  which lets us apply Lemma~\ref{lem:integral-gamma-function} to bound
  the integral.  We obtain
  \begin{align*}
    \Pro{A\text{ is relevant}}
    &=c_1 \prod_{\pnt{s}_i \in A} w_i \cdot \int_{0}^{\sqrt[\p]{d}}
      x^{\alpha d - 1} \exp\left(-\beta x^d\right) \dif x\\
    &\le c_1 \prod_{\pnt{s}_i \in A} w_i \cdot
      \frac{\Gamma\left(\alpha\right)}{\beta^\alpha d}.
  \end{align*}
  As $k$ is an integer, $\Gamma(\alpha) = \Gamma(k - 1) = (k - 2)!$,
  which is constant.  Thus, substituting $\alpha$ and $\beta$ by its
  corresponding values and aggregating all constant factors into $c$
  yields
  \begin{equation*}
    \Pro{A\text{ is relevant}} \le c\frac{\prod_{\pnt{s}_i \in A}
      w_i}{\left(\sum_{\pnt{s}_i \notin A} w_i\right)^{k-1}},
  \end{equation*}
  which is exactly the bound we wanted to prove.
\end{proof}

Having bound the probability that a specific subset of sites
$A \subseteq S$ of size $k$ is relevant, we can now bound the expected
total number of relevant subsets.  By
Lemma~\ref{lem:non-empty-region-implies-relevance}, this also bounds
the number of non-empty Voronoi regions.

\ComplexityRandomVoronoiDiagram
\begin{proof}
  For every subset $A \subseteq S$ with $|A| = k$, let $X_A$ be the
  indicator random variable that has value $1$ if and only if $A$ has
  non-empty order-$k$ Voronoi region.  Moreover, let $X$ be the sum of
  these random variables.  Note that $\EX{X}$ is exactly the quantity,
  we are interested in.  Using linearity of expectation, we obtain
  \begin{equation*}
    \EX{\text{number of regions}} = \EX{X} = \sum_{\substack{A
        \subseteq S\\ |A| = k}} \EX{X_A}
  \end{equation*}
  Due to Lemma~\ref{lem:non-empty-region-implies-relevance}, a subset
  $A$ with non-empty Voronoi region is also relevant.  Thus,
  $\EX{X_A} \le \Pro{A \text{ is relevant}}$ and
  Lemma~\ref{lem:probability-for-relevance} yields
  \begin{equation}
    \label{eq:expected-nr-of-non-empty-regions}
    \sum_{\substack{A
        \subseteq S\\ |A| = k}} \EX{X_A}
    \le \sum_{\substack{A \subseteq S\\ |A| = k}}
    c\frac{\prod_{\pnt{s}_i \in A} w_i}{\left(\sum_{\pnt{s}_i \notin
          A} w_i\right)^{k-1}}.
  \end{equation}
  For technical reasons, we assume $c$ to be the maximum of $1$ and
  the constant from Lemma~\ref{lem:probability-for-relevance}.  We
  continue by proving the following claim:
  \begin{equation}
    \label{eq:claim-expected-nr-of-non-empty-regions}
    \sum_{\substack{A \subseteq S\\ |A| = k}} \EX{X_A} \le
    4^{k^2} c W
  \end{equation}
  In addition to implying the theorem, this claim specifies a constant
  that comes on top of $c$, which is crucial for the rest of the
  proof.

  We first prove the claim for the situation, in which $W$ is not
  dominated by the highest $k$ weights.  Afterwards, we deal with the
  other somewhat special case.  More formally, let the weights
  $w_1, \dots, w_n$ be sorted increasingly and consider the case that
  $\sum_{i = 1}^{n - k} w_i \ge 4^{-k}W$, i.e., if we leave out the
  $k$ largest weights, we still have a significant portion of the
  total weight.  We can use this to estimate the denominator in
  Equation~\eqref{eq:expected-nr-of-non-empty-regions}:
  \begin{align*}
    \sum_{\substack{A \subseteq S\\ |A| = k}} c\frac{\prod_{\pnt{s}_i
    \in A} w_i}{\left(\sum_{\pnt{s}_i \notin A} w_i\right)^{k-1}}
    \le&
         \sum_{\substack{A \subseteq S\\ |A| = k}} c\frac{\prod_{\pnt{s}_i
    \in A} w_i}{\left(4^{-k}W\right)^{k-1}}\\
    =& 4^{k\,(k-1)} c \cdot \frac{\sum_{A \subseteq S, |A| =
       k}\prod_{\pnt{s}_i \in A} w_i}{W^{k-1}}.
  \end{align*}
  %
  %
  To bound the fraction by $W$, observe that the binomial theorem
  yields
  \begin{equation*}
    W^k = \left(\sum_{i = 1}^n w_i\right)^k \ge \sum_{\substack{A
        \subseteq S\\ |A| = k}} \prod_{\pnt{s}_i \in A} w_i,
  \end{equation*}
  as each summand on the on the right-hand side also appears on the
  left-hand side.  This proves the claim in
  Equation~\eqref{eq:claim-expected-nr-of-non-empty-regions} for the
  case $\sum_{i = 1}^{n - k} w_i \ge 4^{-k}W$.

  For $\sum_{i = 1}^{n - k} w_i < 4^{-k}W$, assume for contradiction
  that the claim in
  Equation~\eqref{eq:claim-expected-nr-of-non-empty-regions} does not
  hold for every set of $n$ weights.  Then there exists a minimum
  counterexample, i.e., a smallest number of $n$ weights such that
  the expected number of non-empty regions exceeds $4^{k^2} c W$.  We
  show that, based on this assumption, we can construct an even
  smaller counterexample; a contradiction.  First note that $n > 2k$
  for every counterexample, as there are fewer than $4^{k^2} c W$
  subsets otherwise (recall that $c \ge 1$).

  Now let $w_1, \dots, w_n$ be the minimum counterexample and again
  assume that the weights are ordered increasingly.  Moreover, fix the
  coordinates of the sites $s_1, \dots, s_n$ and consider two
  order-$k$ Voronoi diagrams: one on the set of all sites
  $S = \{s_1, \dots, s_n\}$, and the one on all but the $k$ heaviest
  sites $S' = \{s_1, \dots, s_{n - k}\}$ (note that this is well
  defined as $n > 2k$).  In the following, we call the former Voronoi
  diagram $\mathcal V$ and the latter $\mathcal V'$.  We define a
  mapping from the non-empty regions of $\mathcal V$ to non-empty
  regions of $\mathcal V'$.  Let $A \subseteq \{s_1, \dots, s_n\}$ be
  a subset of size $k$ with non-empty region in $\mathcal V$ and let
  $\pnt{p}$ be an arbitrary point in this region.  Moreover, let $A'$
  be the set of sites corresponding to the region of $\mathcal V'$
  containing $\pnt{p}$.  Then we map the region of $A$ to the region
  of $A'$.  Note that $A$ and $A'$ share all sites that have not been
  deleted: $A \cap A' = A \cap S'$. Thus, any site $A$ that is mapped 
	to $A'$ must satisfy $A\subseteq A'\cup (S\setminus S')$.
	This limits
  the number of different regions in $\mathcal V$ that are mapped to
  the same region of $\mathcal V'$ to at most $4^k$.  Thus, the number
  of regions in $\mathcal V'$ is at least $4^{-k}$ times the number of
  regions in $\mathcal V$.  As this holds for arbitrary coordinates,
  this also holds for the expected number of non-empty regions when
  choosing random coordinates.

  As we assumed $w_1, \dots, w_n$ to be a counterexample for
  Equation~\eqref{eq:claim-expected-nr-of-non-empty-regions}, the
  expected number of regions with these weights is more than
  $4^{k^2} c W$.  Thus, by the above argument, the expected number of
  regions for the weights $w_1, \dots, w_{n - k}$ is at least
  $4^{-k} \cdot 4^{k^2} c W$.  As we consider the case
  $\sum_{i = 1}^{n - k} w_i < 4^{-k}W$, we can substitute $W$ to
  obtain that the weights $w_1, \dots, w_{n - k}$ lead to at least
  $4^{k^2} c \sum_{i = 1}^{n - k} w_i$ non-empty regions in
  expectation.  Thus, the weights $w_1, \dots, w_{n - k}$ also form a
  counterexample for the claim in
  Equation~\eqref{eq:claim-expected-nr-of-non-empty-regions}, which is
  a contradiction to the assumption that $w_1, \dots, w_n$ is the
  minimum counterexample and thus to the assumption that there is a
  counterexample at all.
\end{proof}


  

\section{Geometric SAT with Non-Zero Temperature}
\label{sec:geometric-model-with-temperature}

In the case with temperature~$T = 0$, we used the fact that every
clause contains the $k$ variables with smallest weighted distance;
recall Section~\ref{sec:core-arguments-geometric-sat}.  This is no
longer true for higher temperatures: for $T > 0$, a clause can, in
principle, contain any variable.  However, the probability to contain
a variable that is far away is rather small.  In the remainder of
this section, we show that a constant fraction of clauses actually
behave just like in the $T = 0$ case, i.e., they contain the $k$
closest variables.  With this, we can then apply the argument outlined
in Section~\ref{sec:core-arguments-geometric-sat}.

\subsection{Expected Number of Nice Clauses}
\label{sec:expected-number-nice-clause}

Recall that a clause $c$ is generated by drawing $k$ variables without
repetition with probabilities proportional to the connection weights.
We call $c$ \emph{nice} if the $i$th variable drawn for $c$ has the
$i$th highest connection weight with $c$, i.e., $c$ does not only
contain the $k$ variables with highest connection weight but they are
drawn in descending order.  This is a slightly stronger property than
just requiring $c$ to contain the $k$ variables with lowest weighted
distance.

Let $\bar x$ be the connection weight of a variable $v$ that has
rather high connection weight with~$c$.  To show that the probability
for $v \in c$ is reasonably high, we prove that $\bar x$ is large
compared to the sum of connection weights over all variables with
smaller weight.  The following lemma bounds this sum for a given
$\bar x$.  We use the Iverson bracket to exclude the variables with
weight larger than $\bar x$ from the sum, i.e.,
$\left[X(c, v) \le \bar x\right]$ evaluates to $1$ if
$X(c, v) \le \bar x$ and to $0$ otherwise.

\begin{lemma}
  \label{lem:sum-of-weights-is-small}
  Let $c$ be a clause at any position and let $V$ be a set of $n$
  weighted variables with random positions in $\mathbb T^d$.  For
  $T < 1$ and $\bar x \in \Omega(W^{1/T})$, the expected sum of
  connection weights smaller than $\bar x$ is in $O(\bar x)$, i.e.,
  \begin{equation*}
    \EX{\sum_{v \in V} X(c, v) \cdot \left[X(c, v) < \bar
        x\right]} \in O(\bar x).
  \end{equation*}
\end{lemma}
\begin{proof}
  Using linearity of expectation, the term in the lemma's statement
  equals to the sum over the expectations
  $\EX{X(c, v) \cdot \left[X(c, v) \le \bar x\right]}$.  To bound this
  expectation, we consider the three events
  $X(c, v) \le (2^dw_v)^{1/T}$, $(2^dw_v)^{1/T} < X(c, v) < \bar x$,
  and $\bar x \le X(c, v)$.  Note that $\left[X(c, v) < \bar x\right]$
  is $0$ in the last event and $1$ in the former two.  Thus, we obtain
  \begin{align}
    \notag
    &\EX{X(c, v) \cdot \left[X(c, v) < \bar x\right]}\\
    \label{eq:expected-weight-small}
    &\quad= \Pro{X(c, v) \le (2^dw_v)^{1/T}} \cdot \EX{X(c, v) \mid
      X(c, v) \le (2^dw_v)^{1/T}}\\ 
    \label{eq:expected-weight-larger}
    &\quad+\Pro{(2^dw_v)^{1/T} < X(c, v) < \bar x} \cdot \EX{X(c, v)
      \mid (2^dw_v)^{1/T} < X(c, v) < \bar x} 
  \end{align}

  We bound the first term from above by assuming
  $X(c, v) = (2^dw_v)^{1/T}$ whenever $X(c, v) \le (2^dw_v)^{1/T}$.
  Moreover, using the CDF for $X(c, v)$
  \eqref{eq:cdf-probability-weight} yields
  \begin{align*}
    \eqref{eq:expected-weight-small}
    &\le \Pro{X(c, v) \le (2^dw_v)^{1/T}} \cdot (2^dw_v)^{1/T}\\
    &= \left(1 - \Pi_{d, \p} w_v (2^dw_v)^{-1}\right) \cdot (2^dw_v)^{1/T}\\ 
    &= \left(1 - \Pi_{d, \p} 2^{-d}\right) \cdot (2^dw_v)^{1/T} \in
      \Theta(w_v^{1/T}).
  \end{align*}

  For the second term, we have to integrate over the probability
  density function (PDF) $f_X(x)$ of the connection weights $X(c, v)$,
  which is the derivative of
  $F_X(x)$~\eqref{eq:cdf-probability-weight}.  Thus,
  $f_X(x) = T\Pi_{d, \p} w_v x^{-T - 1}$ for $x \ge (2^dw_v)^{1/T}$, and we
  obtain
  \begin{align*}
    \eqref{eq:expected-weight-larger}
    &= \Pro{(2^dw_v)^{1/T} < X(c, v) < \bar x} \cdot 
      \int_{(2^dw_v)^{1/T}}^{\bar x} \frac{x \cdot
      f_X(x)}{\Pro{(2^dw_v)^{1/T} < X(c, v) < \bar x}}\dif x\\ 
    &= T \Pi_{d, \p} w_v \cdot \int_{(2^dw_v)^{1/T}}^{\bar x} x^{-T} \dif x.\\
    \intertext{For $T < 1$, this evaluates to}
    &= T \Pi_{d, \p} w_v \cdot \left[\frac{x^{1 - T}}{1 -
      T}\right]_{(2^dw_v)^{1/T}}^{\bar x}\\
    &= \frac{T \Pi_{d, \p} w_v}{1 - T} \cdot \left[\bar x^{1 - T} -
      (2^dw_v)^{1/T - 1} \right]\\
    &\le \frac{T \Pi_{d, \p} w_v}{1 - T} \cdot \bar x^{1 - T}\\
    &\in \Theta\left(w_v \bar x^{1 - T}\right).
  \end{align*}
  
  Putting these bounds together yields
  \begin{align*}
    \EX{\sum_v X(c, v) \cdot \left[X(c, v) \le \bar x\right]}
    &= \sum_v \left(\eqref{eq:expected-weight-small} +
      \eqref{eq:expected-weight-larger}\right)\\
    &\in \Oh\left(\sum_v w_v^{1/T} + \sum_v w_v \bar x^{1-T}
      \right)\\
    &\subseteq \Oh\left(\left(\sum_v w_v\right)^{1/T} + \bar x^{1-T}
      \cdot\sum_v w_v \right)\\
    &\in\Oh\left(W^{1/T} + \bar x^{1-T} W \right).\\
  \end{align*}
  As, $\bar x \in \Omega(W^{1/T})$, we have $W^{1/T} \in O(\bar x)$,
  which handles the first term.  The second term is also in
  $O(\bar x)$, as $\bar x \in \Omega(W^{1/T})$ implies
  $W \in O(\bar x^T)$.  Thus, this yields the claimed bound of
  $O(\bar x)$.
\end{proof}

This lets us show that each clause is nice with constant probability.
The only assumption we need for this is the fact that no single weight
is too large, i.e., every weight $w_i$ has to be asymptotically
smaller than the total weight $W$.

\begin{theorem}
  \label{thm:prob-nice-clause}
  Let $\Phi$ be a random formula drawn from the weighted geometric
  model with ground space $\mathbb T^d$ equipped with a $\p$-norm,
  with temperature $T < 1$, and with $w_v/W \in o(1)$ for $v \in V$.
  Let $c$ be a clause of $\Phi$.  Then $c$ is nice with probability
  $\Omega(1)$.
\end{theorem}
\begin{proof}
  We prove two things.  First, we show that, with probability
  $\Omega(1)$, there are at least $k$ variables sufficiently close to
  $c$ that they have connection weight $\Omega(W^{1/T})$.  Second, we
  use Lemma~\ref{lem:sum-of-weights-is-small} to show that the $k$
  variables with highest connection weight are chosen for $c$ with
  constant probability (in descending order).

  For the first part, we show that there is a constant $a$ such that,
  with constant probability, at least $k$ variables have connection
  weight at least $a W^{1/T}$.  For a fixed variable $v$, we can use
  the CDF of $X(c, v)$ (Equation~\eqref{eq:cdf-probability-weight}) to
  obtain
  \begin{align*}
    \Pro{X(c, v) \ge a W^{1/T}}
    &= \Pi_{d, \p} w_v \left(a W^{1/T}\right)^{-T}\\
    &= \frac{\Pi_{d, \p}}{a^T} \frac{w_v}{W}\\
    &= 2k\frac{w_v}{W} \text{, for } a = \left(\frac{\Pi_{d, \p}}{2k}\right)^{1/T}.
  \end{align*}
  Note that this is a valid probability, as $w_v/W \in o(1)$ implies
  that it is below $1$.  For the above choice of $a$, we obtain that
  the expected number of variables with connection weight at least
  $a W^{1/T}$ is $2k$.  As the connection weights for the different
  variables are independent, we can apply the Chernoff-Hoeffding bound
  in Theorem~\ref{thm:chernoff-hoeffding} to obtain that at least $k$
  variables have connection weight $a W^{1/T}$ with constant
  probability.

  For the second part of the proof, let $\bar x$ be the connection
  weight of the $k$th closest variable.  With the argument above, we
  can assume $\bar x \in \Omega(W^{1/T})$ with constant probability,
  which lets us apply Lemma~\ref{lem:sum-of-weights-is-small}.  To do
  so, consider the experiment of drawing the first variable for our
  clause $c$.  Let $v$ be the variable that maximizes the connection
  weight $X(c, v)$.  The probability of drawing $v$ equals $X(c, v)$
  divided by the sum of all connection weights.  By
  Lemma~\ref{lem:sum-of-weights-is-small}, the sum of all connection
  weights smaller than $\bar x$ is in $O(\bar x)$.  Thus, the sum of
  all connection weights is in $O(X(c, v))$, which implies that $v$ is
  chosen with constant probability.  As we draw variables without
  repetition, the exact same argument applies for the second closest
  variable and so on.  Thus, the probability that $c$ contains the $k$
  closest variables drawn in order of descending connection weights is
  at least a constant, if there are $k$ sufficiently close variables.
  As the latter holds with constant probability, $c$ is nice with
  constant probability.
\end{proof}

By the linearity of expectation, this immediately yields the following
bound on the expected number of nice clauses.

\begin{corollary}
  \label{cor:expected-number-nice}
  Let $\Phi$ be a random formula with $m$ clauses drawn from the
  weighted geometric model with ground space $\mathbb T^d$ equipped
  with a $\p$-norm, with temperature $T < 1$, and with
  $w_v/W \in o(1)$ for $v \in V$.  The expected number of nice clauses
  in $\Phi$ is $\Theta(m)$.
\end{corollary}

\subsection{Concentration of Nice Clauses}
\label{sec:conc-nice-claus}

We show that the number of nice clauses is concentrated around its
expectation, i.e., with high probability, a constant fraction of
clauses is nice.  Our main tool for this will be the method of
typical bounded differences~\cite{w-mtbd-16}; see
Section~\ref{sec:meth-typic-bound}.  To this end, we consider several
random variables, e.g., the coordinates of clauses and variables, that
together determine the whole process of generating a random formula.
The number of nice clauses is then a function $f$ of these random
variables and its expectation is $\Theta(m)$, due to
Corollary~\ref{cor:expected-number-nice}.  Roughly speaking, the
method of bounded differences then states that the probability that
$f$ deviates too much from its expectation is low if changing a single
random variable only slightly changes~$f$.

\subsubsection{The Random Variables}
\label{sec:random-variables}

So far, we viewed the generation of a random formula as a two-step
process: first, sample coordinates for the variables and clauses;
second, sample the variables contained in each clause based on their
distances. The first step can be easily expressed via random
variables.  Let $V_1, \dots, V_n$ and $C_1, \dots C_m$ be the
coordinates\footnote{Technically, these are multivariate random
  variables, as they represent $d$-dimensional points in
  $\mathbb T^d$.} of the $n$ variables and $m$ clauses, respectively.
Though the second step heavily depends on the distances determined by
the first, we can determine all random choices in advance.  For all
$i \in [m]$ and $j \in [k]$, let $X_i^j$ be a random variable
uniformly distributed in $[0, 1)$.  The variable $X_i^j$ determines
the $j$th variable of the $i$th clause $c_i$ in the following way.  We
partition the interval $[0, 1)$ such that each variable $v$ not
already chosen for $c_i$ corresponds to a subinterval of length
proportional to the connection weight $X(c_i, v)$.  We order these
subintervals by length such that the largest interval comes first.  The
$j$th variable of $c_i$ is then the variable whose interval contains
$X_i^j$.  Note that this samples $k$ different variables for each
clause, with probabilities proportional to the connection weights
$X(c_i, v)$.  Note further that the whole generation process of a
random formula is determined by evaluating the independent random
variables $V_1, \dots, V_n, C_1, \dots, C_m, X_1^1, \dots, X_m^k$.

To formalize the concept of nice clauses in this context, we require
some more notation.  For $i \in [m]$, let $\mathcal V_i$ be the
sequence of all variables ordered decreasingly by connection weight
with the clause $c_i$.  Moreover, let $\mathcal V_i[a, b]$ denote the
subsequence from the $a$th to the $b$th variable in this sequence,
including the boundaries.  To simplify notation, we abbreviate the
unique element in $\mathcal V_i[a, a]$ with $\mathcal V_i[a]$.  Recall
that clause $c_i$ is nice if, for each of $k$ steps, we choose the
variable with highest connection weight that has not been chosen
before.  With respect to the random variables, this happens if, for
each $j \in [k]$, $X_i^j$ is smaller than the connection weight of
$\mathcal V_i[j]$ divided by the sum of all connection weights of the
remaining variables $\mathcal V_i[j, n]$.  We thus define the
indicator variable
\begin{align}
  \label{eq:indicator-var-niceness}
  N_i =
  \begin{cases}
    1, & \text{if }\forall j \in [k]\colon X_i^j < \dfrac{X(c_i, \mathcal V_i[j])}{\sum_{v\in \mathcal V_i[j, n]} X(c_i, v)},\\
    0 & \text{otherwise,}
  \end{cases}
\end{align}
which is $1$ if and only if the $i$th clause is nice.  With this, we
can define the number of nice clauses as
$f(V_1, \dots, V_n, C_1, \dots, C_m, X_1^1, \dots, X_m^k) = \sum_{i
  \in [m]} N_i$.

\subsubsection{Bounding the Effect on the Number of Nice Clauses}
\label{sec:bounding-effect}

To apply the method of bounded differences
(Theorem~\ref{thm:typical-bounded-differences} or the more specific
Corollary~\ref{cor:typical-bounded-differences}), we have to bound the
effect of changing the value of only one of these random variables on
$f$.  For the variables $C_1, \dots, C_m$, this is easy: Changing
$C_i$ moves the position of the clause $c_i$, which only makes a
difference for $c_i$.  Thus, the number of nice clauses changes by at
most~$1$.  Similarly, changing $X_i^j$ only impacts the clause $c_i$,
which implies that it changes the number of nice clauses by at
most~$1$.

For the variables $V_1, \dots, V_n$, one can actually construct
situations in which changing only a single position drops $f$ from $m$
to $0$.  There are basically two situations in which this can happen.
First, if a single variable is close to many clauses, changing its
position potentially impacts many clauses.  Second, if many
inequalities in Equation~\eqref{eq:indicator-var-niceness} are rather
tight, then moving a single variable slightly closer to many clauses
can increase the denominator on the right hand side by enough to
change $N_i$ for many clauses.  We exclude both situations by defining
unlikely bad events.  By assuming these bad events do not happen, we
can bound the effect of moving a single variable $v$ by
\begin{equation}
  \label{eq:effect-of-variable-movement}
  \delta_v = w_v^{\frac{1}{1 + T}} n^{\frac{T}{1 + T}}\log^{\frac{2}{1+T}}n.
\end{equation}
The following bound gives a simpler estimate for $\delta_v$ that will
be useful later.

\begin{lemma}
  \label{lem:nice-bound-for-delta-v}
  Let $0 < T < 1$ and $w_v \in O(n^{1-\eps})$ for an arbitrary
  $\eps > 0$.  Then
  $\delta_v \in O\left(\frac{\sqrt{w_vn}}{\log n}\right)$.
\end{lemma}
\begin{proof}
  We ignore logarithmic factors and show that
  $\delta_v / \sqrt{w_v n}$ converges polynomially to $0$ for
  $n \to \infty$.  As logarithmic factors grow slower than any
  polynomial, this proves the claim.  We get
  \begin{equation*}
    \frac{\delta_v}{\sqrt{w_vn}} =
    w_v^{\frac{1}{1 + T} - \frac{1}{2}} n^{\frac{T}{1 + T} -
      \frac{1}{2}}.
  \end{equation*}
  Rearranging the exponents yields
  \begin{equation*}
    \frac{1}{1 + T} - \frac{1}{2} =
    \frac{2 - (1 + T)}{2(1 + T)} =
    \frac{1 - T}{2(1 + T)},
    \quad \text{ and } \quad
    \frac{T}{1 + T} - \frac{1}{2} =
    \frac{2T - (1 + T)}{2(1 + T)} =
    - \frac{1 - T}{2(1 + T)}.
  \end{equation*}
  Thus $\delta_v / \sqrt{w_v n} = (w_v / n)^c$ for a positive constant
  $c$.  As $w_v \in O(n^{1 - \eps})$, this yields the claim.
\end{proof}

The following lemma states that, with overwhelming probability, no
point (and therefore no variable) is too close to too many clauses.
This eliminates the first problematic situation (and will also help
with the second).  Note that this statement only assumes random clause
positions and holds for arbitrary variable positions, i.e., when
moving a variable, we can assume that it holds before and after the
movement.

\begin{lemma}
  \label{lem:no-point-close-to-many-clauses}
  Let $m \in O(n)$, $0 < T < 1$, and $w_v \in O(n^{1-\eps})$ for every
  $v \in V$ and arbitrary constant $\eps > 0$.  Let
  $r = ( {w_v \log^2 (n)}/{n} )^{\frac{1}{d(1 + T)}}$.
  With overwhelming probability, for every point $\pnt{p}$, the ball
  $B_{\pnt{p}}(r)$ around $\pnt{p}$ with radius $r$ contains only
  $O(\delta_v)$ clauses.
\end{lemma}
\begin{proof}
  As there are uncountably many points $\pnt{p}$, it is hard to argue
  about them directly.  Thus, we first reduce the statement to one
  about finitely many positions, namely the positions of the clauses.
  Then it remains to show the statement for these positions.

  Consider a fixed point $\pnt{p}$.  As $B_{\pnt{p}}(r)$ has diameter
  $2r$, the pair-wise distance between clauses in $B_{\pnt{p}}(r)$ is
  at most $2r$.  Thus, if there exists a point $\pnt{p}$ such that
  $B_{\pnt{p}}(r)$ contains too many clauses, then there exists a
  clause that has too many other clauses at distance at most $2r$.
  Thus, it suffices to show that for every clause $c \in C$, the
  number of clauses of distance at most $2r$ to $c$ is in
  $O(\delta_v)$.

  Let $c_0$ be a fixed clause (we later apply the union bound over all
  clauses).  We want to bound the probability for another clause $c$
  to be closer than $2r$ to $c_0$.  For this, we use the CDF of the
  distance in Equation~\eqref{eq:cdf-distance}.  Note that the
  restriction of Equation~\eqref{eq:cdf-distance} to the interval
  $[0, 0.5]$ is not an issue here, as $w_v \in O(n^{1 - \eps})$
  implies $r \in o(1)$ and thus $2r \le 0.5$.  Thus, we obtain
  \begin{align*}
    \Pro{\dist{\pnt{c_0}}{\pnt{c}} \le 2r}
    &= \Pi_{d, \p} 2^d r^d \\
    &= \Pi_{d, \p} 2^d \left( \frac{w_v \log^2 n}{n} \right)^{\frac{1}{1 + T}}\\
    &= \Pi_{d, \p} 2^d \left(\frac{w_v}{n}\right)^{\frac{1}{1 + T}}\log^{\frac{2}{1+T}}n.
  \end{align*}
  As there are $m \in O(n)$ clauses, the expected number of clauses
  with distance at most $2r$ to $c_0$ is
  \begin{equation*}
    m \Pi_{d, \p} 2^d \left(\frac{w_v}{n}\right)^{\frac{1}{1 + T}} \log^{\frac{2}{1+T}}n
    \in O\left(w_v^{\frac{1}{1 + T}} n^{1 - \frac{1}{1 + T}}\log^{\frac{2}{1+T}}n\right)
    = O\left(w_v^{\frac{1}{1 + T}} n^{\frac{T}{1 + T}}\log^{\frac{2}{1+T}}n\right),
  \end{equation*}
  which is already the claimed bound of $O(\delta_v)$.  As
  $0 < T < 1$, this upper bound grows polynomially in $n$.  Thus, by
  the Chernoff-Hoeffding bound in Corollary~\ref{cor:chernoff-hoeffding-asymptotic},
  it holds asymptotically with overwhelming probability.  Applying the
  union bound over all $O(n)$ clauses yields the claim.
\end{proof}

The above lemma is stated in terms of the distances.  In the following
it will be useful to think of it in terms of connection weights
instead.  The following lemma translates the radius $r$ in
Lemma~\ref{lem:no-point-close-to-many-clauses} to the corresponding
connection weight between a clause and a variable at distance $r$.

\begin{lemma}
  \label{lem:dist-to-conn-weight-x0}
  Let $v \in V$ be a variable and let $c \in C$ be a clause with
  distance
  $\dist{\pnt{c}}{\pnt{v}} = ( {w_v \log^2 (n)}/{n} )^{\frac{1}{d(1 +
      T)}}$.  They have connection weight
  $X(c, v) = w_v^{\frac{1}{1 + T}} n^{\frac{1}{T(1 +
      T)}}\log^{-\frac{2}{T(1+T)}}n$.
\end{lemma}
\begin{proof}
  Using the definition of the connection weight and inserting the
  above distance, we obtain
  \begin{align*}
    X(c, v) &= \left( \frac{w_v}{\dist{\pnt{c}}{\pnt{v}}^d} \right)^{\frac{1}{T}}\\
            &= \left( w_v \cdot \left( \frac{n}{w_v \log^2 n}
              \right)^{\frac{1}{1 + T}} \right)^{\frac{1}{T}}\\
            &= w_v^{\frac{1}{T} - \frac{1}{T(1 + T)}} n^{\frac{1}{T(1
              + T)}} \log^{-\frac{2}{T(1 + T)}} n\\
            &= w_v^{\frac{1}{1 + T}} n^{\frac{1}{T(1 + T)}} \log^{-\frac{2}{T(1 + T)}} n.
  \end{align*}
\end{proof}

Combining Lemma~\ref{lem:no-point-close-to-many-clauses} and
Lemma~\ref{lem:dist-to-conn-weight-x0}, we obtain that, for arbitrary
variable positions (and random clause positions), no variable has a
high connection weight to too many clauses, as summarized by the
following corollary.

\begin{corollary}
  \label{cor:variables-impact-few-clauses}
  Let $m \in O(n)$, $0 < T < 1$, and $w_v \in O(n^{1-\eps})$ for every
  $v \in V$ and arbitrary constant $\eps > 0$.  With overwhelming
  probability, for every variable $v$ and every possible position of
  $v$, the number of clauses with connection weight at least
  $w_v^{\frac{1}{1 + T}} n^{\frac{1}{T(1 +
      T)}}\log^{-\frac{2}{T(1+T)}}n$ is in $O(\delta_v)$.
\end{corollary}

For the second problematic situation mentioned above, consider for a
clause $c_i$ the $k$ inequalities in
Equation~\eqref{eq:indicator-var-niceness}.  We call $c_i$
\emph{$v$-critical} if for one of these inequalities the difference
between the left and right hand side is at most $\delta_v/n$.  In the
following lemma, we first bound the number of critical clauses.
Afterwards, we show that the concept of critical clauses works as
intended in the sense that moving the variable $v$ does only change
the niceness status of $v$-critical clauses.

\begin{lemma}
  \label{lem:few-critical-clauses}
  Let $m \in O(n)$, $0 < T < 1$ and let $v$ be a variable.  With
  overwhelming probability, there are only $O(\delta_v)$ $v$-critical
  clauses.
\end{lemma}
\begin{proof}
  A clause $c_i$ can only be $v$-critical if one of the random
  variables $X_i^j$ for $j \in [k]$ differs by at most $\delta_v/n$ to
  the right hand side of the inequality in
  Equation~\eqref{eq:indicator-var-niceness}.  The probability for
  this to happen for a single $X_i^j$ is $2\delta_v/n$.  As $k$ is
  constant, $c_i$ is $v$-critical with probability $O(\delta_v/n)$.
  Thus, as $m \in O(n)$, the expected number of $v$-critical clauses
  is in $O(\delta_v)$.  As the event of being $v$-critical is
  independent for the different clauses, and as this bound is
  polynomial in $n$ for $T > 0$ (see
  Equation~\eqref{eq:effect-of-variable-movement}), the
  Chernoff-Hoeffding bound in
  Corollary~\ref{cor:chernoff-hoeffding-asymptotic} yields the claim.
\end{proof}

To prove that the movement of a single variable does not change the
niceness status of too many clauses, we argue along the following
lines.  Let $v$ be the variable we move and consider a clause $c$.
If, before or after the movement, $v$ is so close to $c$ that we get a
very high connection weight $X(c, v)$, we basically give up on $c$ and
assume that $c$ changes its status (from being nice to not being nice
or the other way round).  By
Corollary~\ref{cor:variables-impact-few-clauses} this only happens for
at most $O(\delta_v)$ clauses.  Similarly, if $c$ is $v$-critical, we
also give up on $c$, which happens for at most $O(\delta_v)$ clauses
by Lemma~\ref{lem:few-critical-clauses}.  Then it remains to show that
in all other cases (i.e., when $X(c, v)$ is low before and after the
movement and $c$ is not $v$-critical), the status of $c$ remains
unchanged.

This is done as follows.  As $c$ is not $v$-critical, the difference
between the right and left hand side of the inequality in
Equation~\eqref{eq:indicator-var-niceness} is somewhat high.  Thus, if
moving $v$ does not change the right hand side by too much, then $c$
keeps its niceness status.  To show this, we can use the fact that
$X(c, v)$ is low before and after the movement and thus it cannot
change by too much.  This change of $X(c, v)$ has to be considered
relative to the other connection weights, i.e., changing $X(c, v)$ has
less impact if there are other variables with higher connection
weight.  The following lemma establishes that these other variables
with higher connection weight indeed exist.

\begin{lemma}
  \label{lem:each-clause-has-k-heavy-variables}
  Let $w_v \in O(n^{1-\eps})$ for every $v \in V$ and arbitrary
  constant $\eps > 0$.  With overwhelming probability every clause has
  $k$ variables with connection weight at least
  $W^{\frac{1}{T}} \log^{-\frac{2}{T}} n$.
\end{lemma}
\begin{proof}
  Let $x_0 = W^{\frac{1}{T}} \log^{-\frac{2}{T}} n$ be the above
  connection weight and let $c$ be a clause with fixed position.  For
  every variable $v$, the probability for $X(c, v) \ge x_0$ is
  $\Pi_{d, \p} w_v x_0^{-T} = \Pi_{d, \p} w_v W^{-1} \log^2 n$ by
  Equation~\eqref{eq:cdf-probability-weight}.  Note that we can apply
  Equation~\eqref{eq:cdf-probability-weight} as
  $x_0 \ge (2^d w_v)^{1/T}$ due to the condition
  $w_v \in O(n^{1-\eps})$ and the fact that $W \ge n$.  Summing this
  over all variables yields that the expected number of variables with
  connection weight at least $x_0$ is $\Pi_{d, \p} \log^2 n$.  By
  Corollary~\ref{cor:chernoff-hoeffding-asymptotic}, $c$ has
  $\Omega(\log^2 n)$ variables with connection weight at least $x_0$
  with overwhelming probability.  Applying the union bound over all
  clauses and the fact that $k$ is constant while $\log^2 n$ grows
  with $n$ yields the claim.
\end{proof}

Now we are ready to bound the effect of moving just a single variable
on the number of nice clauses.

\begin{lemma}
  \label{lem:moving-one-variable-small-change}
  Let $m \in O(n)$, $0 < T < 1$, and $w_v \in O(n^{1-\eps})$ for every
  $v \in V$ and arbitrary constant $\eps > 0$.  With overwhelming
  probability, moving a single variable to an arbitrary position
  changes the number of nice clauses by only $O(\delta_v)$.
\end{lemma}
\begin{proof}
  We show the result for a fixed variable $v$.  It then follows for
  all variables using the union bound.

  Consider how the niceness status of clauses changes when moving $v$.
  Due to Corollary~\ref{cor:variables-impact-few-clauses}, there are
  only $O(\delta_v)$ clauses with connection weight at least
  $w_v^{\frac{1}{1 + T}} n^{\frac{1}{T(1 +
      T)}}\log^{-\frac{2}{T(1+T)}}n$ before or after the movement.
  Moreover, due to Lemma~\ref{lem:few-critical-clauses} there are only
  $O(\delta_v)$ $v$-critical clauses.  Thus, even if all these clauses
  change the status from being nice to not being nice or vice versa,
  the number of nice clauses changes by only $O(\delta_v)$.

  Every remaining clause $c$ is not $v$-critical and we have
  $X(c, v) \le w_v^{\frac{1}{1 + T}} n^{\frac{1}{T(1 +
      T)}}\log^{-\frac{2}{T(1+T)}}n$ before and after the movement.
  In the following we show that a clause $c$ with these two properties
  is nice after the movement if and only if it is nice before the
  movement.

  We first observe that $v$ does not belong to the $k$ variables
  closest to $c$ due to
  Lemma~\ref{lem:each-clause-has-k-heavy-variables}: With overwhelming
  probability, there are $k$ variables with connection weight at least
  $W^{\frac{1}{T}} \log^{-\frac{2}{T}} n$, which is asymptotically
  bigger than $X(c, v)$ as $w_v \in O(n^{1 - \eps})$.

  Thus, in the right hand side of the inequality in
  Equation~\eqref{eq:indicator-var-niceness}, the connection weight
  $X(c, v)$ only appears in the denominator.  To show that the right
  hand side does not change by too much, let $x$ be the numerator, let
  $y$ be the denominator before the movement, and let $y'$ be the
  denominator after the movement.  Note that $|y' - y|$ is exactly the
  change in $X(c, v)$ caused by the movement of $v$.  With this, the
  right hand side of the inequality in
  Equation~\eqref{eq:indicator-var-niceness} changes by
  \begin{align*}
    \left| \frac{x}{y} - \frac{x}{y'} \right|
    = \left| \frac{xy' - xy}{yy'} \right|
    = \frac{x}{y'}\cdot \frac{|y' - y|}{y}.
  \end{align*}
  Note that $x$ (the numerator) is the connection weight of one
  variable whose connection weight also appears in the sum of the
  denominator (after and before the movement).  Thus,
  $\frac{x}{y'} \le 1$ and the above change is upper bounded by
  $\frac{|y' - y|}{y}$.  Note that the upper bound on $X(c, v)$ holds
  before and after the movement and thus $X(c, v)$ can only change by
  less than this upper bound, i.e.,
  $|y' - y| < w_v^{\frac{1}{1 + T}} n^{\frac{1}{T(1 +
      T)}}\log^{-\frac{2}{T(1+T)}}n$.  Moreover, $y$ is the sum of
  multiple connection weights including the weight of one of the $k$
  closest variables.  Thus, by
  Lemma~\ref{lem:each-clause-has-k-heavy-variables} and the fact that
  $W \ge n$ we can assume that
  $y \ge n^{\frac{1}{T}} \log^{-\frac{2}{T}} n$.  Putting this
  together yields
  \begin{align*}
    \left| \frac{x}{y} - \frac{x}{y'} \right|
    \le \frac{|y' - y|}{y}
    < \frac{w_v^{\frac{1}{1 + T}} n^{\frac{1}{T(1 + T)}}\log^{-\frac{2}{T(1+T)}}n}{n^{\frac{1}{T}} \log^{-\frac{2}{T}} n}
    = \left(\frac{w_v}{n}\right)^{\frac{1}{1 + T}} \log^{\frac{2}{1 +
    T}} n = \frac{\delta_v}{n}.
  \end{align*}
  As $c$ is not $v$-critical, the difference between the left and
  right side of the inequality in
  Equation~\eqref{eq:indicator-var-niceness} is at least
  $\frac{\delta_v}{n}$ before the movement.  Thus, as the movement can
  change the right hand side by only less than $\frac{\delta_v}{n}$,
  the clause $c$ is nice after the movement if and only if it was nice
  before.
\end{proof}

With this we are ready to prove concentration using the method of
typical bounded differences.

\begin{theorem}
  \label{thm:nr-nice-clauses-whp}
  Let $\Phi$ be a random formula with $n$ variables and
  $m \in \Theta(n)$ clauses drawn from the weighted geometric model
  with ground space $\mathbb T^d$ equipped with a $\p$-norm, with
  temperature $0 < T < 1$, with $W \in O(n)$, and with
  $w_v \in O(n^{1-\eps})$ for every $v \in V$ and arbitrary constant
  $\eps > 0$.  With high probability, $\Theta(m)$ clauses are nice.
\end{theorem}
\begin{proof}
  We want to apply Corollary~\ref{cor:typical-bounded-differences}.
  As defined in Section~\ref{sec:random-variables}, the random
  variables are the variable positions $V_1, \dots, V_n$, the clause
  positions $C_1, \dots, C_m$, and the coin flips
  $X_1^1, \dots, X_m^k$, and the function $f$ is the number of nice
  clauses.  For $N = n + m + km$ note that $|f(X)| \le m \le N$.  For
  the nice event $\Gamma$ we assume that the statement from
  Lemma~\ref{lem:moving-one-variable-small-change} holds.  Due to
  Lemma~\ref{lem:moving-one-variable-small-change}, the probability
  for this is $\Pro{\Gamma} \ge 1 - N^{-c}$ for any constant $c$ and
  sufficiently large $N$.  Thus, when choosing $c \ge 3$, we satisfy
  the condition $|f(X)| \le N^{c - 2}$ of
  Corollary~\ref{cor:typical-bounded-differences}.

  Now we have to bound the change of $f$ when changing only one of the
  random variables, assuming we start with an event in $\Gamma$, i.e.,
  we have to determine the $\Delta_i$ from
  Corollary~\ref{cor:typical-bounded-differences}.  As mentioned
  before, changing a clause position $C_i$ or one of the $X_i^j$
  impacts only one clause and thus changes $f$ by at most $1$.
  Moreover, as we start with a configuration satisfying
  Lemma~\ref{lem:moving-one-variable-small-change}, $f$ changes by
  only $O(\delta_v)$ for variable $v$.  Thus, for the sum in
  Corollary~\ref{cor:typical-bounded-differences} we obtain
  \begin{equation*}
    \sum_{i \in [N]} \Delta_i^2 \in O\left( m + km + \sum_{v \in V} \delta_v^2 \right).
  \end{equation*}
  Due to Lemma~\ref{lem:nice-bound-for-delta-v}, we have
  $\delta_v \in O(\sqrt{w_vn} / \log n)$.  Thus, the above sum can be
  bounded by
  \begin{equation*}
    \sum_{v \in V} \delta_v^2 \in O\left(\sum_{v \in V} \left(\frac{\sqrt{w_v n
    }}{\log n}\right)^2\right) = O\left(\frac{n}{\log^2 n}\sum_{v
        \in V} w_v\right) = O\left(\frac{n^2}{\log^2 n}\right).
  \end{equation*}
  As $\EX{f} \in \Theta(m) = \Theta(n)$, this is exactly the bound
  required by Corollary~\ref{cor:typical-bounded-differences} and thus
  the number of nice clauses is in $\Theta(m)$ with high probability.
\end{proof}

\subsubsection{Putting Things Together}
\label{sec:putt-things-togeth}

Now we are ready to prove our main theorem for the geometric model.

\GeometricSat
\begin{proof}
  Let $m'$ be the number of clauses in $\Phi$ that consist of the $k$
  variables with minimum weighted distance.  By
  Theorem~\ref{thm:nr-nice-clauses-whp} we have
  $m' \in \Theta(m) = \Theta(n)$.  In the following, we consider only
  these clauses.

  Consider the weighted order-$k$ Voronoi diagram of the $n$ variables
  and let $n'$ be the number of non-empty regions.  By
  Theorem~\ref{thm:complexity-random-voronoi-diagram} and due to
  $W \in O(n)$, we have $\EX{n'} \in O(n)$.  Moreover, it follows from
  Markov's inequality that $n' \le n\log n$ holds asymptotically
  almost surely:
  \begin{equation*}
    \Pro{n' \ge n \log n} \le \frac{\EX{n'}}{n \log n} \in
    O\left(\frac{1}{\log n}\right).
  \end{equation*}
  
  Now, determining the $k$ variables of a clause $c$ is equivalent to
  observing which region of the order-$k$ Voronoi diagram contains
  $c$, or more precisely, which $k$ variables define this region.
  Thus, choosing random positions for the clauses is like throwing
  $m'$ balls into $n'$ (non-uniform) bins.  Thus, if
  $m' \in \Omega(n'/\polylog n')$, we can apply
  Corollary~\ref{cor:balls-into-unif-bins}.  With the above bounds,
  which hold asymptotically almost surely, it is not hard to see that
  this condition in fact holds: If $n' \le n$, it clearly holds as
  $m' \in \Omega(n)$.  Otherwise, we have
  $n' \le n \log n \le n \log n'$, which implies $n \ge n'/\log n'$,
  and thus $m \in \Omega(n'/\log n')$.

  Applying Corollary~\ref{cor:balls-into-unif-bins} tells us that,
  asymptotically almost surely, there is a bin with a superconstant
  number of balls.  In other words, there is a superconstant number of
  clauses that share the same set of $k$ variables.  For sufficiently
  large $n$, this is bigger than $2^k$, which implies an unsatisfiable
  subformula consisting of only $2^k$ clauses.  Clearly, it can be
  found in $O(n \log n)$ time by sorting the clauses lexicographically
  with respect to the contained variables.
\end{proof}

} 

\subsection*{Acknowledgments}
\label{sec:acknowledgments}
The authors would like to thank Thomas Sauerwald for the fruitful
discussions on random SAT models and bipartite expansion.

\bibliographystyle{plainnat}
\bibliography{references}

\begin{thebibliography}{63}
\providecommand{\natexlab}[1]{#1}
\providecommand{\url}[1]{\texttt{#1}}
\expandafter\ifx\csname urlstyle\endcsname\relax
  \providecommand{\doi}[1]{doi: #1}\else
  \providecommand{\doi}{doi: \begingroup \urlstyle{rm}\Url}\fi

\bibitem[Ans{\'{o}}tegui et~al.(2009)Ans{\'{o}}tegui, Bonet, and
  Levy]{abl-sisi-09}
Carlos Ans{\'{o}}tegui, Maria~Luisa Bonet, and Jordi Levy.
\newblock On the structure of industrial {SAT} instances.
\newblock In \emph{Principles and Practice of Constraint Programming ({CP})},
  pages 127--141, 2009.
\newblock \doi{10.1007/978-3-642-04244-7_13}.

\bibitem[Ans{\'o}tegui et~al.(2012)Ans{\'o}tegui, Gir{\'a}ldez-Cru, and
  Levy]{agl-cssf-12}
Carlos Ans{\'o}tegui, Jes{\'u}s Gir{\'a}ldez-Cru, and Jordi Levy.
\newblock The community structure of {SAT} formulas.
\newblock In \emph{Theory and Applications of Satisfiability Testing ({SAT})},
  pages 410--423, 2012.
\newblock \doi{10.1007/978-3-642-31612-8_31}.

\bibitem[Ansótegui et~al.(2009)Ansótegui, Bonet, and Levy]{abl-tilrsi-09}
Carlos Ansótegui, María~Luisa Bonet, and Jordi Levy.
\newblock Towards industrial-like random {SAT} instances.
\newblock In \emph{International Joint Conference on Artifical Intelligence
  (IJCAI)}, pages 387--392, 2009.
\newblock \doi{10.5555/1661445.1661507}.

\bibitem[Ansótegui et~al.(2017)Ansótegui, Bonet, and Levy]{levy17}
Carlos Ansótegui, Maria~Luisa Bonet, and Jordi Levy.
\newblock Scale-free random {SAT} instances.
\newblock \emph{CoRR}, abs/1708.06805v3, 2017.
\newblock URL \url{https://arxiv.org/abs/1708.06805v3}.

\bibitem[Aurenhammer and Edelsbrunner(1984)]{ae-oacwvdp-84}
Franz Aurenhammer and Herbert Edelsbrunner.
\newblock An optimal algorithm for constructing the weighted {Voronoi} diagram
  in the plane.
\newblock \emph{Pattern Recognition}, 17\penalty0 (2):\penalty0 251--257, 1984.
\newblock \doi{10.1016/0031-3203(84)90064-5}.

\bibitem[Aurenhammer et~al.(2013)Aurenhammer, Klein, and Lee]{akl-vddt-13}
Franz Aurenhammer, Rolf Klein, and Der-Tsai Lee.
\newblock \emph{Voronoi Diagrams and Delaunay Triangulations}.
\newblock WORLD SCIENTIFIC, 2013.
\newblock \doi{10.1142/8685}.

\bibitem[Beame and Sabharwal(2014)]{bs-nssspe-14}
Paul Beame and Ashish Sabharwal.
\newblock Non-restarting {SAT} solvers with simple preprocessing can
  efficiently simulate resolution.
\newblock In \emph{Conference on Artificial Intelligence (AAAI)}, pages
  2608--2615, 2014.
\newblock URL
  \url{https://www.aaai.org/ocs/index.php/AAAI/AAAI14/paper/view/8397}.

\bibitem[Ben{-}Sasson and Galesi(2003)]{BenGalesi}
Eli Ben{-}Sasson and Nicola Galesi.
\newblock Space complexity of random formulae in resolution.
\newblock \emph{Random Structures and Algorithms}, 23\penalty0 (1):\penalty0
  92--109, 2003.
\newblock \doi{10.1002/rsa.10089}.

\bibitem[Ben-Sasson and Wigderson(2001)]{shortProofs}
Eli Ben-Sasson and Avi Wigderson.
\newblock Short proofs are narrow - resolution made simple.
\newblock \emph{J. ACM}, 48\penalty0 (2):\penalty0 149--169, 2001.
\newblock \doi{10.1145/375827.375835}.

\bibitem[Bienkowski et~al.(2005)Bienkowski, Damerow, {Meyer auf der Heide}, and
  Sohler]{bdhs-accvd-05}
Marcin Bienkowski, Valentina Damerow, Friedhelm {Meyer auf der Heide}, and
  Christian Sohler.
\newblock Average case complexity of {Voronoi} diagrams of n sites from the
  unit cube.
\newblock In \emph{European Workshop on Computational Geometry (EuroCG)}, pages
  167--170, 2005.
\newblock URL \url{http://www.win.tue.nl/EWCG2005/Proceedings/43.pdf}.

\bibitem[Bl{\"a}sius et~al.(2018)Bl{\"a}sius, Freiberger, Friedrich, Katzmann,
  Montenegro-Retana, and Thieffry]{bff-espsf-18}
Thomas Bl{\"a}sius, Cedric Freiberger, Tobias Friedrich, Maximilian Katzmann,
  Felix Montenegro-Retana, and Marianne Thieffry.
\newblock Efficient shortest paths in scale-free networks with underlying
  hyperbolic geometry.
\newblock In \emph{International Colloquium on Automata, Languages, and
  Programming (ICALP)}, pages 20:1--20:14, 2018.
\newblock \doi{10.4230/LIPIcs.ICALP.2018.20}.

\bibitem[Bl{\"a}sius et~al.(2020)Bl{\"a}sius, Fischbeck, Friedrich, and
  Katzmann]{bffk-svcpt-20}
Thomas Bl{\"a}sius, Philipp Fischbeck, Tobias Friedrich, and Maximilian
  Katzmann.
\newblock Solving vertex cover in polynomial time on hyperbolic random graphs.
\newblock In \emph{Annual Symposium on Theoretical Aspects of Computer Science
  (STACS)}, volume 154, pages 25:1--25:14, 2020.
\newblock \doi{10.4230/LIPIcs.STACS.2020.25}.

\bibitem[Bläsius et~al.(2019)Bläsius, Friedrich, and Sutton]{bfs-etcsf-19}
Thomas Bläsius, Tobias Friedrich, and Andrew~M. Sutton.
\newblock On the empirical time complexity of scale-free 3-{SAT} at the phase
  transition.
\newblock In \emph{Tools and Algorithms for the Construction and Analysis of
  Systems (TACAS)}, pages 117--134, 2019.
\newblock \doi{10.1007/978-3-030-17462-0_7}.

\bibitem[Bode et~al.(2013)Bode, Fountoulakis, and Müller]{bfm-gcrhg-13}
Michel Bode, Nikolaos Fountoulakis, and Tobias Müller.
\newblock On the giant component of random hyperbolic graphs.
\newblock In \emph{European Conference on Combinatorics, Graph Theory and
  Applications}, pages 425--429, 2013.
\newblock \doi{10.1007/978-88-7642-475-5_68}.

\bibitem[Bohler et~al.(2015)Bohler, Cheilaris, Klein, Liu, Papadopoulou, and
  Zavershynskyi]{bck-choavd-15}
Cecilia Bohler, Panagiotis Cheilaris, Rolf Klein, Chih-Hung Liu, Evanthia
  Papadopoulou, and Maksym Zavershynskyi.
\newblock On the complexity of higher order abstract voronoi diagrams.
\newblock \emph{Computational Geometry: Theory and Applications}, 48\penalty0
  (8):\penalty0 539--551, 2015.
\newblock \doi{10.1016/j.comgeo.2015.04.008}.

\bibitem[Boissonnat et~al.(1998)Boissonnat, Sharir, Tagansky, and
  Yvinec]{bsty-vdhdc-98}
Jean-Daniel Boissonnat, Micha Sharir, Boaz Tagansky, and Mariette Yvinec.
\newblock {Voronoi} diagrams in higher dimensions under certain polyhedral
  distance functions.
\newblock \emph{Discrete {\&} Computational Geometry}, 19\penalty0
  (4):\penalty0 485--519, 1998.
\newblock \doi{10.1007/PL00009366}.

\bibitem[Boots(1980)]{b-wtp-80}
Barry~N. Boots.
\newblock Weighting thiessen polygons.
\newblock \emph{Economic Geography}, 56\penalty0 (3):\penalty0 248--259, 1980.
\newblock \doi{10.2307/142716}.

\bibitem[Bringmann et~al.(2017)Bringmann, Keusch, and Lengler]{bkl-sgirglt-17}
Karl Bringmann, Ralph Keusch, and Johannes Lengler.
\newblock Sampling geometric inhomogeneous random graphs in linear time.
\newblock In \emph{Annual European Symposium on Algorithms {(ESA)}}, volume~87,
  pages 20:1--20:15, 2017.
\newblock \doi{10.4230/LIPIcs.ESA.2017.20}.

\bibitem[Chung and Lu(2002{\natexlab{a}})]{cl-adrgged-02}
Fan Chung and Linyuan Lu.
\newblock The average distances in random graphs with given expected degrees.
\newblock \emph{Proceedings of the National Academy of Sciences}, 99\penalty0
  (25):\penalty0 15879--15882, 2002{\natexlab{a}}.
\newblock \doi{10.1073/pnas.252631999}.

\bibitem[Chung and Lu(2002{\natexlab{b}})]{cl-ccrgg-02}
Fan Chung and Linyuan Lu.
\newblock Connected components in random graphs with given expected degree
  sequences.
\newblock \emph{Annals of Combinatorics}, 6\penalty0 (2):\penalty0 125--145,
  2002{\natexlab{b}}.
\newblock \doi{10.1007/PL00012580}.

\bibitem[Chv{\'{a}}tal and Szemer{\'{e}}di(1988)]{manyExamples}
Vasek Chv{\'{a}}tal and Endre Szemer{\'{e}}di.
\newblock Many hard examples for resolution.
\newblock \emph{J. {ACM}}, 35\penalty0 (4):\penalty0 759--768, 1988.
\newblock \doi{10.1145/48014.48016}.

\bibitem[Cooper et~al.(2007)Cooper, Frieze, and Sorkin]{cooper2SAT}
Colin Cooper, Alan Frieze, and Gregory Sorkin.
\newblock Random {2SAT} with prescribed literal degrees.
\newblock \emph{Algorithmica}, 48:\penalty0 249--265, 2007.

\bibitem[Cundefinedrbunar et~al.(2006)Cundefinedrbunar, Grama, Vitek, and
  Cundefinedrbunar]{cgvc-rcdsn-06}
Bogdan Cundefinedrbunar, Ananth Grama, Jan Vitek, and Octavian
  Cundefinedrbunar.
\newblock Redundancy and coverage detection in sensor networks.
\newblock \emph{ACM Transactions on Sensor Networks}, 2\penalty0 (1):\penalty0
  94--128, 2006.
\newblock \doi{10.1145/1138127.1138131}.

\bibitem[Davis and Putnam(1960)]{dp-cpqt-60}
Martin Davis and Hilary Putnam.
\newblock A computing procedure for quantification theory.
\newblock \emph{J. ACM}, 7\penalty0 (3):\penalty0 201–215, 1960.
\newblock \doi{10.1145/321033.321034}.

\bibitem[Deza and Deza(2006)]{dd-cvdd-06}
Elena Deza and Michel-Marie Deza.
\newblock \emph{Chapter 20 - Voronoi Diagram Distances}, pages 253--261.
\newblock Elsevier, 2006.
\newblock \doi{10.1016/B978-044452087-6/50020-2}.

\bibitem[Driemel et~al.(2016)Driemel, Har-Peled, and Raichel]{dhr-ecvdt-16}
Anne Driemel, Sariel Har-Peled, and Benjamin Raichel.
\newblock On the expected complexity of {Voronoi} diagrams on terrains.
\newblock \emph{ACM Transactions on Algorithms}, 12\penalty0 (3), 2016.
\newblock \doi{10.1145/2846099}.

\bibitem[Dubhashi and Panconesi(2012)]{dp-cmara-12}
Devdatt~P. Dubhashi and Alessandro Panconesi.
\newblock \emph{Concentration of Measure for the Analysis of Randomized
  Algorithms}.
\newblock Cambridge University Press, 2012.
\newblock \doi{10.1017/CBO9780511581274}.

\bibitem[Dwyer(1991)]{d-hdvdlet-91}
Rex~A. Dwyer.
\newblock Higher-dimensional {Voronoi} diagrams in linear expected time.
\newblock \emph{Discrete {\&} Computational Geometry}, 6\penalty0 (3):\penalty0
  343--367, 1991.
\newblock \doi{10.1007/BF02574694}.

\bibitem[Erickson(2001)]{e-npsch-01}
Jeff Erickson.
\newblock Nice point sets can have nasty {Delaunay} triangulations.
\newblock In \emph{Annual Symposium on Computational Geometry (SoCG)}, pages
  96--105, 2001.
\newblock \doi{10.1145/378583.378636}.

\bibitem[Erickson(2005)]{e-dpshsdt-02}
Jeff Erickson.
\newblock Dense point sets have sparse {Delaunay} triangulations: Or ``\dots
  but not too nasty''.
\newblock \emph{Discrete {\&} Computational Geometry}, 33:\penalty0 83--115,
  2005.
\newblock \doi{10.1007/s00454-004-1089-3}.

\bibitem[Esteban and Tor{\'{a}}n(2001)]{spaceBounds}
Juan~Luis Esteban and Jacobo Tor{\'{a}}n.
\newblock Space bounds for resolution.
\newblock \emph{Information and Computation}, 171\penalty0 (1):\penalty0
  84--97, 2001.
\newblock \doi{10.1006/inco.2001.2921}.

\bibitem[Fan and Raichel(2020)]{fr-lecdmvd-20}
Chenglin Fan and Benjamin Raichel.
\newblock Linear expected complexity for directional and multiplicative voronoi
  diagrams.
\newblock In \emph{Annual European Symposium on Algorithms {(ESA)}}, pages
  45:1--45:18, 2020.
\newblock \doi{10.4230/LIPIcs.ESA.2020.45}.

\bibitem[Friedrich and Rothenberger(2019)]{fr-stnr2-19}
Tobias Friedrich and Ralf Rothenberger.
\newblock The satisfiability threshold for non-uniform random 2-{SAT}.
\newblock In \emph{International Colloquium on Automata, Languages, and
  Programming, (ICALP)}, volume 132 of \emph{LIPIcs}, pages 61:1--61:14, 2019.
\newblock \doi{10.4230/LIPIcs.ICALP.2019.61}.

\bibitem[Friedrich et~al.(2017{\natexlab{a}})Friedrich, Krohmer, Rothenberger,
  Sauerwald, and Sutton]{fkrss-bstpldrs-17}
Tobias Friedrich, Anton Krohmer, Ralf Rothenberger, Thomas Sauerwald, and
  Andrew~M. Sutton.
\newblock Bounds on the satisfiability threshold for power law distributed
  random {SAT}.
\newblock In \emph{Annual European Symposium on Algorithms (ESA)}, pages
  37:1--37:15, 2017{\natexlab{a}}.
\newblock \doi{10.4230/LIPIcs.ESA.2017.37}.

\bibitem[Friedrich et~al.(2017{\natexlab{b}})Friedrich, Krohmer, Rothenberger,
  and Sutton]{fkrs-ptssf-17}
Tobias Friedrich, Anton Krohmer, Ralf Rothenberger, and Andrew~M. Sutton.
\newblock Phase transitions for scale-free {SAT} formulas.
\newblock In \emph{Conference on Artificial Intelligence,(AAAI)}, pages
  3893--3899, 2017{\natexlab{b}}.
\newblock URL \url{http://aaai.org/ocs/index.php/AAAI/AAAI17/paper/view/14755}.

\bibitem[Galvão et~al.(2006)Galvão, Novaes, de~Cursi, and
  Souza]{gncs-mwvdald-06}
Lauro~C. Galvão, Antonio~G.N. Novaes, J.E.~Souza de~Cursi, and João~C. Souza.
\newblock A multiplicatively-weighted {Voronoi} diagram approach to logistics
  districting.
\newblock \emph{Computers \& Operations Research}, 33\penalty0 (1):\penalty0
  93--114, 2006.
\newblock \doi{10.1016/j.cor.2004.07.001}.

\bibitem[Gemsa et~al.(2012)Gemsa, Lee, Liu, and Wagner]{gllw-hocvd-12}
Andreas Gemsa, D.~T. Lee, Chih-Hung Liu, and Dorothea Wagner.
\newblock Higher order city {Voronoi} diagrams.
\newblock In \emph{Scandinavian Workshop on Algorithm Theory (SWAT)}, pages
  59--70, 2012.
\newblock \doi{10.1007/978-3-642-31155-0_6}.

\bibitem[Gir\'{a}ldez-Cru and Levy(2015)]{gl-mbrsig-15}
Jes\'{u}s Gir\'{a}ldez-Cru and Jordi Levy.
\newblock A modularity-based random {SAT} instances generator.
\newblock In \emph{International Joint Conference on Artifical Intelligence
  (IJCAI)}, pages 1952–--1958, 2015.
\newblock \doi{10.5555/2832415.2832520}.

\bibitem[Giráldez-Cru and Levy(2017{\natexlab{a}})]{gl-dpssi-17}
Jesús Giráldez-Cru and Jordi Levy.
\newblock Description of popularity-similarity {SAT} instances.
\newblock In Tomáš Balyo, Marijn J.~H. Heule, and Matti Järvisalo, editors,
  \emph{Proceedings of {SAT} Competition 2017: Solver and Benchmark
  Descriptions}, pages 49--50, 2017{\natexlab{a}}.
\newblock URL \url{https://helda.helsinki.fi/handle/10138/224324}.

\bibitem[Giráldez-Cru and Levy(2017{\natexlab{b}})]{gl-lrsi-17}
Jesús Giráldez-Cru and Jordi Levy.
\newblock Locality in random {SAT} instances.
\newblock In \emph{International Joint Conference on Artificial Intelligence
  (IJCAI)}, pages 638--644, 2017{\natexlab{b}}.
\newblock \doi{10.24963/ijcai.2017/89}.

\bibitem[Golin and Na(2003)]{gn-acvdr-03}
Mordecai~J. Golin and Hyeon-Suk Na.
\newblock On the average complexity of {3D-Voronoi} diagrams of random points
  on convex polytopes.
\newblock \emph{Computational Geometry}, 25\penalty0 (3):\penalty0 197--231,
  2003.
\newblock \doi{10.1016/S0925-7721(02)00123-2}.

\bibitem[Har-Peled and Raichel(2015)]{hr-crwmvd-15}
Sariel Har-Peled and Benjamin Raichel.
\newblock On the complexity of randomly weighted multiplicative {Voronoi}
  diagrams.
\newblock \emph{Discrete {\&} Computational Geometry}, 53\penalty0
  (3):\penalty0 547--568, 2015.
\newblock \doi{10.1007/s00454-015-9675-0}.

\bibitem[Klee(1980)]{k-cdvd-80}
Victor Klee.
\newblock On the complexity of $d$-dimensional {Voronoi} diagrams.
\newblock \emph{Archiv der Mathematik}, 34\penalty0 (1):\penalty0 75--80, 1980.
\newblock \doi{10.1007/BF01224932}.

\bibitem[Krioukov et~al.(2010)Krioukov, Papadopoulos, Kitsak, Vahdat, and
  Bogu\~n\'a]{kpk-hgcn-10}
Dmitri Krioukov, Fragkiskos Papadopoulos, Maksim Kitsak, Amin Vahdat, and
  Mari\'an Bogu\~n\'a.
\newblock Hyperbolic geometry of complex networks.
\newblock \emph{Physical Review E}, 82\penalty0 (3):\penalty0 036106, 2010.
\newblock \doi{10.1103/PhysRevE.82.036106}.

\bibitem[L{\^e}(1996)]{l-vdlrd-96}
Ng\d{o}c-Minh L{\^e}.
\newblock On {Voronoi} diagrams in the ${L}_p$-metric in $\mathbb{R}^d$.
\newblock \emph{Discrete {\&} Computational Geometry}, 16\penalty0
  (2):\penalty0 177--196, 1996.
\newblock \doi{10.1007/BF02716806}.

\bibitem[Lee(1982)]{l-nnvdp-82}
Der-Tsai Lee.
\newblock On k-nearest neighbor {Voronoi} diagrams in the plane.
\newblock \emph{IEEE Transactions on Computers}, 31:\penalty0 478--487, 1982.
\newblock \doi{10.1109/TC.1982.1676031}.

\bibitem[Liu et~al.(2011)Liu, Papadopoulou, and Lee]{lpl-osall-11}
Chih-Hung Liu, Evanthia Papadopoulou, and D.~T. Lee.
\newblock An output-sensitive approach for the ${L}_1$/${L}_\infty$
  $k$-nearest-neighbor {Voronoi} diagram.
\newblock In \emph{Annual European Symposium on Algorithms {(ESA)}}, pages
  70--81, 2011.
\newblock \doi{10.1007/978-3-642-23719-5_7}.

\bibitem[Müller and Staps(2019)]{ms-dkrg-19}
Tobias Müller and Merlijn Staps.
\newblock The diameter of {KPKVB} random graphs.
\newblock \emph{Advances in Applied Probability}, 51\penalty0 (2):\penalty0
  358--377, 2019.
\newblock \doi{10.1017/apr.2019.23}.

\bibitem[Mull et~al.(2016)Mull, Fremont, and Seshia]{mds-hscs-16}
Nathan Mull, Daniel~J. Fremont, and Sanjit~A. Seshia.
\newblock On the hardness of {SAT} with community structure.
\newblock In \emph{Theory and Applications of Satisfiability Testing {(SAT)}},
  pages 141--159, 2016.
\newblock \doi{10.1007/978-3-319-40970-2_10}.

\bibitem[Mulmuley(1991)]{m-lavd-91}
Ketan Mulmuley.
\newblock On levels in arrangements and {Voronoi} diagrams.
\newblock \emph{Discrete {\&} Computational Geometry}, 6\penalty0 (3):\penalty0
  307--338, 1991.
\newblock \doi{10.1007/BF02574692}.

\bibitem[Newsham et~al.(2014)Newsham, Ganesh, Fischmeister, Audemard, and
  Simon]{ngf-icsssp-14}
Zack Newsham, Vijay Ganesh, Sebastian Fischmeister, Gilles Audemard, and
  Laurent Simon.
\newblock Impact of community structure on {SAT} solver performance.
\newblock In Carsten Sinz and Uwe Egly, editors, \emph{Theory and Applications
  of Satisfiability Testing (SAT)}, pages 252--268, 2014.
\newblock \doi{10.1007/978-3-319-09284-3_20}.

\bibitem[Omelchenko and Bulatov(2019)]{omelchenko2019}
Oleksii Omelchenko and Andrei~A. Bulatov.
\newblock Satisfiability threshold for power law random 2-sat in configuration
  model.
\newblock In Mikol{\'a}{\v{s}} Janota and In{\^e}s Lynce, editors, \emph{Theory
  and Applications of Satisfiability Testing ({SAT})}, pages 53--70, Cham,
  2019. Springer International Publishing.

\bibitem[Papadopoulos et~al.(2012)Papadopoulos, Kitsak, Serrano, Boguñá, and
  Krioukov]{pks-pvsgn-12}
Fragkiskos Papadopoulos, Maksim Kitsak, M.~Ángeles Serrano, Marián Boguñá,
  and Dmitri Krioukov.
\newblock Popularity versus similarity in growing networks.
\newblock \emph{Nature}, 489:\penalty0 537--540, 2012.
\newblock \doi{10.1038/nature11459}.

\bibitem[Pipatsrisawat and Darwiche(2011)]{pd-pcssre-11}
Knot Pipatsrisawat and Adnan Darwiche.
\newblock On the power of clause-learning {SAT} solvers as resolution engines.
\newblock \emph{Artificial Intelligence}, 175\penalty0 (2):\penalty0 512--525,
  2011.
\newblock \doi{10.1016/j.artint.2010.10.002}.

\bibitem[Raab and Steger(1998)]{rs-bb-98}
Martin Raab and Angelika Steger.
\newblock ``{Balls} into bins'' --- a simple and tight analysis.
\newblock In \emph{Randomization and Approximation Techniques in Computer
  Science}, pages 159--170, 1998.
\newblock \doi{10.1007/3-540-49543-6_13}.

\bibitem[Seidel(1987)]{s-nfhdvd-87}
Raimund Seidel.
\newblock On the number of faces in higher-dimensional {Voronoi} diagrams.
\newblock In \emph{Annual Symposium on Computational Geometry (SoCG)}, pages
  181--185, 1987.
\newblock \doi{10.1145/41958.41977}.

\bibitem[Shamos and Hoey(1975)]{sh-cpp-75}
Michael~Ian Shamos and Dan Hoey.
\newblock Closest-point problems.
\newblock In \emph{Annual Symposium on Foundations of Computer Science (FOCS)},
  pages 151--162, 1975.
\newblock \doi{10.1109/SFCS.1975.8}.

\bibitem[van Beek(2006)]{b-bsa-06}
Peter van Beek.
\newblock Backtracking search algorithms.
\newblock In Francesca Rossi, Peter van Beek, and Toby Walsh, editors,
  \emph{Handbook of Constraint Programming}, volume~2 of \emph{Foundations of
  Artificial Intelligence}, pages 85--134. Elsevier, 2006.
\newblock \doi{10.1016/S1574-6526(06)80008-8}.
\newblock URL \url{https://doi.org/10.1016/S1574-6526(06)80008-8}.

\bibitem[Vardi(2014)]{v-bs-14}
Moshe~Y. Vardi.
\newblock Boolean satisfiability: Theory and engineering.
\newblock \emph{Communications of the {ACM}}, 57\penalty0 (3):\penalty0 5--5,
  2014.
\newblock \doi{10.1145/2578043}.

\bibitem[Voitalov et~al.(2019)Voitalov, van~der Hoorn, van~der Hofstad, and
  Krioukov]{vhhk-snwd-18}
Ivan Voitalov, Pim van~der Hoorn, Remco van~der Hofstad, and Dmitri Krioukov.
\newblock Scale-free networks well done.
\newblock \emph{Physical Review Research}, 1:\penalty0 033034, 2019.
\newblock \doi{10.1103/PhysRevResearch.1.033034}.

\bibitem[Warnke(2016)]{w-mtbd-16}
Lutz Warnke.
\newblock On the method of typical bounded differences.
\newblock \emph{Combinatorics, Probability and Computing}, 25\penalty0
  (2):\penalty0 269–--299, 2016.
\newblock \doi{10.1017/S0963548315000103}.

\bibitem[Zulkoski et~al.(2017)Zulkoski, Martins, Wintersteiger, Robere, Liang,
  Czarnecki, and Ganesh]{zmw-rcpssp-17}
Edward Zulkoski, Ruben Martins, Christoph~M. Wintersteiger, Robert Robere, Jia
  Liang, Krzysztof Czarnecki, and Vijay Ganesh.
\newblock Relating complexity-theoretic parameters with {SAT} solver
  performance.
\newblock \emph{CoRR}, abs/1706.08611, 2017.
\newblock URL \url{http://arxiv.org/abs/1706.08611}.

\bibitem[Zulkoski et~al.(2018)Zulkoski, Martins, Wintersteiger, Liang,
  Czarnecki, and Ganesh]{zmw-esmmssp-18}
Edward Zulkoski, Ruben Martins, Christoph~M. Wintersteiger, Jia~Hui Liang,
  Krzysztof Czarnecki, and Vijay Ganesh.
\newblock The effect of structural measures and merges on {SAT} solver
  performance.
\newblock In \emph{Principles and Practice of Constraint Programming (CP)},
  pages 436--452, 2018.
\newblock \doi{10.1007/978-3-319-98334-9_29}.

\end{thebibliography}

\shortOrLong{
}{

\appendix

\section{Basic Technical Tools}
\label{sec:basic-techn-tools}

This section is a collection of tools we use throughout the paper that
were either known before or are straight-forward to prove but distract
from the core arguments we make in the paper.

\subsection{Discrete Power-Law Weights}
\label{sec:discrete-power-law}

The following lemma summarizes some properties of the probability
distribution given by the discrete power-law weights.

\begin{lemma} \label{lem:powerlaw}
Let $\beta>2$ and \[p_i=\frac{i^{-1/(\beta-1)}}{\sum_{j=1}^n j^{-1/(\beta-1)}}\] for $i\in[n]$.
It holds that
\[\sum_{j=1}^n j^{-1/(\beta-1)}=\left(1+o(1)\right)\frac{\beta-1}{\beta-2}\cdot n^{(\beta-2)/(\beta-1)},\]
\[F(i):= \sum_{j=1}^i p_j \in \Oh{\left(\left(\frac{i}{n}\right)^{(\beta-2)/(\beta-1)}\right)},\]
and
\begin{equation}\sum_{j=1}^n p_j^2\in
\begin{cases}
\Theta\left(n^{-2\frac{\beta-2}{\beta-1}}\right),& \beta<3;\\
\Theta\left(\ln n/n\right),& \beta=3;\\
\Theta\left(n^{-1}\right),& \beta>3.
\end{cases}
\end{equation}
\end{lemma}
\begin{proof}
Since $j^{-1/(\beta-1)}$ is monotonically decreasing, it holds that
\begin{align*}
\sum_{j=1}^n j^{-1/(\beta-1)} 
	& \le 1 + \int_{j=1}^{n} j^{-1/(\beta-1)} \mathrm{d}j\\
	& = 1 + \frac{\beta-1}{\beta-2}\left(n^{(\beta-2)/(\beta-1)}-1\right) = \frac{\beta-1}{\beta-2}\cdot n^{(\beta-2)/(\beta-1)}-\frac{1}{\beta-2}
\end{align*}
and
\begin{align*}
\sum_{j=1}^n j^{-1/(\beta-1)} 
	& \ge n^{-1/(\beta-1)}+\int_{j=1}^{n} j^{-1/(\beta-1)} \mathrm{d}j\\
	& = \frac{\beta-1}{\beta-2}\cdot n^{(\beta-2)/(\beta-1)}-\frac{\beta-1}{\beta-2}+n^{-1/(\beta-1)}.
\end{align*}
Equivalently, we get
\begin{align*}
F(i) = \sum_{j=1}^i p_j 
	& = \frac{\sum_{j=1}^i j^{-1/(\beta-1)}}{\sum_{j=1}^n j^{-1/(\beta-1)}}\\
	& \le\frac{1}{\sum_{j=1}^n j^{-1/(\beta-1)}}\left(1 + \int_{j=1}^{i} j^{-1/(\beta-1)} \mathrm{d}j\right)\\
	& \le\frac{1}{\sum_{j=1}^n j^{-1/(\beta-1)}}\left(\frac{\beta-1}{\beta-2}\cdot i^{(\beta-2)/(\beta-1)}-\frac{1}{\beta-2}\right)\\
	& \in \Oh\left(\left(\frac{i}{n}\right)^{(\beta-2)/(\beta-1)}\right).
\end{align*}
Finally, we want to bound 
\[\sum_{j=1}^n p_j^2=\frac{\sum_{j=1}^n j^{-2/(\beta-1)}}{\left(\sum_{j=1}^n j^{-1/(\beta-1)}\right)^2}.\]
First, note that for $\beta=3$ this equation yields
\[\sum_{j=1}^n p_j^2=\frac{H_n}{\left(\sum_{j=1}^n j^{-1/(\beta-1)}\right)^2} \in \Theta(\ln n/n),\]
where $H_n$ denotes the n-th harmonic number.
For $\beta\neq3$ we can achieve
\begin{align*}
\sum_{j=1}^n p_j^2
	& \le \frac{1}{\left(\sum_{j=1}^n j^{-1/(\beta-1)}\right)^2}\left(1 + \int_{j=1}^{n} j^{-2/(\beta-1)} \mathrm{d}j\right)\\
	& = \frac{1}{\left(\sum_{j=1}^n j^{-1/(\beta-1)}\right)^2}\left(1 + \frac{\beta-1}{\beta-3}\cdot\left(n^{(\beta-3)/(\beta-1)}-1\right)\right)\\
\end{align*}
and
\begin{align*}
\sum_{j=1}^n p_j^2
	& \ge \frac{1}{\left(\sum_{j=1}^n j^{-1/(\beta-1)}\right)^2}\left(n^{-2/(\beta-1)}+\int_{j=1}^{n} j^{-2/(\beta-1)} \mathrm{d}j\right)\\
	& = \frac{1}{\left(\sum_{j=1}^n j^{-1/(\beta-1)}\right)^2}\left(n^{-2/(\beta-1)} + \frac{\beta-1}{\beta-3}\cdot\left(n^{(\beta-3)/(\beta-1)}-1\right)\right)\\
\end{align*}
If $\beta<3$, the expressions above yield
\[\sum_{j=1}^n p_j^2\in\Theta\left(\frac{1}{\left(\sum_{j=1}^n j^{-1/(\beta-1)}\right)^2}\right)\subseteq\Theta\left(n^{-2(\beta-2)/(\beta-1)}\right).\]
For $\beta>3$, they yield
\[\sum_{j=1}^n p_j^2\in\Theta\left(\frac{n^{(\beta-3)/(\beta-1)}}{n^{2(\beta-2)/(\beta-1)}}\right)\subseteq\Theta\left(n^{-1}\right).\]
This proves all statements of the lemma.
\end{proof}

\subsection{CDF of Connection Weights in the Geometric Model}

The CDF $F_X(x))$ of the connection weights $X(c, v)$ in the geometric
SAT model satisfies the following lemma.

\begin{lemma}
  \label{lem:cdf-connection-weights}
  $F_X(x) = 1 - \Pi_{d, \p} w_v x^{-T}$ for
  $x \ge \big(2^d w_v\big)^{1/T}$.
\end{lemma}
\begin{proof}
  Inserting the definition of the connection weight and rearranging
  slightly yields
  \begin{align*}
    F_X(x) &= \Pro{X({c}, {v}) \le x}\\
           &= \Pro{\left(\frac{w_v}{\dist{\pnt{c}}{\pnt{v}}^d}\right)^{1/T} \le x}\\
           &= \Pro{\dist{\pnt{c}}{\pnt{v}} \ge w_v^{1/d} x^{-T/d}}\\
           &= 1 - \Pro{\dist{\pnt{c}}{\pnt{v}} < w_v^{1/d} x^{-T/d}}.
  \end{align*}
  As $\pnt{c}$ and $\pnt{v}$ are two random points, we can use the CDF
  for the distances between random points in
  Equation~\eqref{eq:cdf-distance} to obtain
  \begin{equation*}
    F_X(x) = 1 - \Pi_{d, \p} w_v x^{-T} \quad \text{for } x \ge
    \left(2^d w_v\right)^{1/T},
  \end{equation*}
  which concludes the proof.
\end{proof}

\subsection{Volume of Balls in a Hypercube}

We are regularly concerned with the asymptotic behavior of a ball's
volume depending on its radius.  The following lemma helps us to deal
with the edge case, where the ball stretches beyond the boundary of
our ground space.

\begin{lemma}
  \label{lem:volume-intersection-ball-cube}
  Let $H$ be a $d$-dimensional unit-hypercube in $\mathbb R^d$
  equipped with a $\p$-norm.  There exists a constant $c > 0$ such
  that, for every $\pnt{p} \in H$ and $r > 0$, the intersection of $H$
  with the ball $B_{\pnt{p}}(r)$ of radius $r$ around $\pnt{p}$ has
  volume at least $\min\{1, cr^d\}$.
\end{lemma}
\begin{proof}
  In the following, we assume $H = [-0.5, 0.5]^d$ (rather than
  $[0, 1]^d$), as it makes the proof more convenient.  If $r$ is
  sufficiently small, then $B_{\pnt{p}}(r)$ is completely contained in
  $H$.  Thus, in this case, the claim follows from the fact that the
  volume of a ball with radius $r$ in $d$-dimensional space is
  proportional to $r^d$.  Thus, we have to prove that the parts of
  $B_{\pnt{p}}(r)$ outside of $H$ are asymptotically not relevant.

  Let $p_1, \dots, p_d$ be the coordinates of $\pnt{p}$ and assume
  without loss of generality that $\pnt{p}$ lies in the all-negative
  orthant, i.e., $p_i \le 0$ for $i \in [d]$.  We proof the claim by
  defining a box $B$ with the following three properties.  First, the
  box $B$ has volume proportional to $r^d$.  Second, $B$ is a subset
  of the ball $B_{\pnt{p}}(r)$.  Third, $B$ is a subset of the hypercube
  $H$ or $H$ is a subset of $B_{\pnt{p}}(r)$.  Note that the lemma's
  statement clearly holds if $H$ is a subset of $B_{\pnt{p}}(r)$ as the
  intersection has volume $1$ in this case.  If $H$ is not a subset of
  $B_{\pnt{p}}(r)$, the second and third property imply that $B$ is a
  subset of the intersection of $B_{\pnt{p}}(r)$ and $H$.  Thus, the
  volume of $B$ given by the first property is a lower bound for the
  volume of the intersection, which proves the claim.

  It remains to define $B$ and prove the three properties.  The box
  $B$ has $\pnt{p}$ as corner and extends from there in the direction
  of the all-positive orthant.  The side lengths are chosen
  proportional to the distance from the edge of $H$ in this direction.
  Formally, the corners of $B$ are
  $\{p_1, p_1 + r(0.5 - p_1)/\sqrt[\p]{d}\} \times \dots \times \{p_d,
  p_d + r(0.5 - p_d)/\sqrt[\p]{d}\}$.

  To prove the first property, note that the side length of $B$ in
  dimension $i$ is $r(0.5 - p_i)/\sqrt[\p]{d}$.  As $p_i \le 0$, this
  is at least $0.5r/\sqrt[\p]{d}$, which implies that the volume of
  $B$ is at least $(0.5r/\sqrt[\p]{d})^d$.  For the second property,
  note that the point of $B$ with maximum distance from $\pnt{p}$ is
  the opposite corner, i.e., the point with coordinates
  $(p_i + r(0.5 - p_i)/\sqrt[\p]{d})$.  The distance from $\pnt{p}$ is
  given by
  \begin{equation*}
    \sqrt[\p]{\sum_{i = 1}^d \left(\frac{r(0.5 - p_i)}{\sqrt[\p]{d}}\right)^\p}
    \le \sqrt[\p]{\sum_{i = 1}^d \frac{r^\p}{d}} = r.
  \end{equation*}
  Finally, for the third property, assume $r = \sqrt[\p]{d}$.  Then
  the coordinates $p_i + r(0.5 - p_i)/\sqrt[\p]{d}$ of the corners of
  $B$ simplify to $0.5$.  Thus, all corners of $B$ are still in the
  hypercube $H$ if $r \le \sqrt[\p]{d}$.  On the other hand, if
  $r \ge \sqrt[\p]{d}$, then $H$ is completely contained in
  $B_{\pnt{p}}(r)$, which concludes the proof of the last property.
  We note that it is easy to verify that all above arguments also hold
  for the limit $\p = \infty$.
\end{proof}

\subsection{Derivative of the Incomplete Gamma Function}
\label{sec:deriv-incompl-gamma}

We need the following somewhat technical bound that is easy to verify.

\begin{lemma}
  \label{lem:integral-gamma-function}
  Let $\Gamma$ be the gamma function.  For any
  $\alpha, \beta, \gamma, d \in \mathbb R$ with
  $\beta, \gamma, d > 0$,
  \begin{equation*}
    \int_0^\gamma x^{\alpha d - 1} \exp\left(-\beta x^d\right) \dif x
    \le \frac{\Gamma\left(\alpha\right)}{\beta^\alpha d}.
  \end{equation*}
\end{lemma}
\begin{proof}
  Let $\Gamma(\alpha, x)$ be the incomplete gamma function.  Its
  derivative is
  \begin{equation*}
    \frac{\partial \Gamma(\alpha, x)}{\partial x} = -x^{\alpha-1}
    \exp(-x).
  \end{equation*}
  Thus, it follows that
  \begin{equation*}
    \frac{\partial}{\partial x} \left(-\frac{\Gamma\left(\alpha, \beta
        x^d\right)}{\beta^\alpha d}\right) = \beta dx^{d-1} \frac{\left(\beta
        x^d\right)^{\alpha-1} \exp\left(-\beta
        x^d\right)}{\beta^\alpha d} = x^{\alpha d - 1} \exp\left(-\beta
      x^d\right).
  \end{equation*}
  Using this, the given integral evaluates to
  \begin{align*}
    \int_0^\gamma x^{\alpha d - 1} \exp\left(-\beta x^d\right) \dif x
    &= \left[-\frac{\Gamma\left(\alpha, \beta x^d\right)}{\beta^\alpha
      d}\right]_0^\gamma\\
    &= \frac{1}{\beta^\alpha d} \left(\Gamma\left(\alpha, 0\right)-
      \Gamma\left(\alpha, \beta \gamma^d\right)\right)\\
    &\le \frac{\Gamma\left(\alpha, 0\right)}{\beta^\alpha d}.
  \end{align*}
\end{proof}

\subsection{Balls Into Heterogeneous Bins}
\label{sec:balls-into-heter}

Consider throwing $m$ balls into $n$ uniform bins, i.e., for each ball
we draw one of the $n$ bins uniformly at random and place the ball
into the drawn bin.  The \emph{maximum load} $L$ is the random
variable that describes the maximum number of balls that are together
in the same bin.  From the analysis by \citet[Theorem~1]{rs-bb-98}, we
immediately get the following corollary.

\begin{corollary}[\cite{rs-bb-98}, Theorem~1]
  \label{cor:balls-into-unif-bins}
  Throw $m$ balls into $n$ uniform bins and let $L$ be the maximum
  load.  If $m \in \Omega(\frac{n}{\polylog n})$, then
  $L \in \Omega(\frac{\log n}{\log \log n})$ asymptotically almost
  surely.
\end{corollary}

Now assume we have non-uniform bins, i.e., the probability for each
ball to end up in the $i$th bin is $p_i$ with $\sum_i p_i = 1$.
Intuitively, Corollary~\ref{cor:balls-into-unif-bins} should still
hold in this setting, as increasing the probability of some bins only
makes it more likely that a bin gets many balls.  Making this argument
formal yields the following theorem.

\begin{theorem}
  Corollary~\ref{cor:balls-into-unif-bins} also holds for the
  non-uniform bins.
\end{theorem}
\begin{proof}
  Let $B = [n]$ be the set of all bins and let $B'$ be the subset of
  bins with probability at least $1/(2n)$.  These are the bins whose
  probability either increased, or decreased by a factor of at most
  $2$.  Without loss of generality, let $B' = [n']$.  Note that the
  probability for a ball to land in a bin of $B'$ is at least a
  constant, as every bin not in $B'$ has probability at most $1/(2n)$.
  Thus, by the Chernoff-Hoeffding bound in
  Corollary~\ref{cor:chernoff-hoeffding-asymptotic}, a constant
  fraction of the balls end up in a bin of $B'$ with high probability.
  We make a case distinction on how large $n'$ is.

  First, assume $n' \le m/\log n$.  Thus, with high probability, we
  end up with $\Theta(m)$ balls in at most $m/\log n$ bins, which
  means that at least one bin contains $\Omega(\log n)$ balls.  Thus,
  clearly $L \in \Omega(\log n/\log \log n)$.

  Second, assume $n' > m / \log n$.  Recall that each bin in $B'$ has
  probability at least $1/(2n)$.  We consider the alternative
  experiment where, for every ball, each bin in $B'$ has probability
  exactly $1/(2n)$ to get the ball.  Balls not landing in $B'$ are
  discarded.  Let $L'$ denote the maximum number of balls that share a
  bin in $B'$.  Clearly, we can couple the two experiments such that
  $L \ge L'$ holds in every outcome.  It remains to show that
  $L' \in \Omega(\log n / \log \log n)$.  For this, let $m'$ be the
  number of balls ending up in $B'$.  Note that $m'$ is a random
  variable.  However, if we condition on $m'$, then we are back to the
  normal homogeneous balls into bins, except that we throw $m'$ balls
  into $n'$ bins.  If we show that $m' \in \Omega(n'/\polylog n')$,
  then Corollary~\ref{cor:balls-into-unif-bins} tells us that
  $L' \in \Omega(\log n' / \log\log n')$.  First note that this is
  sufficient for our purpose: as $n' > m/\log n$ and
  $m \in \Omega(n / \polylog n)$, we get
  \begin{align*}
    \frac{\log n'}{\log\log n'}
    &> \frac{\log m - \log\log n}{\log(\log m - \log\log n)}\\
    &\in \Omega\left(\frac{\log n - \log\polylog n - \log\log
      n}{\log(\log n - \log\polylog n - \log\log n)}\right)\\
    &\subseteq \Omega\left(\frac{\log n}{\log\log n}\right).
  \end{align*}

  It remains to show that $m' \in \Omega(n'/\polylog n')$ so that we
  can actually apply Corollary~\ref{cor:balls-into-unif-bins}.  To do
  so, recall that $B'$ has $n'$ bins, each with probability $1/(2n)$.
  Thus, the probability that a single ball lands in $B'$ is $n'/(2n)$,
  which shows that $m'$ is $m n' / (2n)$ in expectation.  As $n'$ is
  almost $m$ (up to logarithmic factors) and $m$ is almost $n$, this
  expectation is almost linear in $n$.  Thus, by the
  Chernoff-Hoeffding bound in
  Corollary~\ref{cor:chernoff-hoeffding-asymptotic}, we can assume
  that
  \begin{equation*}
    m' \in \Theta\left(\frac{mn'}{n}\right)
  \end{equation*}
  holds with high probability.  Using that $n' > m / \log n$ and
  $m \in \Omega(n / \polylog n)$, we obtain
  \begin{equation*}
    \frac{mn'}{n} > \frac{m^2}{n\log n}
    \in \Omega\left(\frac{n}{\polylog n}\right).
  \end{equation*}
  As $n' \le n$, it follows that $m \in \Omega(n' / \polylog n')$,
  which concludes the proof.
\end{proof}

\subsection{Concentration Bounds}
\label{sec:concentration-bounds}

For a random experiment, we say that an event happens \emph{with high
  probability (\whp)} if the probability is at least $1 - O(1/n)$.
The type of event we are usually interested in is that a random
variable assumes a value close to its expectation, i.e., that the
random variable is concentrated.  In the following, we state two well
known techniques to prove concentration, namely a Chernoff-Hoeffding
bound and the method of bounded differences.  In both cases we derive
asymptotic variants that suite our purpose better than the original
exact bounds.

\subsubsection{Chernoff-Hoeffding}
\label{sec:chernoff-hoeffding}

\begin{theorem}[Theorem 1.1 in~\cite{dp-cmara-12}]
  \label{thm:chernoff-hoeffding}
  Let $X_1, \dots, X_n$ be independent random variables with values in
  $\{0, 1\}$ and let $X = \sum_{i \in [n]} X_i$ be their sum. Then,
  for all $0 < \eps < 1$,
  \begin{align*}
    \Pro{X > (1 + \eps)\cdot\EX{X}}
    &\le \exp\left(-\frac{\eps^2}{3} \EX{X}\right), \text{and}\\
    \Pro{X < (1 - \eps)\cdot\EX{X}}
    &\le \exp\left(-\frac{\eps^2}{2} \EX{X}\right).
  \end{align*}
\end{theorem}

We use this bound multiple times in a similar way, which is captured
by the following direct corollary.

\begin{corollary}
  \label{cor:chernoff-hoeffding-asymptotic}
  Let $X_1, \dots, X_n$, and $X$ be as in
  Theorem~\ref{thm:chernoff-hoeffding}.  Let $f(n) \in \omega(\log n)$
  be an upper or lower bound for $\EX{X}$. With overwhelming probability, $X \in O(f(n))$ and $X \in \Omega(f(n))$,
  respectively.
\end{corollary}
\begin{proof}
  Assume $f(n)$ is a lower bound, i.e., $f(n) \le \EX{X}$.  We show
  $X \in \Omega(f(n))$ with the desired probability.  By the second
  inequality of Theorem~\ref{thm:chernoff-hoeffding}, we have
  \begin{align*}
    \Pro{X < (1 - \eps)\cdot f(n)}
    &\le \Pro{X < (1 - \eps)\cdot\EX{X}}\\
    &\le \exp\left(-\frac{\eps^2}{2} \EX{X}\right)\\
    &\le \exp\left(-\frac{\eps^2}{2} f(n)\right) = n^{-\omega(1)},
  \end{align*}
  where the last equality is due to the fact that
  $f(n) \in \omega(\log n)$.  Thus, for any constant $c$, this
  probability is below $n^{-c}$ for sufficiently large $n$.  Hence,
  for any constant $\eps \in (0, 1)$, $X \ge (1 - \eps)\cdot f(n)$ with
  probability $1 - n^{-c}$.

  Assume $f(n)$ is an upper bound, i.e., $\EX{X} \le f(n)$.  Let $X'$
  be a random variable with $f(n) = \EX{X'}$ such that $X'$ dominates
  $X$ in the sense that $X \le X'$ for every outcome.  We show that
  $X' \in O(f(n))$ with probability $1 - n^{-c}$, which implies
  $X \in O(f(n))$ with at least the same probability.  The first
  inequality of Theorem~\ref{thm:chernoff-hoeffding} yields
  \begin{align*}
    \Pro{X' > (1 + \eps)\cdot f(n)}
    &= \Pro{X' > (1 + \eps)\cdot \EX{X'}}\\
    &\le \exp\left(-\frac{\eps^2}{3} \EX{X'}\right)\\
    &= \exp\left(-\frac{\eps^2}{3} f(n)\right) = n^{-\omega(1)}.
  \end{align*}
  As before, the last inequality comes from the fact that
  $f(n) \in \omega(\log n)$.  All remaining arguments are as in the
  case where $f(n)$ was a lower bound.
\end{proof}

\subsubsection{Method of Typical Bounded Differences}
\label{sec:meth-typic-bound}

\begin{theorem}[Theorem 2 in \cite{w-mtbd-16}]
  \label{thm:typical-bounded-differences}
  Let $X = (X_1, \dots, X_N)$ be a family of independent random
  variables with $X_k$ taking values in $\Lambda_k$ and let
  $\Lambda = \prod_{j \in [N]} \Lambda_j$.  Let $\Gamma \subseteq \Lambda$
  be an event and assume that the function
  $f \colon \Lambda \to \mathbb R$ satisfies the following
  \emph{typical Lipschitz condition}.
  \begin{itemize}[(TL)]
  \item There are numbers $(c_k)_{k \in [N]}$ and $(d_k)_{k \in [N]}$
    with $c_k \le d_k$ such that whenever $x, \tilde x \in \Lambda$
    differ only in the $k$th coordinate, we have
    \begin{equation*}
      |f(x) - f(\tilde x)| \le
      \begin{cases}
        c_k & \text{ if } x \in \Gamma,\\
        d_k & \text{ otherwise.}
      \end{cases}
    \end{equation*}
  \end{itemize}
  For any numbers $(\gamma_k)_{k \in [N]}$ with $\gamma_k \in (0, 1]$,
  there is an event
  $\mathcal B = \mathcal B(\Gamma, (\gamma_k)_{k \in [N]})$ satisfying
  \begin{equation*}
    \Pro{\mathcal B} \le \sum_{k \in [N]} \gamma_k^{-1} \cdot \Pro{X
      \notin \Gamma} \; \text{ and }\; \neg\mathcal B \subseteq \Gamma,
  \end{equation*}
  such that for $\mu = \EX{f}$, $e_k = \gamma_k(d_k - c_k)$ and any
  $t \ge 0$, we have
  \begin{equation*}
    \Pro{f(X) \ge \mu + t \text{ and } \neg \mathcal B} \le \exp\left(
      - \frac{t^2}{2 \sum_{k \in [N]} (c_k + e_k)^2} \right).
  \end{equation*}
\end{theorem}

We derive the following corollary from this, which is more convenient
for our purpose and uses a notation more compatible with the rest of
the paper.

\begin{corollary}
  \label{cor:typical-bounded-differences}
  Let $X = (X_1, \dots, X_N) \in \Lambda$ be a family of independent
  random variables and let $\Gamma \subseteq \Lambda$ be an event with
  $\Pro{\Gamma} \ge 1 - N^{-c}$.  Moreover, let
  $f \colon \Lambda \to \mathbb R$ with $|f(X)| \le N^{c - 2}$ and let
  $(\Delta_i)_{i \in [N]} \in \Omega(1)$ be numbers such that for any two
  $x \in \Gamma$ and $\tilde x \in \Lambda$ that differ only in the
  $i$th coordinate, we have $|f(x) - f(\tilde x)| \le \Delta_i$.
  If $\sum_{i \in [N]} \Delta_i^2 \in O(\EX{f}^2 / \log^2 N)$ then
  $f(X) \in \Theta(\EX{f})$ holds with high probability.
\end{corollary}
\begin{proof}
  We want to apply Theorem~\ref{thm:typical-bounded-differences}.
  First note that $|f(X)| \le N^{c - 1}$ implies that $f$ satisfies
  the typical Lipschitz condition when setting $c_i = \Delta_i$ and
  $d_i = 2N^{c - 2}$ for every $i$.  We set $\gamma_i$ in
  Theorem~\ref{thm:typical-bounded-differences} to
  $\gamma_i = 1 / d_i$ yielding $e_i \le 1$.  Thus, we get the event
  $\mathcal B$ with
  \begin{align*}
    \Pro{\mathcal B}
    &\le \sum_{i \in [N]} \gamma_i^{-1} \cdot \Pro{\neg \Gamma}\\
    &= \sum_{i \in [N]} d_i \cdot \Pro{\neg \Gamma}\\
    &\le \sum_{i \in [N]} 2N^{c - 2} \cdot N^{-c}\\
    &= 2N^{-1} \in O(N^{-1}),
  \end{align*}
  such that
  \begin{equation*}
    \Pro{f(X) \ge \EX{f} + t \text{ and } \neg\mathcal B} \le \exp\left(
      - \frac{t^2}{2 \sum_{i \in [N]} (\Delta_i + e_i)^2} \right).
  \end{equation*}
  As $\Delta_i \in \Omega(1)$ and $e_i \le 1$, we get that the sum in the
  denominator is up to constants equal to
  $\sum_{i \in [N]} \Delta_i^2 \in O(\EX{f}^2 / \log^2 N)$, i.e., for
  sufficiently large $N$, there exists a positive constant $a$ such
  that
  \begin{equation*}
    \Pro{f(X) \ge \EX{f} + t \text{ and } \neg\mathcal B} \le \exp\left(
      - \frac{a \cdot t^2 \cdot \log^2 N}{\EX{f}^2 } \right).
  \end{equation*}
  Choosing $t = b\cdot \EX{f}$ for any positive constant $b$ yields an
  upper bound of $N^{-a b^2 \log N} \in O(N^{-1})$ for the
  probability.
  Thus, we obtain
  \begin{equation*}
    \Pro{f(x) \ge (1 + b)\EX{f}}
    \le \Pro{f(x) \ge (1 + b) \EX{f} \text{ and } \neg\mathcal
      B} + \Pro{\mathcal B} \in O(N^{-1}).
  \end{equation*}

  As already noted by Warnke~\cite{w-mtbd-16}, one can obtain the same
  bound as stated in Theorem~\ref{thm:typical-bounded-differences} for
  the opposite direction ($\Pro{f(x) \le \EX{f} - t}$) by using $-f$.
  The above argument works exactly the same for this case, yielding
  $\Pro{f(x) \le (1 - b) \EX{f}} \in O(N^{-1})$.  Note that for this
  direction it is crucial that we can choose $b$ to be an arbitrary
  positive constant.  This yields the claim that
  $f(x) \in \Theta(\EX{f})$ with high probability.
\end{proof}

} 

\end{document}